\documentclass{article}
\usepackage{latexsym}
\usepackage{graphicx}
\usepackage{color}
\usepackage{mathrsfs}
\usepackage{amsmath,amssymb}
\usepackage{amsthm}
\usepackage{mathtools}
\usepackage{breqn}
\usepackage{bm}
\usepackage{bbm}
\usepackage{ascmac}
\usepackage{braket}
\usepackage{geometry}
\usepackage{here}
\theoremstyle{plain}
\usepackage[unicode,bookmarks=true,colorlinks=true]{hyperref}
\numberwithin{equation}{section}
\newtheorem{thm}{Theorem}[section]
\newtheorem{lem}[thm]{Lemma}

\newtheorem{dfn}[thm]{Definition}
\newtheorem{asm}[thm]{Assumption}
\newtheorem{rem}[thm]{Remark}

\mathtoolsset{showonlyrefs}
\newcommand{\esssup}{\mathop{\rm ess~sup}\limits}

\providecommand{\keywords}
{
  \small
  \textbf{{Keywords~}} asset pricing, backward stochastic differential equation, exponential utility, mean-field game, default risk
  \normalsize
}
\providecommand{\MSC}
{
  \small
  \textbf{{Mathematics Subject Classification (2020)~}} 49N80, 60H10, 91B51, 91G30, 91G40
  \normalsize
}
\providecommand{\JEL}
{
  \small
  \textbf{{JEL Classification~}} C62, C73, D52, G11, G12
  \normalsize
}
\title{\textbf{\Large{Mean-field equilibrium price formation under single-default risk}}}
\author{Masashi Sekine\thanks{Tokyo, Japan. E-mail: msekine.econ@gmail.com. ORCID: 0009-0002-5042-9895.}}
\date{July 20, 2026}
\begin{document}

\maketitle
\begin{abstract}
We study equilibrium price formation in an incomplete financial market with a large population of agents, where stock prices are subject to a single-default event.
Agents are assumed to be heterogeneous in their risk aversion and terminal liabilities, and maximize exponential utility of terminal net wealth.
We first characterize each agent's optimal strategy by a quadratic-growth backward stochastic differential equation (BSDE) driven by Brownian motions and a compensated default martingale.
We then formulate the market-clearing condition in terms of aggregate optimal demand and derive a mean-field quadratic-growth BSDE for the equilibrium risk premium.
The resulting characterization quantifies how default intensity, jump size, and agent heterogeneity jointly shape the default-risk component of equilibrium security risk premia.
Under a Markovian factor model, we establish short-time solvability of the mean-field BSDE through a fixed-point argument based on estimates for a coupled semilinear PDE system.
Finally, we show that the risk premium characterized by the mean-field BSDE asymptotically clears the market as the population size tends to infinity.
\end{abstract}
\keywords~\\
\MSC~\\
\JEL

\section{Introduction}

Equilibrium asset pricing studies how risk premia are determined when agents' optimal trading demands clear the market.
We examine this question in an incomplete market where risky assets are exposed to Brownian shocks and a single default occurring at a totally inaccessible time.
Agents differ in their risk aversion and terminal liabilities and maximize exponential utility of terminal net wealth.
Default affects agents through both a jump in asset prices and a possible change in their terminal liabilities.
Our aim is to analyze the equilibrium risk premium in the large-population setting using mean-field game theory.

Mean-field game (MFG) theory was independently introduced by Lasry \& Lions \cite{lasryMeanFieldGames2007} and by Huang, Malham\'e \& Caines \cite{huangLargePopulationStochastic2006}.
The idea is to replace multi-agent games, which are usually mathematically intractable, with a representative-agent control problem and a self-consistency condition describing the aggregate behavior of the population.
In the analytic formulation, this leads to a coupled system of Hamilton-Jacobi-Bellman and Fokker-Planck equations.
The probabilistic formulation developed by Carmona \& Delarue \cite{carmonaProbabilisticAnalysisMeanField2013,carmonaForwardBackwardStochastic2015} instead uses forward-backward stochastic differential equations (FBSDEs) of McKean-Vlasov type.
See the monographs \cite{CarmonaDelarue2018I,CarmonaDelarue2018II} for comprehensive treatments.

In recent years, MFG theory has been increasingly applied to equilibrium price formation, where the price process is constructed endogenously so that aggregate supply and demand balance in a large population of agents.
Gomes \& Sa\'ude \cite{GomesSaude2021} present a deterministic mean-field model of electricity price formation, which is extended to random supply by Gomes, Gutierrez \& Ribeiro \cite{GomesGutierrezRibeiro2023}.
Shrivats, Firoozi \& Jaimungal \cite{ShrivatsFirooziJaimungal2022} solve a price formation problem for solar renewable energy certificates via FBSDEs of McKean-Vlasov type, and F\'eron, Tankov \& Tinsi \cite{FeronTankovTinsi2022} develop an equilibrium model of intraday electricity markets.
For securities markets, Fujii \& Takahashi \cite{FujiiTakahashi2022SICON} formulate agents' decisions as optimal inventory management problems and characterize the market-clearing equilibrium price using a McKean-Vlasov FBSDE with common noise.
They prove strong convergence of finite-agent equilibria to the mean-field limit \cite{FujiiTakahashi2022SIFIN} and extend the model to include a major player \cite{FujiiTakahashi2022ESAIM}.
Fujii \cite{Fujii2023} allows the co-presence of cooperative and non-cooperative populations, while the recent binomial-tree models of \cite{Fujii2025Tree,Fujii2025TreeRelative} accommodate multiple populations, non-rational agents, and relative performance concerns.

A particularly relevant line of work is the mean-field equilibrium asset pricing framework developed by Fujii \& Sekine \cite{FujiiSekine2025}.
They study an incomplete market in which heterogeneous agents with exponential utilities trade self-financing portfolios and face terminal liabilities that are unspanned by the security prices.
Following the martingale optimality method of Hu, Imkeller \& M\"uller \cite{HIM2005}, each agent's optimal strategy is characterized by a quadratic-growth BSDE.
The market-clearing condition is then imposed on the aggregate optimal demand, and the equilibrium risk premium is characterized by a mean-field BSDE whose driver has quadratic growth in both the stochastic integrands and their conditional expectations.
This framework has been extended to models with consumption and habit formation \cite{FujiiSekineHabit2026} and to partially observable markets \cite{Sekine2025EQG}.
Across these models, the equilibrium risk premium reflects aggregate hedging demand generated by common Brownian risk.
A simplified analytic example and its economic interpretation are given in \cite[Example 4.3.15]{SekineThesis2025}.
This line of work, however, remains purely diffusive and does not address how a market prices a default event that changes asset returns and may also alter agents' liabilities.
Single-agent utility maximization with jumps has been studied in the BSDE literature \cite{Morlais2009,Morlais2010,LimQuenez2011}, but the endogenous default-risk premium generated by market clearing has not been studied within the mean-field equilibrium asset-pricing framework developed by Fujii \& Sekine or its extensions.

This paper develops a mean-field equilibrium pricing model that incorporates single-default risk into the above setting.
In the defaultable setting, market clearing generates an additional compensation for the default event.
A basic implication is that a security that would lose value at default requires a positive default-risk component in its pre-default risk premium.
This model further describes quantitatively how default intensity, jump size, and the cross-sectional distribution of agents' risk aversion and liabilities are aggregated by market clearing into this default-risk component.

The default jump also requires further mathematical treatment.
In the purely diffusive framework of \cite{FujiiSekine2025}, the agent's optimal strategy admits an explicit representation in terms of the BSDE integrands.
The explicit representation allows the market-clearing condition to be solved directly for the candidate equilibrium risk premium, leading to a single self-contained mean-field BSDE.
In the present defaultable model, each agent's optimal strategy is still characterized by the solution of a quadratic-growth BSDE, but it can no longer be represented explicitly using only elementary functions of the BSDE integrands.
Instead, the strategy is recovered from the unique root of a nonlinear scalar equation for its jump exposure.
The equilibrium construction therefore couples the mean-field BSDE with a scalar fixed-point problem, rather than reducing to a single self-contained BSDE.
Accordingly, it is necessary to establish the well-posedness of this scalar fixed-point problem, which yields a well-defined driver, before addressing the solvability of the mean-field BSDE.

In this paper, we first solve the individual optimization problem in the single-default market.
The martingale optimality argument \cite{HIM2005} leads to a quadratic-growth BSDE with a compensated default martingale, whose solution characterizes the optimal strategy.
We then derive the mean-field BSDE for the equilibrium risk premium from the market-clearing condition.
Under a Markovian factor model, we prove short-time existence of the mean-field equilibrium through a fixed-point argument based on estimates for a coupled semilinear PDE system.
Finally, we justify the mean-field formulation from the finite-agent market by showing that the risk premium characterized by the mean-field BSDE asymptotically clears the market in the large population limit.

The contributions are twofold.
On the economic side, this paper extends the mean-field equilibrium pricing framework of Fujii \& Sekine \cite{FujiiSekine2025} to a defaultable market, providing a quantitative characterization of the equilibrium risk premium in such a market and showing that the premium decomposes into the Brownian hedging component inherited from the diffusive model and a default-risk component.
On the mathematical side, we derive a mean-field quadratic-growth BSDE with a compensated default martingale that characterizes the equilibrium risk premium.
Under the Markovian factor model, we prove short-time solvability of the mean-field BSDE through a fixed-point argument based on estimates for a coupled semilinear PDE system.
We also show that the resulting risk premium asymptotically clears the finite-agent market in the large population limit.

The paper is organized as follows.
Section \ref{sec:notation} introduces notation used throughout the paper.
Section \ref{sec:individual-bsde} studies the exponential-utility optimization problem of an individual agent in a market with single-default risk and derives the associated quadratic BSDE with jumps.
Section \ref{sec:mfg} defines the market-clearing condition, gives a heuristic derivation of the equilibrium risk premium, and derives the mean-field BSDE that characterizes the equilibrium.
It then introduces the Markovian factor model, derives the coupled PDE system, and proves existence of a mean-field equilibrium for sufficiently short horizons.
Section \ref{sec:clearing} shows that the risk premium constructed in Section \ref{sec:mfg} asymptotically clears the finite-agent market as the population size tends to infinity.
The paper concludes with Section \ref{sec:conclusion}, which suggests some directions for future research.

\section{Notation}\label{sec:notation}

Throughout this paper, $T>0$ is a fixed finite time horizon.
For a given filtered probability space with usual conditions $(\Omega,\mathcal{F},\mathbb{P},\mathbb{F}~(:=(\mathcal{F}_t)_{t\in[0,T]}))$ and a vector space $E$ over $\mathbb{R}$, we use the following notation to describe frequently used sets and function spaces.\\

\noindent
(1) $\mathcal{T}(\mathbb{F})$ is the set of all $\mathbb{F}$-stopping times with values in $[0,T]$. \\

\noindent
(2) $\mathcal{P}(\mathbb{F})$ denotes the predictable $\sigma$-field on $[0,T]\times\Omega$, i.e., the $\sigma$-field generated by all real-valued left-continuous $\mathbb{F}$-adapted processes.  \\

\noindent
(3) $\mathbb{L}^0(\mathcal{F},E)$ is the set of $E$-valued $\mathcal{F}$-measurable random variables.\\

\noindent
(4) $\mathbb{L}^2(\mathbb{P},\mathcal{F},E)$ is the set of $E$-valued $\mathcal{F}$-measurable random variables $\xi$ satisfying $\|\xi\|_{\mathbb{L}^2}:=\mathbb{E}^{\mathbb{P}}[|\xi|^2]^{\frac{1}{2}}<\infty$.\\

\noindent
(5) $\mathbb{L}^{\infty}(\mathbb{P},\mathcal{F},E)$ is the set of $E$-valued $\mathcal{F}$-measurable random variables $\xi$ that are essentially bounded, $\|\xi\|_{\mathbb{L}^{\infty}}:=\esssup_{\omega\in\Omega}|\xi(\omega)|<\infty$.\\

\noindent
(6) $\mathbb{L}^0(\mathbb{F},E)$ is the set of $E$-valued $\mathbb{F}$-progressively measurable stochastic processes.
$\mathbb{L}^0_{\mathrm{pred}}(\mathbb{F},E)$ is the subclass of $\mathbb{F}$-predictable processes.\\

\noindent
(7) $\mathbb{H}^2(\mathbb{P},\mathbb{F},E)$ is the set of $E$-valued $\mathbb{F}$-predictable processes $X$ satisfying
\begin{equation}
  \|X\|_{\mathbb{H}^2} := \mathbb{E}^{\mathbb{P}}\Bigl[\int_0^T|X_t|^2dt\Bigr]^{\frac{1}{2}}<\infty.
\end{equation}

\noindent
(8) $\mathbb{H}^2_{\mathrm{BMO}}(\mathbb{P},\mathbb{F},E)$ is the subspace of $\mathbb{H}^2(\mathbb{P},\mathbb{F},E)$ satisfying
\begin{equation}
  \|X\|_{\mathbb{H}^2_{\mathrm{BMO}}}^2:= \sup_{\tau\in\mathcal{T}(\mathbb{F})}\Bigl\|\mathbb{E}^{\mathbb{P}}\Bigl[\int_{\tau}^T|X_t|^2dt \Big| \mathcal{F}_{\tau}\Bigr]\Bigr\|_{\mathbb{L}^{\infty}}<\infty.
\end{equation}

\noindent
(9) $\mathbb{L}^2(\mathbb{P},\mathbb{F},\lambda,E)$ is the set of $E$-valued $\mathbb{F}$-predictable processes $X$ satisfying
\begin{equation}
  \|X\|_{\mathbb{L}^2(\lambda)}^2 := \mathbb{E}^{\mathbb{P}}\Bigl[\int_0^T\lambda_t|X_t|^2dt\Bigr] <\infty,
\end{equation}
where $\lambda=(\lambda_t)_{t\in[0,T]}$ is a non-negative, bounded, $\mathbb{F}$-predictable process. 
Processes in this space are identified up to $\lambda_tdt\otimes\mathbb{P}$-null sets.
When $\lambda$ is the compensator density of a martingale $M$, we write $\mathbb{L}^2(M;\mathbb{P},\mathbb{F},E)$ and $\|\cdot\|_{\mathbb{L}^2(M)}$ for this space and norm.\\

\noindent
(10) $\mathbb{L}^2_{\mathrm{BMO}}(\mathbb{P},\mathbb{F},\lambda,E)$ is the subspace of $\mathbb{L}^2(\mathbb{P},\mathbb{F},\lambda,E)$ satisfying
\begin{equation}
  \|X\|_{\mathbb{L}^2_{\mathrm{BMO}}(\lambda)}^2 := \sup_{\tau\in\mathcal{T}(\mathbb{F})} \Bigl\|\mathbb{E}^{\mathbb{P}}\Bigl[\int_{\tau}^T\lambda_t|X_t|^2dt \Big| \mathcal{F}_{\tau}\Bigr]\Bigr\|_{\mathbb{L}^{\infty}} <\infty.
\end{equation}
When $\lambda$ is the compensator density of $M$, we write $\mathbb{L}^2_{\mathrm{BMO}}(M;\mathbb{P},\mathbb{F},E)$ and $\|\cdot\|_{\mathbb{L}^2_{\mathrm{BMO}}(M)}$ for this space and norm.\\

\noindent
(11) $\mathbb{L}^{\infty}(\mathbb{P},\mathbb{F},E)$ is the set of $E$-valued $\mathbb{F}$-predictable processes $X$ that are essentially bounded,
\begin{equation}
  \|X\|_{\mathbb{L}^{\infty}} := \esssup_{(t,\omega)\in[0,T]\times\Omega}|X_t(\omega)| < \infty.
\end{equation}

\noindent
(12) $\mathbb{S}^2(\mathbb{P},\mathbb{F},E)$ (resp.~$\mathbb{S}^{\infty}(\mathbb{P},\mathbb{F},E)$) is the set of $E$-valued $\mathbb{F}$-adapted c\`adl\`ag processes $X$ satisfying
\begin{equation}
  \|X\|_{\mathbb{S}^2} := \mathbb{E}^{\mathbb{P}}\Bigl[\sup_{t\in[0,T]}|X_t|^2\Bigr]^{\frac{1}{2}} <\infty ~~~ \Bigl(\text{resp.~} \|X\|_{\mathbb{S}^{\infty}} :=\esssup_{\omega\in\Omega}\sup_{t\in[0,T]}|X_t(\omega)| <\infty\Bigr).
\end{equation}

\noindent
(13) For an open set $\mathcal{O}\subseteq\mathbb{R}^d$ and an integer $k\geq0$, $C^k(\mathcal{O},E)$ is the space of $k$ times continuously differentiable functions $f:\mathcal{O}\to E$.
$C^k_b(\mathcal{O},E)$ is the set of $f\in C^k(\mathcal{O},E)$ such that all derivatives up to order $k$ are bounded. We write
\begin{equation}
  \|f\|_{C^k}:=\sum_{|\beta|\leq k}\sup_{x\in\mathcal{O}}|\partial^{\beta}f(x)|.
\end{equation}

\noindent
(14) Let $\mathcal{Q}\subset[0,T]\times\mathbb{R}^d$ be an open set.
\begin{itemize}
  \item For $p\in[1,\infty]$, $L^p(\mathcal{Q},E)$ is the set of measurable functions $f:\mathcal{Q}\to E$ such that
  \begin{equation}
    \begin{split}
       \|f\|_{L^p(\mathcal{Q})}:=\Bigl(\int_{\mathcal{Q}}|f(t,x)|^pdtdx\Bigr)^{1/p}<\infty,~~~p\in[1,\infty),~~~\|f\|_{L^{\infty}(\mathcal{Q})}:=\esssup_{(t,x)\in\mathcal{Q}}|f(t,x)|<\infty.
    \end{split}
  \end{equation}
  $L^p_{\mathrm{loc}}(\mathcal{Q},E)$ is the set of functions $f:\mathcal{Q}\to E$ such that $f\in L^p(\mathcal{Q}',E)$ for every relatively compact open set $\mathcal{Q}'\subset\mathcal{Q}$.
  \item $C^{1,2}(\mathcal{Q},E)$ is the set of functions $f:\mathcal{Q}\to E$ such that $f,\partial_tf,\nabla_xf$, and $\nabla_x^2f$ are continuous. $C^{1,2}_b(\mathcal{Q},E)$ is the set of $f\in C^{1,2}(\mathcal{Q},E)$ such that $f,\partial_tf,\nabla_xf$, and $\nabla_x^2f$ are bounded.
  \item For $p\in(1,\infty)$, $W^{1,2}_p(\mathcal{Q},E)$ is the Sobolev space consisting of functions $f\in L^p(\mathcal{Q},E)$ whose weak derivatives $\partial_tf$, $\nabla_xf$, and $\nabla_x^2f$ exist and belong to the corresponding $L^p$ spaces.
  It is equipped with the norm
  \begin{equation}
    \|f\|_{W^{1,2}_p(\mathcal{Q})} :=\|f\|_{L^p}+\|\partial_tf\|_{L^p}+\|\nabla_xf\|_{L^p}+\|\nabla_x^2f\|_{L^p}<\infty.
  \end{equation}
  $W^{1,2}_{p,\mathrm{loc}}(\mathcal{Q},E)$ is the set of functions $f:\mathcal{Q}\to E$ such that $f\in W^{1,2}_p(\mathcal{Q}',E)$ for every relatively compact open set $\mathcal{Q}'\subset\mathcal{Q}$.
\end{itemize}

\noindent
(15) $\mathbb{R}^d_+:=\{x\in\mathbb{R}^d; x\geq 0\}$ and $\mathbb{R}^d_{++}:=\{x\in\mathbb{R}^d; x> 0\}$ for $d\in\mathbb{N}$, where $x\geq 0$ and $x> 0$ mean that all elements of $x$ are nonnegative and strictly positive, respectively.\\

\noindent
For (1) to (14), we may omit the arguments $(\mathbb{P},\mathcal{F},\mathbb{F},\lambda,E)$ and the domains $\mathcal{O}, \mathcal{Q}$ when they are clear from context. Throughout the paper, $C$ denotes a generic non-negative constant that may change from line to line. The argument $\omega\in\Omega$ is omitted whenever there is no risk of misinterpretation.

\section{Individual optimization and the BSDE with jumps}\label{sec:individual-bsde}

\subsection{Preliminaries}\label{sec:preliminary}
(1) $(\Omega^0,\mathcal{F}^0,\mathbb{P}^0)$ is a complete probability space with complete and right-continuous filtration $\mathbb{F}^0:=(\mathcal{F}^0_t)_{t\in[0,T]}$, which is generated by a $d_0$-dimensional standard Brownian motion $W^0=(W^0_t)_{t\in[0,T]}$ and a jump process $N:=(N_t)_{t\in[0,T]}$. 
Here, the process $N$ is defined by $N_t:=\mathbf{1}_{\{\tau\leq t\}}$ for $t\in[0,T]$, where $\tau$ is an $\mathbb F^0$-stopping time representing a single default time and satisfying $\mathbb{P}^0(\tau>0)=1$ and $\mathbb{P}^0(\tau>t)>0$ for every $t\in[0,T]$.
We assume that $W^0$ is an $\mathbb{F}^0$-Brownian motion.
We set $\mathcal{F}^0:=\mathcal{F}^0_T$.\\

\noindent
(2) $(\Omega^1,\mathcal{F}^1,\mathbb{P}^1)$ is a complete probability space with complete and right-continuous filtration $\mathbb{F}^1:=(\mathcal{F}^1_t)_{t\in[0,T]}$, which is generated by a $d$-dimensional standard Brownian motion $W^1=(W^1_t)_{t\in[0,T]}$, a bounded $\mathbb{R}$-valued random variable $\xi^1$, and a strictly positive bounded random variable $\gamma^1$. $W^1$ is independent of $(\xi^1,\gamma^1)$.
$\mathcal{F}^1_0$ is the completion of $\sigma(\xi^1,\gamma^1)$.
We set $\mathcal{F}^1:=\mathcal{F}^1_T$.\\

\noindent
(3) $(\Omega^{0,1},\mathcal{F}^{0,1},\mathbb{P}^{0,1})$ is the completion of the product $(\Omega^0\times\Omega^1,\mathcal{F}^0\otimes\mathcal{F}^1,\mathbb{P}^0\otimes\mathbb{P}^1)$. The product filtration $\mathbb{F}^{0,1}:=(\mathcal{F}^{0,1}_t)_{t\in[0,T]}$ is the right-continuous augmentation of $(\mathcal{F}^0_t\otimes\mathcal{F}^1_t)_{t\in[0,T]}$.\\

We do not distinguish a random variable defined on a marginal probability space from its trivial extension to a product space for notational simplicity.
For example, we will use the same symbol $X$ for a random variable $X(\omega^0)$ defined on the space $(\Omega^0,\mathcal{F}^0,\mathbb{P}^0)$  and for its trivial extension $X(\omega^0,\omega^1):=X(\omega^0)$
defined on the product space $(\Omega^{0,1}, \mathcal{F}^{0,1},\mathbb{P}^{0,1})$.
In this section, the expectation with respect to $\mathbb{P}^{0,1}$ is denoted by $\mathbb{E}[ \cdot ]$. We write $\mathcal{T}^0:=\mathcal{T}(\mathbb{F}^0)$, $\mathcal{T}^{0,1}:=\mathcal{T}(\mathbb{F}^{0,1})$ and $\mathcal{P}^0:=\mathcal{P}(\mathbb{F}^0)$ for short.

The compensator of $N$ is assumed to be absolutely continuous with respect to Lebesgue measure and has bounded non-negative $\mathbb{F}^0$-predictable density $\lambda$.
The process
\begin{equation}
  M_t := N_t - \int_0^t\lambda_sds, ~~~ t\in[0,T],
\end{equation}
is an $(\mathbb{F}^0,\mathbb{P}^0)$-martingale. By the product structure, it remains an $(\mathbb{F}^{0,1},\mathbb{P}^{0,1})$-martingale.
The absolute continuity of the compensator implies that $\tau$ is totally inaccessible. See \cite[Section 7.7.1]{jeanblanc_mathematical_2009}.
Since $N$ is a single-default process, we take this density to be zero after default.

\begin{rem}\label{rem:MRT}~\\
In this setting, every square-integrable $(\mathbb{P}^{0,1},\mathbb{F}^{0,1})$-martingale $X$ admits a representation
\begin{equation}
  X_t=X_0+\int_0^ta^0_sdW^0_s +\int_0^ta^1_sdW^1_s +\int_0^tb_sdM_s,~~~t\in[0,T],~~~\mathbb{P}^{0,1}\text{-a.s.},
\end{equation}
with $X_0\in\mathbb{L}^2(\mathcal{F}^{0,1}_0,\mathbb{R})$, $a^0\in\mathbb{H}^2(\mathbb{F}^{0,1},\mathbb{R}^{1\times d_0})$, $a^1\in\mathbb{H}^2(\mathbb{F}^{0,1},\mathbb{R}^{1\times d})$ and $b\in\mathbb{L}^2(M;\mathbb{F}^{0,1},\mathbb{R})$.
See Lim-Quenez \cite[Lemma 2.1]{LimQuenez2011} and Jeanblanc-Yor-Chesney \cite[Theorem 7.5.5.1]{jeanblanc_mathematical_2009}.
\end{rem}

\subsection{Market and utility function}

We consider a financial market specified as follows.
\begin{asm}~\label{asm:market}
  \begin{enumerate}
    \item[(i)] The risk-free interest rate is zero.
    \item[(ii)] There are $n\in\mathbb{N}$ non-dividend paying risky stocks whose price dynamics is given by
          \begin{equation}\label{eq:stock}
            S_t = S_0 + \int_0^t\mathrm{diag}(S_{s-})\Bigl(\mu_s ds + \sigma_sdW^0_s+ \beta_sdN_s \Bigr),~~~ t\in[0,T],
          \end{equation}
          where $S_0\in\mathbb{R}_{++}^n$, $\mu=(\mu_t)_{t\in[0,T]}$ is an $\mathbb{R}^n$-valued, $\mathbb{F}^0$-predictable and bounded process, $\sigma=(\sigma_t)_{t\in[0,T]}$ is an $\mathbb{R}^{n\times d_0}$-valued, $\mathbb{F}^0$-predictable and bounded process satisfying
            \begin{equation}
              c_{\sigma}I_n\leq \sigma_t\sigma_t^{\top}\leq C_{\sigma} I_n,~~~ dt\otimes\mathbb{P}^0\text{-a.e.}
            \end{equation}
            for some $0<c_{\sigma}\leq C_{\sigma}<\infty$, and $\beta=(\beta_t)_{t\in[0,T]}$ is an $\mathbb{R}^n$-valued, $\mathbb{F}^0$-predictable, bounded process satisfying $\beta^i_t>-1$ for each $i=1,\ldots,n$, $dt\otimes\mathbb{P}^0$-a.e.
            We assume $n\leq d_0$.

    \item[(iii)] The default compensator density $\lambda$ is $\mathbb{F}^0$-predictable and satisfies $0\leq\lambda_t\leq\overline{\lambda}$ and $\lambda_t\geq\underline{\lambda}$ on $\{\lambda\neq0\}$, $dt\otimes\mathbb{P}^0$-a.e., for constants $0<\underline{\lambda}\leq\overline{\lambda}<\infty$.
    \item[(iv)] For some $\underline\beta>0$, the process $\beta$ satisfies
          \begin{equation}
            |\beta_t|^2\geq \underline\beta^2\mathbf{1}_{\{\lambda>0\}},~~~ dt\otimes\mathbb P^0\text{-a.e.}
          \end{equation}
  \end{enumerate}
\end{asm}

Set $\eta_t:=\sigma_t^\top(\sigma_t\sigma_t^\top)^{-1}\beta_t$ for $t\in[0,T]$ so that $\beta_t=\sigma_t\eta_t$.
Since $\beta$ is bounded and $\sigma\sigma^\top$ is uniformly non-degenerate, Assumption \ref{asm:market} (iv) implies that there exist constants $0<\underline h\leq\overline h<\infty$ such that
\begin{equation}
  \underline h^2\mathbf{1}_{\{\lambda>0\}}\leq |\eta_t|^2\leq\overline h^2,~~~ dt\otimes\mathbb P^0\text{-a.e.}
\end{equation}

Since the interest rate is zero, the Brownian risk-premium process $\theta=(\theta_t)_{t\in[0,T]}$ is defined by
\begin{equation}
  \theta_t := \sigma_t^{\top}(\sigma_t\sigma_t^{\top})^{-1}\mu_t,~~~t\in[0,T].
\end{equation}
Note that $\theta_t\in\mathrm{Range}(\sigma_t^{\top})=\mathrm{Ker}(\sigma_t)^{\perp}$.  Since $\mu$ and $\sigma$ are bounded and $\sigma\sigma^{\top}\geq c_{\sigma}I_n$, we have $\theta\in\mathbb{L}^{\infty}(\mathbb{F}^0,\mathbb{R}^{d_0})$.

\begin{rem}~\\\label{rem:lambda-gap}
    The assumption $\beta^i_t>-1$ ensures that the stock price $S^i$ remains positive for each $i=1,\ldots,n$.
\end{rem}

\begin{dfn}~\\\label{def:Ls}
  For each $s\in[0,T]$, let $L_s := \{u^{\top}\sigma_s; u\in\mathbb{R}^n\}$ be the linear subspace of $\mathbb{R}^{1\times d_0}$ spanned by the $n$ row-vectors of $\sigma_s$. 
  For $z\in\mathbb{R}^{1\times d_0}$, $\Pi_s(z)$ denotes the orthogonal projection of $z$ onto $L_s$.
\end{dfn}

We use the shorthand $z^{\|}_s := \Pi_s(z)$ and $z^{\perp}_s := z - \Pi_s(z)$. Note that $\theta_s^{\top}\in L_s$ for every $s\in[0,T]$.

\begin{rem}~\\\label{rem:Ls-predictable}
Since $\sigma$ is $\mathbb{F}^0$-predictable and $\sigma_s\sigma_s^\top$ is uniformly non-degenerate, the map $(s,\omega^0,z)\mapsto\Pi_s(\omega^0,z)$ is $\mathcal{P}^0\otimes\mathcal{B}(\mathbb{R}^{1\times d_0})$-measurable (see \cite[Chapter 1, Lemma 4.4]{karatzas_methods_1998}).
In particular, for every $\mathbb{F}^0$-predictable process $z$, the projected process $z^\|$ is also predictable.
\end{rem}

We characterize agent-1 as follows.

\begin{asm}~\\\label{asm:agent}
  The idiosyncratic environment of agent-1 is specified by the triple $(\xi^1,\gamma^1,F^1)$ where:
  \begin{enumerate}
    \item[(i)] $\xi^1$ is an $\mathbb{R}$-valued, bounded, $\mathcal{F}^1_0$-measurable random variable representing agent-1's initial wealth.
    \item[(ii)] $\gamma^1$ is an $\mathbb{R}_{++}$-valued, bounded, and $\mathcal{F}^1_0$-measurable random variable with $\underline{\gamma}\leq\gamma^1\leq\overline{\gamma}$ for some $0<\underline{\gamma}\leq\overline{\gamma}$ representing agent-1's risk-aversion parameter.
    \item[(iii)] $F^1$ is an $\mathbb{R}$-valued, bounded, and $\mathcal{F}^{0,1}_T$-measurable random variable representing agent-1's terminal liability.
    \item[(iv)] Agent-1 is a price taker, so agent-1's action has no impact on the market.
  \end{enumerate}
\end{asm}

The wealth process of agent-1 under self-financing strategy $\pi=(\pi_t)_{t\in[0,T]}$ (representing the amount of money invested in each stock) is
\begin{equation}\label{eq:wealth}
  \mathcal{W}^{1,\pi}_t= \xi^1 + \int_0^t\pi_s^{\top}\sigma_sdW^0_s + \int_0^t\pi_s^{\top}\sigma_s\theta_sds+ \int_0^t\pi_s^{\top}\beta_sdN_s,~~~t\in[0,T].
\end{equation}
Agent-1's problem is
\begin{equation}
  \sup_{\pi\in\mathbb{A}^1}\mathcal U^1(\pi), ~~~\text{where}~~~
  \mathcal U^1(\pi):= \mathbb{E}\Bigl[-\exp\Bigl(-\gamma^1\Bigl(\mathcal{W}^{1,\pi}_T-F^1\Bigr)\Bigr)\Bigr],
\end{equation}
subject to \eqref{eq:wealth}. The admissible set for agent-1 is defined as follows.

\begin{dfn}~\\\label{def:admissible}
  The admissible set $\mathbb{A}^1$ consists of all $\mathbb{R}^n$-valued, $\mathbb{F}^{0,1}$-predictable processes $\pi$ such that
  \begin{enumerate}
    \item[(i)] $\displaystyle\mathbb{E}\Bigl[\int_0^T|\pi_s^{\top}\sigma_s|^2ds\Bigr]<\infty$ and $\displaystyle\mathbb{E}\Bigl[\int_0^T\lambda_s|\pi_s^{\top}\beta_s|^2ds\Bigr]<\infty$,
    \item[(ii)] the family $\{\exp(-\gamma^1\mathcal{W}^{1,\pi}_{\tau}); \tau\in\mathcal{T}(\mathbb{F}^{0,1})\}$ is uniformly integrable.
  \end{enumerate}
\end{dfn}

\begin{rem}~\\\label{rem:admissible-compensator}
  Let $\pi\in\mathbb{A}^1$. Since $\mathcal{W}^{1,\pi}_{\tau}=\mathcal{W}^{1,\pi}_{\tau-}+\pi_{\tau}^{\top}\beta_{\tau}$ on $\{\tau\leq T\}$, we have
  \begin{equation}
    \mathbb{E}\Bigl[\int_0^T\lambda_s e^{-\gamma^1(\mathcal{W}^{1,\pi}_{s-}+\pi_s^{\top}\beta_s)}ds\Bigr]
    =\mathbb{E}\Bigl[\int_0^Te^{-\gamma^1(\mathcal{W}^{1,\pi}_{s-}+\pi_s^{\top}\beta_s)}dN_s\Bigr]
    =\mathbb{E}\bigl[e^{-\gamma^1\mathcal{W}^{1,\pi}_{\tau}}\mathbf{1}_{\{\tau\leq T\}}\bigr]
    \leq\mathbb{E}\bigl[e^{-\gamma^1\mathcal{W}^{1,\pi}_{\tau\wedge T}}\bigr]<\infty.
  \end{equation}
  This shows
  \begin{equation}
    \int_0^T\lambda_s e^{-\gamma^1(\mathcal{W}^{1,\pi}_{s-}+\pi_s^{\top}\beta_s)}ds<\infty,~~~\mathbb{P}^{0,1}\text{-a.s.}
  \end{equation}
  Since $\sup_{s\in[0,T]}\mathcal{W}^{1,\pi}_{s-}<\infty$ pathwise, we deduce that
  \begin{equation}\label{eq:admissible-compensator}
    \int_0^T\lambda_s e^{-\gamma^1\pi_s^{\top}\beta_s}ds<\infty,~~~\mathbb{P}^{0,1}\text{-a.s.}
  \end{equation}
\end{rem}

Define $p_t:=\pi_t^{\top}\sigma_t\in L_t\subset\mathbb{R}^{1\times d_0}$ for $t\in[0,T]$ so that the wealth dynamics reads
\begin{equation}\label{eq:wealth-p}
  \mathcal{W}^{1,p}_t= \xi^1 + \int_0^tp_s dW^0_s + \int_0^tp_s\theta_s ds+ \int_0^tp_s\eta_s dN_s,~~~~~~ t\in[0,T].
\end{equation}
Agent-1's optimization problem can be written as
\begin{equation}
  \sup_{p\in\mathcal{A}^1}\widetilde{\mathcal U}^1(p),~~~~~~ \widetilde{\mathcal U}^1(p):= \mathbb{E}\Bigl[-\exp\Bigl(-\gamma^1\bigl(\mathcal{W}^{1,p}_T-F^1\bigr)\Bigr)\Bigr],
\end{equation}
subject to \eqref{eq:wealth-p}, where $\mathcal{A}^1:=\{p=\pi^{\top}\sigma; \pi\in\mathbb{A}^1\}$.

\subsection{Derivation of the BSDE for optimality}
\label{sec:ito-driver}

We apply the martingale optimality principle of \cite{HIM2005} in the jump setting. We seek a family $\{R^{1,p}\}_{p\in\mathcal{A}^1}$ of c\`adl\`ag processes satisfying the following conditions.

\begin{dfn}[Condition-R]~\\\label{def:condR}
  A family $\{R^{1,p}\}_{p\in\mathcal{A}^1}$ of c\`adl\`ag processes is said to satisfy Condition-R if it meets the following conditions.
  \begin{enumerate}
    \item[(i)]   $R^{1,p}_T=-\exp(-\gamma^1(\mathcal{W}^{1,p}_T-F^1))$, $\mathbb{P}^{0,1}$-a.s. for all $p\in\mathcal{A}^1$.
    \item[(ii)]  $R^{1,p}_0=R_0$, $\mathbb{P}^{0,1}$-a.s. for some $\mathcal{F}^{0,1}_0$-measurable random variable $R_0$, for all $p\in\mathcal{A}^1$.
    \item[(iii)] $R^{1,p}$ is an $(\mathbb{F}^{0,1},\mathbb{P}^{0,1})$-supermartingale for all $p\in\mathcal{A}^1$, and there exists $p^{*}\in\mathcal{A}^1$ for which $R^{1,p^*}$ is a martingale.
  \end{enumerate}
\end{dfn}

\begin{rem}~\\\label{rem:condition-R-verification}
If a family $\{R^{1,p}\}_{p\in\mathcal A^1}$ satisfies Condition-R, then for every admissible $p\in\mathcal A^1$ and the martingale control $p^*\in\mathcal{A}^1$ in Condition-R(iii), we have
\begin{equation}
  \widetilde{\mathcal U}^1(p)=\mathbb E[R^{1,p}_T] \leq \mathbb E[R^{1,p}_0] =\mathbb E[R^{1,p^*}_T] =\widetilde{\mathcal U}^1(p^*).
\end{equation}
This indicates that $p^*$ is an optimal control.
\end{rem}

We try the ansatz
\begin{equation}\label{eq:Rpi}
  R^{1,p}_t= -\exp\Bigl(-\gamma^1\Bigl(\mathcal{W}^{1,p}_t-Y^1_t\Bigr)\Bigr),~~~t\in[0,T],
\end{equation}
where $Y^1$ is an $\mathbb{F}^{0,1}$-adapted c\`adl\`ag solution to the BSDE:
\begin{equation}\label{eq:BSDE-Y}
  Y^1_t= F^1 + \int_t^Tf^1(s,Z^{1,0}_s,Z^1_s,U^1_s)ds - \int_t^TZ^{1,0}_sdW^0_s - \int_t^TZ^1_sdW^1_s - \int_t^TU^1_sdM_s,~~~ t\in[0,T].
\end{equation}

From \eqref{eq:wealth-p} and \eqref{eq:BSDE-Y}, we obtain
\begin{equation}\label{eq:dW-Y}
  d(\mathcal{W}^{1,p}_t-Y^1_t)=\bigl(p_t\theta_t+f^1_t+\lambda_tU^1_t\bigr)dt + (p_t-Z^{1,0}_t)dW^0_t - Z^1_tdW^1_t + (p_t\eta_t-U^1_t)dN_t.
\end{equation}
Set $q^{1,p}:=\gamma^1(p\eta-U^1)$. Applying It\^o formula to $R^{1,p}$, we obtain
\begin{equation}\label{eq:drift-Rpi}
  \begin{split}
  dR^{1,p}_t &= R^{1,p}_t\Biggl[-\gamma^1\bigl(p_t\theta_t +f^1_t+\lambda_tU^1_t\bigr)+\frac{(\gamma^1)^2}{2}\bigl(|p_t-Z^{1,0}_t|^2+|Z^1_t|^2\bigr)+ \lambda_t\bigl(e^{-q^{1,p}_t}-1\bigr)\Biggr]dt\\
  &~~~+ R^{1,p}_t\Bigl[-\gamma^1(p_t-Z^{1,0}_t)dW^0_t + \gamma^1Z^1_tdW^1_t\Bigr] + R^{1,p}_{t-}\bigl(e^{-q^{1,p}_t}-1\bigr)dM_t,~~~ t\in[0,T].
  \end{split}
\end{equation}
Note that Remark \ref{rem:admissible-compensator} gives $\int_0^T\lambda_s e^{-\gamma^1p_s\eta_s}ds<\infty$, $\mathbb{P}^{0,1}$-a.s. for every $p\in\mathcal{A}^1$.
Hence the compensated jump integral in \eqref{eq:drift-Rpi} is a local martingale whenever $e^{\gamma^1U^1}$ is bounded on $\{\lambda>0\}$.
In order for $R^{1,p}$ to be a supermartingale for every $p\in\mathcal{A}^1$, we need
\begin{equation}\label{eq:supermart-cond}
  -\gamma^1\bigl(p_t\theta_t +f^1_t+\lambda_tU^1_t\bigr)+\frac{(\gamma^1)^2}{2}\bigl(|p_t-Z^{1,0}_t|^2+|Z^1_t|^2\bigr)+ \lambda_t\bigl(e^{-q^{1,p}_t}-1\bigr) \geq 0,~~~~~~ t\in[0,T],
\end{equation}
for every $p\in\mathcal{A}^1$, with equality attained at some $p^{*}\in\mathcal{A}^1$, which is necessary to make $R^{1,p^{*}}$ a true martingale.
Hence, the driver $f^1$ must satisfy
\begin{equation}\label{eq:driver-inf}
  f^1(s,Z^{1,0}_s,Z^1_s,U^1_s)= \inf_{p\in L_s}\Bigl\{-p\theta_s+\frac{\gamma^1}{2}|p-Z^{1,0}_s|^2+\frac{\gamma^1}{2}|Z^1_s|^2 -\lambda_sU^1_s+\frac{\lambda_s}{\gamma^1}\bigl(e^{-q^{1,p}_s}-1\bigr)\Bigr\}.
\end{equation}
Since the objective in \eqref{eq:driver-inf} is strongly convex in $p\in L_s$, it admits a unique minimizer $p^*_s\in L_s$.
We have
\begin{equation}\label{eq:driver-explicit}
  f^1(s,Z^{1,0}_s,Z^1_s,U^1_s) = -Z^{1,0\|}_s\theta_s -\frac{|\theta_s|^2}{2\gamma^1} +\frac{\gamma^1}{2}\Bigl(|Z^{1,0\perp}_s|^2+|Z^1_s|^2\Bigr) -\lambda_sU^1_s +\frac{\gamma^1}{2}\Bigl|p^{*}_s-Z^{1,0\|}_s -\frac{\theta_s^{\top}}{\gamma^1}\Bigr|^2 +\frac{\lambda_s}{\gamma^1} \bigl(e^{-q^{1,*}_s}-1\bigr),
\end{equation}
where $q^{1,*}:=q^{1,p^*}$ and the relation $|p^*_t-Z^{1,0}_t|^2 = |p^*_t-Z^{1,0\|}_t|^2 + |Z^{1,0\perp}_t|^2$ is used.

The minimizer $p^{*}_s\in L_s$ satisfies the first order condition
\begin{equation}\label{eq:FOC}
  p^{*}_s = Z^{1,0\|}_s + \frac{\theta^{\top}_s}{\gamma^1} + \frac{\lambda_s}{\gamma^1}\eta^{\top}_se^{-q^{1,*}_s},~~~s\in[0,T].
\end{equation}
This implies
\begin{equation}\label{eq:residual-square}
  \frac{\gamma^1}{2}\Bigl|p^{*}_s-Z^{1,0\|}_s-\frac{\theta^{\top}_s}{\gamma^1}\Bigr|^2 = \frac{\lambda^2_s|\eta_s|^2}{2\gamma^1}e^{-2q^{1,*}_s},~~~s\in[0,T].
\end{equation}
Then, the driver becomes
\begin{equation}\label{eq:driver-q}
  f^1(s,Z^{1,0}_s,Z^1_s,U^1_s) = -Z^{1,0\|}_s\theta_s-\frac{|\theta_s|^2}{2\gamma^1}+\frac{\gamma^1}{2}\bigl(|Z^{1,0\perp}_s|^2+|Z^1_s|^2\bigr)-\lambda_sU^1_s+\frac{\lambda^2_s|\eta_s|^2}{2\gamma^1} e^{-2q^{1,*}_s}+\frac{\lambda_s}{\gamma^1}\bigl(e^{-q^{1,*}_s}-1\bigr)
\end{equation}
for $s\in[0,T]$. Introducing rescaled variables $(y^1,z^{1,0},z^1,u^1):=(\gamma^1Y^1,\gamma^1Z^{1,0},\gamma^1Z^1,\gamma^1U^1)$ and $G^1:=\gamma^1F^1$, define the rescaled driver
\begin{equation}\label{eq:driver-g}
  g^1_s(z^0,z^1,u)
  :=
  -z^{0\|}\theta_s-\frac{|\theta_s|^2}{2}
  +\frac{1}{2}\bigl(|z^{0\perp}|^2+|z^1|^2\bigr)-\lambda_su
  +\frac{\lambda_s^2|\eta_s|^2}{2}e^{-2q^{1,*}_s}
  +\lambda_s(e^{-q^{1,*}_s}-1).
\end{equation}
Then the BSDE \eqref{eq:BSDE-Y} becomes
\begin{equation}\label{eq:BSDE-y}
  y^1_t = G^1 + \int_t^Tg^1_s(z^{1,0}_s,z^1_s,u^1_s)ds - \int_t^Tz^{1,0}_s dW^0_s - \int_t^Tz^1_s dW^1_s - \int_t^Tu^1_s dM_s, ~~~~~~ t\in[0,T],
\end{equation}
where $q^{1,*}_s\in\mathbb{R}$ is implicitly determined by
\begin{equation}\label{eq:FOC-rescaled}
  q^{1,*}_s-\lambda_s|\eta_s|^2 e^{-q^{1,*}_s} = z^{1,0\|}_s\eta_s-u^1_s+\theta^{\top}_s\eta_s,~~~s\in[0,T].
\end{equation}

For $a\geq0$, define $Q_a:\mathbb R\to\mathbb R$ as the inverse of $q\mapsto q-ae^{-q}$, equivalently,
\begin{equation}
  Q_a(r)-ae^{-Q_a(r)}=r,~~~ r\in\mathbb R.
\end{equation}
Such a function $Q_a$ is well-defined since $q\mapsto q-ae^{-q}$ is a strictly increasing bijection on $\mathbb{R}$. Moreover, $(a,r)\mapsto Q_a(r)$ is continuous, and, by the implicit function theorem,
\begin{equation}\label{eq:qstar-derivative}
  \frac{d}{dr}Q_a(r)=\frac{1}{1+ae^{-Q_a(r)}}\in(0,1],~~~ r\in\mathbb R.
\end{equation}
Slightly abusing notation, we hereafter write
\begin{equation}\label{eq:Q-def}
  Q_s(r):=Q_{\lambda_s|\eta_s|^2}(r),~~~ s\in[0,T],~~~r\in\mathbb R,
\end{equation}
for the predictable random field $(s,\omega^0,r)\mapsto Q_{\lambda_s(\omega^0)|\eta_s(\omega^0)|^2}(r)$.
For later use, set
\begin{equation}
  J_s(q):=\frac{\lambda_s^2|\eta_s|^2}{2}e^{-2q} +\lambda_s(e^{-q}-1),~~~ s\in[0,T],~~~q\in\mathbb R.
\end{equation}

Therefore, for given predictable processes $z^{1,0}$ and $u^1$, there is a unique predictable process $q^{1,*}$ satisfying \eqref{eq:FOC-rescaled}.
Once such $q^{1,*}$ is determined, $p^{1,*}$ can be recovered from \eqref{eq:FOC}.

\begin{rem}\label{rem:econ-q}~\\
  The equation \eqref{eq:FOC} provides the candidate for the optimal control. Its optimality and admissibility will be verified in Section \ref{sec:wellposed}.
  In contrast to the diffusive model of \cite{FujiiSekine2025}, \eqref{eq:FOC} contains the additional term $\dfrac{\lambda_s}{\gamma^1}\eta_s^{\top}e^{-q^{1,*}_s}$.
  For the economic interpretation of this term, consider the scalar case $n=d_0=1$. Then
  \begin{equation}
    \pi^{1,*}_s=\frac{p^{1,*}_s}{\sigma_s}
    =
    \frac{Z^{1,0}_s}{\sigma_s}
    + \frac{\theta_s}{\gamma^1\sigma_s}
    + \frac{\lambda_s\beta_s}{\gamma^1\sigma_s^2}e^{-q^{1,*}_s},~~~s\in[0,T].
  \end{equation}
  The last term is the demand component generated by default exposure.
  Its sign is determined by $\beta_s$. If $\beta_s<0$, the stock price would decrease if default occurred at time $s$, and this term reduces the pre-default demand, whereas if $\beta_s>0$, the security pays off at default and the term increases the demand.
  The factor $e^{-q^{1,*}_s}$ scales this adjustment.
  The process $R^{1,p^*}$ represents the conditional expected utility (verified below), and
  \begin{equation}
    e^{-q^{1,*}_{\tau}}=\frac{R^{1,p^*}_{\tau}}{R^{1,p^*}_{\tau-}}
  \end{equation}
  on $\{\tau\leq T\}$ by \eqref{eq:dW-Y} and \eqref{eq:Rpi}.
  Thus, $e^{-q^{1,*}_s}$ is the potential jump ratio of the conditional expected utility if default occurred at time $s$.
  Since $R^{1,p^*}$ is negative, a larger $e^{-q^{1,*}_s}$ corresponds to a less favorable default state and, through the optimal-demand formula above, increases the magnitude of the default-induced demand adjustment.
\end{rem}

\begin{lem}~\\\label{lem:qstar-apriori}
  Let Assumption \ref{asm:market} be in force.
  Let $\varrho$ be an $\mathbb{R}$-valued $\mathbb F^{0,1}$-predictable process.
  Then $(Q_s(\varrho_s))_{s\in[0,T]}$ is predictable and
  \begin{equation}\label{eq:qstar-root-exp-bound}
    |Q_s(\varrho_s)|\leq |\varrho_s|+\lambda_s|\eta_s|^2 \leq |\varrho_s|+\overline\lambda\overline h^2,~~~\lambda_s e^{-Q_s(\varrho_s)}\leq C(1+|\varrho_s|),~~~dt\otimes\mathbb P^{0,1}\text{-a.e.}
  \end{equation}
  Moreover, there exists a constant $C_J$, depending only on
  $\overline\lambda,\overline h,\underline h$, such that
  \begin{equation}\label{eq:qstar-J-bound}
    |J_s(Q_s(\varrho_s))|\leq C_J(1+|\varrho_s|^2),
    ~~~ dt\otimes\mathbb P^{0,1}\text{-a.e.}
  \end{equation}
  Finally,
  \begin{equation}\label{eq:qstar-J-derivative}
    \frac{d}{dr}J_s\bigl(Q_s(r)\bigr) = -\lambda_s e^{-Q_s(r)},~~~r\in\mathbb{R},~~~dt\otimes\mathbb P^{0,1}\text{-a.e.}
  \end{equation}
\end{lem}

\noindent\textbf{\textit{Proof}}\\
The predictability of $Q_s(\varrho_s)$ follows from the measurability of the random field $Q_s$ and the predictability of $\varrho$.
By \eqref{eq:qstar-derivative}, $r\mapsto Q_s(r)$ is $1$-Lipschitz.  Since $Q_s(-\lambda_s|\eta_s|^2)=0$, we have
\begin{equation}
  |Q_s(\varrho_s)|=|Q_s(\varrho_s)-Q_s(-\lambda_s|\eta_s|^2)|
  \leq |\varrho_s+\lambda_s|\eta_s|^2|
  \leq |\varrho_s|+\overline\lambda\overline h^2.
\end{equation}
If $Q_s(\varrho_s)\geq0$, then $e^{-Q_s(\varrho_s)}\leq1$, and hence $\lambda_s e^{-Q_s(\varrho_s)}\leq\overline\lambda$.
If $Q_s(\varrho_s)<0$, then $\varrho_s=Q_s(\varrho_s)-\lambda_s|\eta_s|^2e^{-Q_s(\varrho_s)}<0$ and $|\varrho_s|=|Q_s(\varrho_s)|+\lambda_s|\eta_s|^2e^{|Q_s(\varrho_s)|}$.
Thus $\lambda_s|\eta_s|^2e^{-Q_s(\varrho_s)}\leq|\varrho_s|$.  
On $\{\lambda_s>0\}$, we have $|\eta_s|^2\geq\underline h^2$, and therefore
\begin{equation}
  \lambda_s e^{-Q_s(\varrho_s)}
  =\frac{\lambda_s|\eta_s|^2e^{-Q_s(\varrho_s)}}{|\eta_s|^2}
  \leq\frac{|\varrho_s|}{\underline h^2}.
\end{equation}
On $\{\lambda_s=0\}$ the same estimate is trivial.  
This proves \eqref{eq:qstar-root-exp-bound}.

For the bound on $J$, \eqref{eq:qstar-J-bound} is trivial on $\{\lambda_s=0\}$.  
On $\{\lambda_s>0\}$, if $Q_s(\varrho_s)\geq0$, then
\begin{equation}
  \frac{\lambda_s^2|\eta_s|^2}{2}e^{-2Q_s(\varrho_s)}
  \leq \frac{\overline\lambda^2\overline h^2}{2},
  ~~~
  |\lambda_s(e^{-Q_s(\varrho_s)}-1)|\leq\overline\lambda .
\end{equation}
If $Q_s(\varrho_s)<0$, then $\lambda_s|\eta_s|^2e^{-Q_s(\varrho_s)}\leq|\varrho_s|$ by the previous paragraph. Hence
\begin{equation}
  \frac{\lambda_s^2|\eta_s|^2}{2}e^{-2Q_s(\varrho_s)}
  =\frac{(\lambda_s|\eta_s|^2e^{-Q_s(\varrho_s)})^2}{2|\eta_s|^2}
  \leq\frac{|\varrho_s|^2}{2\underline h^2},
  ~~~
  |\lambda_s(e^{-Q_s(\varrho_s)}-1)|
  \leq \lambda_s e^{-Q_s(\varrho_s)}
  \leq \frac{|\varrho_s|}{\underline h^2}.
\end{equation}
Combining the two cases proves \eqref{eq:qstar-J-bound}.

Finally,
\begin{equation}
  \partial_qJ_s(q)=-\lambda_s^2|\eta_s|^2e^{-2q}-\lambda_se^{-q}  =-\lambda_se^{-q}\bigl(1+\lambda_s|\eta_s|^2e^{-q}\bigr),
\end{equation}
and multiplying by \eqref{eq:qstar-derivative} gives \eqref{eq:qstar-J-derivative}. $\square$

\subsection{Well-posedness of the BSDE and verification}\label{sec:wellposed}

\begin{lem}~\\\label{lem:Sbounded-implies-BMO}
  Let Assumptions \ref{asm:market} and \ref{asm:agent} hold.
  If $(y^1,z^{1,0},z^1,u^1)\in\mathbb{S}^{\infty}(\mathbb{P}^{0,1},\mathbb{F}^{0,1},\mathbb{R})\times\mathbb{H}^2(\mathbb{P}^{0,1},\mathbb{F}^{0,1},\mathbb{R}^{1\times d_0}) \times\mathbb{H}^2(\mathbb{P}^{0,1},\mathbb{F}^{0,1},\mathbb{R}^{1\times d})\times\mathbb{L}^2(M;\mathbb{P}^{0,1},\mathbb{F}^{0,1},\mathbb{R})$ is a solution to the BSDE \eqref{eq:BSDE-y}, then one has
  \begin{equation}
    |u^1_t|\leq 2\|y^1\|_{\mathbb{S}^{\infty}},~~~  dt\otimes \mathbb{P}^{0,1}\text{-a.e. on }\{\lambda>0\}:=\{(t,\omega^0)\in[0,T]\times\Omega^0:\lambda_t(\omega^0)>0\},
  \end{equation}
  $(z^{1,0},z^1)\in\mathbb{H}^2_{\mathrm{BMO}}(\mathbb{P}^{0,1},\mathbb{F}^{0,1},\mathbb{R}^{1\times d_0})\times\mathbb{H}^2_{\mathrm{BMO}}(\mathbb{P}^{0,1},\mathbb{F}^{0,1},\mathbb{R}^{1\times d})$ and $u^1\in\mathbb{L}^2_{\mathrm{BMO}}(M;\mathbb{P}^{0,1},\mathbb{F}^{0,1},\mathbb{R})$.
\end{lem}

\noindent\textbf{\textit{Proof}}\\
For notational simplicity, we omit the superscript ``1'' in this proof when it is obvious.
Write $K:=\|y\|_{\mathbb{S}^{\infty}}$.
Since the jump of $y$ occurs only at $\tau$ and $\Delta M_\tau=\Delta N_\tau=1$, it holds that $\Delta y_\tau=u_\tau$ on $\{\tau\leq T\}$. This yields
\begin{equation}
  |u_\tau|=|y_\tau-y_{\tau-}|\leq 2K,~~~\text{on }\{\tau\leq T\}.
\end{equation}
Set $H_s:=(|u_s|-2K)^+\mathbf{1}_{\{\lambda_s>0\}}\geq0$ for $s\in[0,T]$. 
Then, $H$ is predictable and satisfies $H_\tau=0$.
Moreover, we have
\begin{equation}
  \mathbb{E}\int_0^T H_s\lambda_sds =\mathbb{E}\int_0^T H_sdN_s =\mathbb{E}\bigl[H_\tau\mathbf{1}_{\{\tau\leq T\}}\bigr]=0.
\end{equation}
As $H\lambda\geq0$, this forces $H=0$ on $\{\lambda>0\}$, i.e., $|u_t|\leq 2K$, $dt\otimes \mathbb{P}^{0,1}$-a.e. on $\{\lambda>0\}$.
$u\in \mathbb{L}^2_{\mathrm{BMO}}(M)$ is now straightforward since
\begin{equation}
  \|u\|^2_{\mathbb{L}^2_{\mathrm{BMO}}(M)}
  =\sup_{\tau_0\in\mathcal{T}^{0,1}}\Bigl\|\mathbb{E}\Bigl[\int_{\tau_0}^T\lambda_s|u_s|^2 ds\big|\mathcal{F}^{0,1}_{\tau_0}\Bigr]\Bigr\|_{\mathbb{L}^{\infty}} \leq 4K^2\overline{\lambda} T<\infty.
\end{equation}

Set $\varrho_s:=z^{0\|}_s\eta_s-u_s+\theta_s^\top\eta_s$.  
The rescaled driver \eqref{eq:driver-g} reads
\begin{equation}
  g_s(z^0,z^1,u)=-z^{0\|}\theta_s-\frac{|\theta_s|^2}{2}
  +\frac{1}{2}(|z^{0\perp}|^2+|z^1|^2)-\lambda_su
  +J_s(Q_s(\varrho_s)).
\end{equation}
The bound \eqref{eq:qstar-J-bound} gives $|J_s(Q_s(\varrho_s))|\leq C_J(1+|\varrho_s|^2)$.
Since $J_s(Q_s(\varrho_s))=0$ on $\{\lambda=0\}$, using $|\eta|\leq\overline{h}$ and $|u|\leq 2K$ on $\{\lambda>0\}$, we obtain
\begin{equation}\label{eq:J-bound-Sbmolem}
  |J_s(Q_s(\varrho_s))|\leq C'_J\bigl(1+|z^0_s|^2+|\theta_s|^2\bigr),~~~
  dt\otimes\mathbb{P}^{0,1}\text{-a.e.}
\end{equation}
for some constant $C'_J>0$.
Applying It\^o formula to $e^{2\nu y_t}$ for $\nu>0$ to be chosen,
\begin{equation}
\begin{split}
  d e^{2\nu y_t}&= 2\nu e^{2\nu y_t}\bigl[(-g_t-\lambda_t u_t)dt+z^0_td W^0_t+z^1_td W^1_t\bigr]+2\nu^2e^{2\nu y_t}\bigl(|z^0_t|^2+|z^1_t|^2\bigr)dt\\
  &~~~+e^{2\nu y_{t-}}(e^{2\nu u_t}-1)dN_t,~~~t\in[0,T].
\end{split}
\end{equation}
Integrating over $[\tau_0,T]$ and taking conditional expectation $\mathbb{E}[\cdot|\mathcal{F}^{0,1}_{\tau_0}]$, we have
\begin{equation}\label{eq:exp-energy}
  \mathbb{E}\bigl[e^{2\nu y_T}-e^{2\nu y_{\tau_0}}\bigm|\mathcal{F}^{0,1}_{\tau_0}\bigr]= \mathbb{E}\Bigl[\int_{\tau_0}^T2\nu e^{2\nu y_s}\Bigl\{-g_s-\lambda_su_s+\nu(|z^0_s|^2+|z^1_s|^2)+\frac{\lambda_s(e^{2\nu u_s}-1)}{2\nu}\Bigr\}ds\Big|\mathcal{F}^{0,1}_{\tau_0}\Bigr].
\end{equation}

Choose $\nu:=C'_J+1$.
Using \eqref{eq:J-bound-Sbmolem} and $z^{0\|}_s\theta_s\geq -\dfrac{1}{2}(|z^{0\|}_s|^2 + |\theta_s|^2)$,
\begin{equation}
  \begin{split}
    &-g_s-\lambda_s u_s+\nu(|z^0_s|^2+|z^1_s|^2)+\frac{\lambda_s(e^{2\nu u_s}-1)}{2\nu}\\
    &~~= \frac{|\theta_s|^2}{2}+\Bigl(\nu-\frac{1}{2}\Bigr)\bigl(|z^{0\perp}_s|^2+|z^1_s|^2\bigr)+\nu|z^{0\|}_s|^2+z^{0\|}_s\theta_s-J_s(Q_s(\varrho_s))+\frac{\lambda_s(e^{2\nu u_s}-1)}{2\nu}\\
    &~~\geq \frac{|\theta_s|^2}{2}+\frac{1}{2}|z^{0\perp}_s|^2+\Bigl(C'_J+\frac{1}{2}\Bigr)|z^1_s|^2+|z^{0\|}_s|^2+z^{0\|}_s\theta_s-C'_J|\theta_s|^2-C'_J-\frac{\overline{\lambda}}{2\nu}\\
    &~~\geq \frac{1}{2}\bigl(|z^0_s|^2+|z^1_s|^2\bigr)-C'_J|\theta_s|^2-C'_J-\frac{\overline{\lambda}}{2\nu}\\
    &~~\geq \frac{1}{2}\bigl(|z^0_s|^2+|z^1_s|^2\bigr)-C\bigl(1+|\theta_s|^2\bigr).
  \end{split}
\end{equation}
Plugging into \eqref{eq:exp-energy}, and using $e^{-2\nu K}\leq e^{2\nu y_s}\leq e^{2\nu K}$, and $|e^{2\nu y_T}-e^{2\nu y_{\tau_0}}|\leq 2e^{2\nu K}$, we have
\begin{equation}\label{eq:bmo-z}
  \mathbb{E}\Bigl[\int_{\tau_0}^T(|z^0_s|^2+|z^1_s|^2)ds\Big|\mathcal{F}^{0,1}_{\tau_0}\Bigr]\leq C e^{4\nu K}\bigl(1+\|\theta\|^2_{\mathbb{L}^{\infty}}\bigr).
\end{equation}
Taking the supremum over $\tau_0\in\mathcal{T}^{0,1}$ yields $(z^0,z^1)\in\mathbb{H}^2_{\mathrm{BMO}}$. $\square$

\begin{lem}~\\\label{lem:driver-linearisation}
  Let Assumption \ref{asm:market} be in force, and take
  \begin{equation}
    (z^0,z^1,u),(\widehat z^0,\widehat z^1,\widehat u)
    \in
    \mathbb{H}^2_{\mathrm{BMO}}(\mathbb{P}^{0,1},\mathbb{F}^{0,1},\mathbb{R}^{1\times d_0})
    \times\mathbb{H}^2_{\mathrm{BMO}}(\mathbb{P}^{0,1},\mathbb{F}^{0,1},\mathbb{R}^{1\times d})
    \times\mathbb{L}^2_{\mathrm{BMO}}(M;\mathbb{P}^{0,1},\mathbb{F}^{0,1},\mathbb{R})
  \end{equation}
  with $|u|,|\widehat u|\leq K$ on $\{\lambda>0\}$ for some constant $K>0$.
  Define $\varrho_s:=z^{0\|}_s\eta_s-u_s+\theta_s^{\top}\eta_s$ and $\widehat\varrho_s:=\widehat z^{0\|}_s\eta_s-\widehat u_s+\theta_s^{\top}\eta_s$ for $s\in[0,T]$.
  Write $\Delta z^0:=z^0-\widehat z^0$, $\Delta z^1:=z^1-\widehat z^1$ and $\Delta u:=u-\widehat u$.
  Set $\varrho^{\ell}_s:=(1-\ell)\varrho_s+\ell\widehat\varrho_s$.
  Define
  \begin{equation}
    A_t:=\int_0^1\lambda_te^{-Q_t(\varrho^\ell_t)}d\ell,~~~B_t:=\int_0^1e^{-Q_t(\varrho^\ell_t)}d\ell-1,~~~t\in[0,T].
  \end{equation}
  Then $B_t>-1$, $dt\otimes\mathbb{P}^{0,1}$-a.e., $A\in\mathbb{H}^2_{\mathrm{BMO}}$, and $B\in\mathbb{L}^2_{\mathrm{BMO}}(M)$.  
  Moreover,
  \begin{equation}\label{eq:driver-linearisation}
    \begin{split}
    &g_t(z^0_t,z^1_t,u_t)-g_t(\widehat z^0_t,\widehat z^1_t,\widehat u_t)\\
    &=\frac{1}{2}(z^{0\perp}_t+\widehat z^{0\perp}_t)(\Delta z^{0\perp}_t)^{\top}
      +\frac{1}{2}(z^1_t+\widehat z^1_t)(\Delta z^1_t)^{\top}
      -\Delta z^{0\|}_t\theta_t
      -A_t\Delta z^{0\|}_t\eta_t
      +\lambda_tB_t\Delta u_t,~~~dt\otimes\mathbb{P}^{0,1}\text{-a.e.}
    \end{split}
  \end{equation}
\end{lem}

\noindent\textbf{\textit{Proof}}\\
For every $t\in[0,T]$, it is easy to see that
\begin{equation}
  \begin{split}
  &\frac{1}{2}\bigl(|z^{0\perp}_t|^2+|z^1_t|^2\bigr) -\frac{1}{2}\bigl(|\widehat z^{0\perp}_t|^2+|\widehat z^1_t|^2\bigr) -(z^{0\|}_t\theta_t-\widehat{z}^{0\|}_t\theta_t) -(\lambda_tu_t-\lambda_t\widehat u_t)\\
    &=
  \frac{1}{2}(z^{0\perp}_t+\widehat z^{0\perp}_t)(\Delta z^{0\perp}_t)^{\top}
  +\frac{1}{2}(z^1_t+\widehat z^1_t)(\Delta z^1_t)^{\top}
  -\Delta z^{0\|}_t\theta_t
  -\lambda_t\Delta u_t.
  \end{split}
\end{equation}
Moreover, \eqref{eq:qstar-J-derivative} gives
\begin{equation}
  \begin{split}
  J_t(Q_t(\varrho_t))-J_t(Q_t(\widehat\varrho_t))
  &=
  -\Bigl(\int_0^1\lambda_te^{-Q_t(\varrho^\ell_t)}d\ell\Bigr)(\varrho_t-\widehat\varrho_t)\\
  &=
  -A_t(\varrho_t-\widehat\varrho_t)\\
  &=
  -A_t\Delta z^{0\|}_t\eta_t+A_t\Delta u_t,~~~dt\otimes\mathbb{P}^{0,1}\text{-a.e.}
  \end{split}
\end{equation}
Combining these gives \eqref{eq:driver-linearisation}.

Set $\varrho^{*}_t:=|\varrho_t|\vee|\widehat\varrho_t|$.
Applying \eqref{eq:qstar-root-exp-bound} with $\varrho^\ell_t$ in place of $\varrho_t$ gives $\lambda_te^{-Q_t(\varrho^\ell_t)}\leq C(1+\varrho^{*}_t)$.
Thus
\begin{equation}
  0\leq A_t\leq C(1+\varrho^{*}_t),~~~
  A_t^2\leq C(1+(\varrho_t^*)^2),
  ~~~ dt\otimes\mathbb{P}^{0,1}\text{-a.e.}
\end{equation}
Since $e^{-Q_t(\varrho^\ell_t)}>0$, it is clear that $B_t>-1$, $dt\otimes\mathbb{P}^{0,1}$-a.e.
Finally, on $\{\lambda>0\}$, using $\lambda_t\geq\underline\lambda$,
\begin{equation}
  |B_t| \leq 1+\int_0^1 e^{-Q_t(\varrho^\ell_t)}d\ell \leq C(1+\varrho^{*}_t), ~~~dt\otimes\mathbb{P}^{0,1}\text{-a.e.}
\end{equation}
Multiplying by $\lambda_t\leq\overline{\lambda}$ gives
\begin{equation}
  \lambda_tB_t^2\leq C(1+(\varrho_t^*)^2),~~~ dt\otimes\mathbb{P}^{0,1}\text{-a.e.},
\end{equation}
and the same bound is trivial on $\{\lambda=0\}$.  
Since $|u|,|\widehat u|\leq K$ on $\{\lambda>0\}$ and $\eta,\theta$ are bounded, $|\varrho^{*}_s|\leq C(1+|z^0_s|+|\widehat z^0_s|)$ on $\{\lambda>0\}$, while $A=0$ and $\lambda B^2=0$ on $\{\lambda=0\}$.  
Hence $A\in\mathbb{H}^2_{\mathrm{BMO}}$ and $B\in\mathbb{L}^2_{\mathrm{BMO}}(M)$. $\square$

\begin{lem}~\\\label{lem:single-default-exponential}
  Let Assumption \ref{asm:market} hold.  Let $\psi\in\mathbb{L}^2(M;\mathbb{P}^{0,1},\mathbb{F}^{0,1},\mathbb{R})$ admit a predictable representative satisfying $\psi_t>-1$ for all $t\in[0,T]$.
  Then, $\displaystyle\mathcal{E}\bigl(\int_0^\cdot\psi_sdM_s\bigr)$ is a martingale of class-$\mathcal{D}$.
  If, in addition, $\varphi=(\varphi^0,\varphi^1)\in  \mathbb{H}^2_{\mathrm{BMO}}(\mathbb{P}^{0,1},\mathbb{F}^{0,1},\mathbb{R}^{d_0+d})$ and $\psi\in\mathbb{L}^2_{\mathrm{BMO}}(M;\mathbb{P}^{0,1},\mathbb{F}^{0,1},\mathbb{R})$, then the Dol\'eans-Dade exponential
  \begin{equation}
    \mathcal{E}\Bigl(\int_0^\cdot (\varphi^0_sdW^0_s + \varphi^1_sdW^1_s + \psi_sdM_s) \Bigr)_{t\in[0,T]}
  \end{equation}
  is also a martingale of class-$\mathcal{D}$.
\end{lem}

\noindent\textbf{\textit{Proof}}\\
  First set $\displaystyle\mathcal{J}^{\psi}:=\int_0^\cdot \psi_sdM_s$ and $\displaystyle\mathcal{J}^{\psi,n}:=\int_0^\cdot \psi^n_sdM_s$ for $\psi^n:=\psi\land n$ where $n\in\mathbb{N}$.
  Since the Dol\'eans-Dade exponential $\mathcal{E}(\mathcal{J}^{\psi})$ is a non-negative local martingale, it is a supermartingale.
  Moreover, for each $n\in\mathbb{N}$, $\mathcal{E}(\mathcal{J}^{\psi,n})$ is a uniformly bounded local martingale and thus is a martingale of class-$\mathcal{D}$.
  By the standard Dol\'eans exponential formula for a semi-martingale (see, e.g., \cite[Chapter II, Theorem 37]{Protter2005}), we have
  \begin{equation}\label{eq:DD-single-default}
    \mathcal{E}\left(\mathcal{J}^{\psi,n}\right)_T = \exp\left(-\int_0^T\psi^n_s\lambda_sds\right) (1+\psi^n_{\tau})^{\mathbf{1}_{\{\tau\leq T\}}}.
  \end{equation}
  Since $-\psi^n<1$ and $\psi^n\leq\psi$, we have
  \begin{equation}
    \mathcal{E}\left(\mathcal{J}^{\psi,n}\right)_T \leq e^{\overline{\lambda} T}\bigl(1+|\psi_{\tau}|\mathbf{1}_{\{\tau\leq T\}}\bigr) = e^{\overline{\lambda} T}\Bigl(1+\int_0^T |\psi_s| dN_s\Bigr).
  \end{equation}
  Note that
  \begin{equation}
    \mathbb{E}\Bigl[\int_0^T |\psi_s| dN_s\Bigr]
    =
    \mathbb{E}\Bigl[\int_0^T\lambda_s|\psi_s|ds\Bigr]
    \leq
    \overline{\lambda} T+\mathbb{E}\Bigl[\int_0^T\lambda_s\psi_s^2ds\Bigr]<\infty .
  \end{equation}
  Therefore, dominated convergence yields
  \begin{equation}
    \mathbb{E}[\mathcal{E}(\mathcal{J}^{\psi})_T]=\lim_{n\to\infty}\mathbb{E}[\mathcal{E}(\mathcal{J}^{\psi,n})_T]=1.
  \end{equation}
  Together with the supermartingale property of $\mathcal{E}(\mathcal{J}^{\psi})$, we deduce that it is a martingale of class-$\mathcal{D}$.

  For the second assertion, set $\displaystyle\mathcal{M}^{\varphi}:=\int_0^\cdot \varphi^0_sdW^0_s+\int_0^\cdot\varphi^1_sdW^1_s$.
  Since $\varphi^i\in\mathbb{H}^2_{\mathrm{BMO}}$ ($i=0,1$), $\mathcal{E}(\mathcal{M}^{\varphi})$ is a martingale of class-$\mathcal{D}$ (see, e.g., \cite[Theorem 2.3]{Kazamaki1994}).
  Notice that $\mathcal{E}(\mathcal{M}^{\varphi}+\mathcal{J}^{\psi})$ is a supermartingale and that the quadratic covariation satisfies $\langle \mathcal{M}^{\varphi},\mathcal{J}^{\psi}\rangle =0$, as $\mathcal{M}^{\varphi}$ is continuous and $\mathcal{J}^{\psi}$ is purely discontinuous.
  This implies $\mathcal{E}(\mathcal{M}^{\varphi}+\mathcal{J}^{\psi})=\mathcal{E}(\mathcal{M}^{\varphi})\mathcal{E}(\mathcal{J}^{\psi})$.

  Define a new measure $\mathbb{Q}^\varphi$ by $\dfrac{d\mathbb{Q}^{\varphi}}{d\mathbb{P}^{0,1}}:=\mathcal{E}(\mathcal{M}^{\varphi})_T$.
  Since $\mathcal{E}(\mathcal{M}^{\varphi})$ and $M$ are both $\mathbb{P}^{0,1}$-martingales and their quadratic covariation is zero, It\^o formula shows that $\mathcal{E}(\mathcal{M}^{\varphi})M$ is a $\mathbb{P}^{0,1}$-local martingale.
  In particular, as $M$ is uniformly bounded ($|M_t|\leq 1\lor \overline{\lambda} T$ for $t\in[0,T]$), $\mathcal{E}(\mathcal{M}^{\varphi})M$ is a $\mathbb{P}^{0,1}$-martingale.
  Bayes' formula shows that $M$ is a $\mathbb{Q}^\varphi$-martingale.

  Write $K:=\|\psi\|^2_{\mathbb{L}^2_{\mathrm{BMO}}(M)}$.
  Choose an integer $q\geq2$ large enough such that the reverse H\"older inequality (see, e.g., \cite[Theorem 3.1]{Kazamaki1994}) for $\mathcal{E}(\mathcal{M}^{\varphi})$ holds with the conjugate exponent $p:=q/(q-1)>1$. 
  For any stopping time $\tau_0\in\mathcal{T}(\mathbb{F}^{0,1})$, Bayes' formula and the conditional H\"older inequality give
  \begin{equation}
    \begin{split}
      \mathbb{E}^{\mathbb{Q}^{\varphi}}\Bigl[\int_{\tau_0}^T \lambda_s\psi_s^2ds|\mathcal{F}_{\tau_0}\Bigr]
      &=
      \mathbb{E}\Bigl[\frac{\mathcal{E}(\mathcal{M}^{\varphi})_T}{\mathcal{E}(\mathcal{M}^{\varphi})_{\tau_0}}\int_{\tau_0}^T \lambda_s\psi_s^2ds|\mathcal{F}_{\tau_0}\Bigr]\\
      &\leq
      \mathbb{E}\Bigl[\frac{\mathcal{E}(\mathcal{M}^{\varphi})_T^p}{\mathcal{E}(\mathcal{M}^{\varphi})_{\tau_0}^p}\mid\mathcal{F}_{\tau_0}\Bigr]^{1/p}\mathbb{E}\Bigl[\Bigl(\int_{\tau_0}^T \lambda_s\psi_s^2ds\Bigr)^q\mid\mathcal{F}_{\tau_0}\Bigr]^{1/q}.
    \end{split}
  \end{equation}
  The reverse H\"older estimate yields
  \begin{equation}
    \mathbb{E}\Bigl[\frac{\mathcal{E}(\mathcal{M}^{\varphi})_T^p}{\mathcal{E}(\mathcal{M}^{\varphi})_{\tau_0}^p}\mid\mathcal{F}_{\tau_0}\Bigr]\leq C_p
  \end{equation}
  uniformly in $\tau_0$. The energy inequality (see, e.g., the proof of \cite[Theorem 7.2.3, Eq.~(7.2.6)]{ZhangBSDE}) gives
  \begin{equation}
   \mathbb{E}\Bigl[\Bigl(\int_{\tau_0}^T \lambda_s\psi_s^2ds\Bigr)^q\mid\mathcal{F}_{\tau_0}\Bigr]\leq q!K^q.
  \end{equation}
  Hence, taking supremum over $\tau_0\in\mathcal{T}(\mathbb{F}^{0,1})$, we obtain
  \begin{equation}
    \sup_{\tau_0\in\mathcal{T}(\mathbb{F}^{0,1})}
    \Bigl\|\mathbb{E}^{\mathbb{Q}^{\varphi}}\Bigl[\int_{\tau_0}^T\lambda_s\psi_s^2ds
    \Bigm|\mathcal{F}^{0,1}_{\tau_0}\Bigr]\Bigr\|_{\mathbb{L}^{\infty}}
    \leq C_p^{1/p}(q!)^{1/q}K<\infty,
  \end{equation}
  and in particular, $\displaystyle\mathbb{E}^{\mathbb{Q}^{\varphi}}\Bigl[\int_0^T\lambda_s\psi_s^2ds\Bigr]<\infty$.

  Applying the first result under $\mathbb{Q}^{\varphi}$ gives
  \begin{equation}
    \mathbb{E}[\mathcal{E}(\mathcal{M}^{\varphi}+\mathcal{J}^{\psi})_T]
    =
    \mathbb{E}^{\mathbb{Q}^{\varphi}}[\mathcal{E}(\mathcal{J}^{\psi})_T]
    =1.
  \end{equation}
  Again, together with the supermartingale property, we deduce that $\mathcal{E}(\mathcal{M}^{\varphi}+\mathcal{J}^{\psi})$ is a martingale of class-$\mathcal{D}$. $\square$

\begin{thm}[Well-posedness of the BSDE]~\\\label{thm:BSDE-wellposed}
  Let Assumptions \ref{asm:market} and \ref{asm:agent} hold.
  Then, the BSDE \eqref{eq:BSDE-y} admits a unique bounded solution
  \begin{equation}
    (y^1,z^{1,0},z^1,u^1)\in\mathbb{S}^{\infty}(\mathbb{P}^{0,1},\mathbb{F}^{0,1},\mathbb{R})\times\mathbb{H}^2_{\mathrm{BMO}}(\mathbb{P}^{0,1},\mathbb{F}^{0,1},\mathbb{R}^{1\times d_0}) \times\mathbb{H}^2_{\mathrm{BMO}}(\mathbb{P}^{0,1},\mathbb{F}^{0,1},\mathbb{R}^{1\times d})\times\mathbb{L}^2_{\mathrm{BMO}}(M;\mathbb{P}^{0,1},\mathbb{F}^{0,1},\mathbb{R}),
  \end{equation}
  and it satisfies $|u^1_t|\leq 2\|y^1\|_{\mathbb{S}^{\infty}}$, $dt\otimes\mathbb{P}^{0,1}$-a.e. on $\{\lambda>0\}$.
\end{thm}

\noindent\textbf{\textit{Proof}}\\
As usual, we omit the superscript ``1'' when it is obvious for notational simplicity.\\
\noindent
(1) Uniqueness\\
We first show the uniqueness of the solution. Let $(y,z^0,z^1,u),(\widehat y,\widehat z^0,\widehat z^1,\widehat u)\in\mathbb{S}^{\infty}\times\mathbb{H}^2_{\mathrm{BMO}}\times\mathbb{H}^2_{\mathrm{BMO}}\times\mathbb{L}^2_{\mathrm{BMO}}(M)$ be two bounded solutions.  
We take the canonical predictable representatives $u=\widehat u=0$ on $\{\lambda=0\}$. 
By Lemma \ref{lem:Sbounded-implies-BMO}, $|u|\leq 2\|y\|_{\mathbb{S}^{\infty}}$ and $|\widehat u|\leq 2\|\widehat y\|_{\mathbb{S}^{\infty}}$ on $\{\lambda>0\}$.
Hence $u$ and $\widehat u$ are bounded with these representatives.  
Write
\begin{equation}
  \Delta y:=y-\widehat y,~~~
  \Delta z^0:=z^0-\widehat z^0,~~~
  \Delta z^1:=z^1-\widehat z^1,~~~
  \Delta u:=u-\widehat u.
\end{equation}
By Lemma \ref{lem:driver-linearisation}, there are predictable processes $A\in\mathbb{H}^2_{\mathrm{BMO}}$ and $B\in\mathbb{L}^2_{\mathrm{BMO}}(M)$ with $B>-1$ such that
\begin{equation}
  g_t(z^0_t,z^1_t,u_t)-g_t(\widehat z^0_t,\widehat z^1_t,\widehat u_t)
  =
  \alpha^0_t(\Delta z^0_t)^\top
  +\alpha^1_t(\Delta z^1_t)^\top
  +\lambda_tB_t\Delta u_t,
\end{equation}
$dt\otimes\mathbb{P}^{0,1}$-a.e., where $\alpha^1_t:=\dfrac{1}{2}(z^1_t+\widehat z^1_t)$ and $\alpha^0_t:=\dfrac{1}{2}(z^{0\perp}_t+\widehat z^{0\perp}_t)-A_t\eta_t^\top-\theta_t^\top$ for $t\in[0,T]$.
By the estimate for $A$ obtained in Lemma \ref{lem:driver-linearisation}, we deduce that $(\alpha^0,\alpha^1)\in\mathbb{H}^2_{\mathrm{BMO}}$.

Now, we have
\begin{equation}
\begin{split}
  \Delta y_t
  &=
  \int_t^T\left(\alpha^0_s(\Delta z^0_s)^\top+\alpha^1_s(\Delta z^1_s)^\top +\lambda_sB_s\Delta u_s \right)ds
  -\int_t^T\Delta z^0_sdW^0_s -\int_t^T\Delta z^1_sdW^1_s -\int_t^T\Delta u_sdM_s .
\end{split}
\end{equation}
Let $\mathbb{Q}\sim\mathbb{P}^{0,1}$ be the probability measure defined by
\begin{equation}
  \left.\frac{d\mathbb{Q}}{d\mathbb{P}^{0,1}}\right|_{\mathcal{F}^{0,1}_t}:=\mathcal{E}\left(\int_0^\cdot\alpha^0_sdW^0_s +\int_0^\cdot\alpha^1_sdW^1_s+\int_0^\cdot B_sdM_s\right)_t,~~~t\in[0,T],
\end{equation}
which is well-defined by Lemma \ref{lem:single-default-exponential}.
Then, by Girsanov's theorem,
\begin{equation}
  W^{0,\mathbb Q}_t:=W^0_t-\int_0^t(\alpha^0_s)^\top ds,~~~  W^{1,\mathbb Q}_t:=W^1_t-\int_0^t(\alpha^1_s)^\top ds,~~~t\in[0,T],
\end{equation}
are Brownian motions under $\mathbb Q$, and
\begin{equation}
  M^{\mathbb Q}_t:=M_t-\int_0^t\lambda_sB_sds,~~~t\in[0,T],
\end{equation}
is the compensated default martingale under $\mathbb Q$, with intensity $(1+B_t)\lambda_t$.  Substituting these processes, we obtain
\begin{equation}
  \Delta y_t  =  -\int_t^T\Delta z^0_sdW^{0,\mathbb Q}_s -\int_t^T\Delta z^1_sdW^{1,\mathbb Q}_s -\int_t^T\Delta u_sdM^{\mathbb Q}_s.
\end{equation}
Since $\Delta y$ is bounded, the stochastic integrals above are true $\mathbb{Q}$-martingales. Taking conditional expectation, we have $\Delta y_t=0$, $t\in[0,T]$.
Applying It\^o formula to $|\Delta y|^2$ on $[t,T]$ under $\mathbb Q$ and conditioning on $\mathcal F^{0,1}_t$ then gives, for every $t\in[0,T]$,
\begin{equation}
 \mathbb E^{\mathbb Q}\Bigl[\int_t^T(|\Delta z^0_s|^2+|\Delta z^1_s|^2)ds +\int_t^T(1+B_s)\lambda_s|\Delta u_s|^2ds \Bigm|\mathcal F^{0,1}_t \Bigr] = 0.
\end{equation}
Since $1+B>0$, we obtain $\Delta z^0=0$, $\Delta z^1=0$ in $\mathbb H^2_{\mathrm{BMO}}$ and $\Delta u=0$ in $\mathbb L^2_{\mathrm{BMO}}(M)$.

\noindent
(2) Existence\\
Set $a_t:=\theta_t^\top+\lambda_t\eta_t^\top$ for $t\in[0,T]$ and $\phi(q):=e^{-q}-1+q$ for $q\in\mathbb{R}$.
Multiplying \eqref{eq:driver-inf} by $\gamma$ and substituting $(Z^{0}_s,Z^1_s,U_s)=(z^0/\gamma,z^1/\gamma,u/\gamma)$, the driver $g_t$ satisfies
\begin{equation}
\begin{split}
  g_t(z^0,z^1,u)
  &=\inf_{p\in L_t}\Bigl\{-\gamma^1p\theta_t+\frac{1}{2}|\gamma^1p-z^0|^2+\frac{1}{2}|z^1|^2-\lambda_tu+\lambda_t\bigl(e^{-(\gamma^1p\eta_t-u)}-1\bigr)\Bigr\}\\
  &=\inf_{p\in L_t}\Bigl\{-\gamma^1pa_t^\top+\frac{1}{2}|\gamma^1p-z^0|^2+\frac{1}{2}|z^1|^2+\lambda_t\phi(\gamma^1p\eta_t-u)\Bigr\}\\
  &=-z^0a_t^\top-\frac{1}{2}|a_t|^2+\inf_{p\in L_t}\Bigl\{\frac{1}{2}|\gamma^1p-z^0-a_t|^2+\frac{1}{2}|z^1|^2+\lambda_t\phi(\gamma^1p\eta_t-u)\Bigr\}.
\end{split}
\end{equation}
For $m\geq1$, set 
\begin{equation}
  L^m_t:=\{p\in L_t;|(\sigma_t\sigma^\top_t)^{-1}\sigma_t p^\top|\leq m\},
\end{equation}
which is convex and compact, and define
\begin{equation}\label{eq:app-gm-square}
  g^m_t(z^0,z^1,u):=-z^0a_t^\top-\frac{1}{2}|a_t|^2 +\inf_{p\in L^m_t}\Bigl\{\frac{1}{2}|\gamma^1p-z^0-a_t|^2+\frac{1}{2}|z^1|^2  +\lambda_t\phi(\gamma^1p\eta_t-u)\Bigr\}.
\end{equation}
Equivalently, we have
\begin{equation}
  g^m_t(z^0,z^1,u)=-z^0a_t^\top-\frac{1}{2}|a_t|^2 +\inf_{|\pi|\leq m}\Bigl\{\frac{1}{2}|\gamma^1\pi^\top\sigma_t-z^0-a_t|^2+\frac{1}{2}|z^1|^2  +\lambda_t\phi(\gamma^1\pi^\top\beta_t-u)\Bigr\}.
\end{equation}
By compactness and strict convexity, there exists a unique $p^{m,*}_t\in L^m_t$ for each $t\in[0,T]$ such that
\begin{equation}
  \frac{1}{2}|\gamma^1p^{m,*}_t-z^0-a_t|^2+\frac{1}{2}|z^1|^2  +\lambda_t\phi(\gamma^1p^{m,*}_t\eta_t-u)=\inf_{p\in L^m_t}\Bigl\{\frac{1}{2}|\gamma^1p-z^0-a_t|^2+\frac{1}{2}|z^1|^2  +\lambda_t\phi(\gamma^1p\eta_t-u)\Bigr\}.
\end{equation}
By the equivalent expression above, the objective is a Carath\'eodory function of $\pi$ (continuous in $\pi$, predictable in $(t,\omega)$) on the fixed compact ball $\{|\pi|\leq m\}$.
Hence, for $z^0,z^1,u\in\mathbb{L}^0_{\mathrm{pred}}$, $g^m_t(z^0_t,z^1_t,u_t)$ and the minimizer $\pi^{m,*}_t$ are predictable by the measurable maximum theorem (see, e.g., \cite[Theorem 18.19]{AliprantisBorder2006}), and so is $p^{m,*}_t=(\pi^{m,*}_t)^{\top}\sigma_t$.

Taking $p=0$ (as $0\in L^m_t$) gives
\begin{equation}\label{eq:app-H1}
  -z^0a_t^\top-\frac{1}{2}|a_t|^2 \leq g^m_t(z^0,z^1,u) \leq \frac{1}{2}\bigl(|z^0|^2+|z^1|^2\bigr)+\lambda_t(e^u-1-u)
\end{equation}
uniformly in $m\geq 1$. Compactness of $L^m_t$ gives the local Lipschitz property:
\begin{equation}
  |g^m_t(z^0,z^1,u)-g^m_t(\acute z^0,\acute z^1,u)|
   \leq C_m\bigl(1+|z^0|+|\acute z^0|+|z^1|+|\acute z^1|\bigr)\bigl(|z^0-\acute z^0|+|z^1-\acute z^1|\bigr).
\end{equation}
Moreover, for $u\neq u'$ define
\begin{equation}
  b^m_t(u,u'):=\frac{g^m_t(z^0,z^1,u)-g^m_t(z^0,z^1,u')}{\lambda_t(u-u')}\mathbf{1}_{\{\lambda_t>0\}},
\end{equation}
with the convention $b^m_t(u,u):=0$, so that $g^m_t(z^0,z^1,u)-g^m_t(z^0,z^1,u')=\lambda_tb^m_t(u,u')(u-u')$ for $t\in[0,T]$.
Note that
\begin{equation}
  g^m_t(z^0,z^1,u)+\lambda_t u=-z^0a_t^\top-\frac{1}{2}|a_t|^2+\inf_{p\in L^m_t}\Bigl\{\frac{1}{2}|\gamma^1p-z^0-a_t|^2+\frac{1}{2}|z^1|^2+\lambda_t\bigl(e^{-\gamma^1p\eta_t+u}-1+\gamma^1p\eta_t\bigr)\Bigr\}.
\end{equation}
Since $u\mapsto\lambda_te^{-\gamma^1p\eta_t+u}$ is strictly increasing if $\lambda_t>0$, so is $g^m_t(z^0,z^1,u)+\lambda_t u$  (the infimum being attained on the compact set $L^m_t$).
This observation yields that, for $u>u'$, we have
\begin{equation}
  \lambda_tb^m_t(u,u')(u-u')=g^m_t(z^0,z^1,u)-g^m_t(z^0,z^1,u')>-\lambda_t(u-u'),
\end{equation}
i.e., $b^m_t(u,u')>-1$. We have the same result for $u<u'$.
In addition, the local Lipschitz continuity of $\phi$ gives, for $|u|,|u'|\leq K$,
\begin{equation}
  \begin{split}
  &|g^m_t(z^0,z^1,u)-g^m_t(z^0,z^1,u')|\\
  &\leq
  \Bigl|\inf_{p\in L^m_t}\Bigl\{\frac{1}{2}|\gamma^1p-z^0-a_t|^2+\frac{1}{2}|z^1|^2  +\lambda_t\phi(\gamma^1p\eta_t-u)\Bigr\} - \inf_{p\in L^m_t}\Bigl\{\frac{1}{2}|\gamma^1p-z^0-a_t|^2+\frac{1}{2}|z^1|^2  +\lambda_t\phi(\gamma^1p\eta_t-u')\Bigr\}\Bigr|\\
  &\leq
  \sup_{p\in L^m_t}\lambda_t\bigl|\phi(\gamma^1p\eta_t-u)-\phi(\gamma^1p\eta_t-u')\bigr|\\
  &\leq
  C_{m,K}\lambda_t|u-u'|,
  \end{split}
\end{equation}
and thus $|b^m_t(u,u')|\leq C_{m,K}$.
Then, for each $m\geq1$, the BSDE
\begin{equation}\label{eq:app-ym}
\begin{split}
  y^m_t=G+\int_t^T g^m_s(z^{0,m}_s,z^{1,m}_s,u^m_s)ds-\int_t^T z^{0,m}_sdW^0_s-\int_t^T z^{1,m}_sdW^1_s-\int_t^T u^m_sdM_s,~~~t\in[0,T],
\end{split}
\end{equation}
has a bounded solution $(y^m,z^{0,m},z^{1,m},u^m)\in\mathbb S^\infty\times\mathbb H^2\times\mathbb H^2\times\mathbb L^2(M)$ by Morlais \cite[Theorem 2]{Morlais2009}.
Strictly speaking, \cite[Theorem 2]{Morlais2009} is proved under a deterministic compensator $n(dx)dt$, whereas the compensator here is $\lambda_t\delta_1(dx)dt$ with a bounded $\mathbb{F}^0$-predictable intensity.
Nevertheless, by \eqref{eq:app-H1} and the estimates established above, the analogues of conditions (H1) and (H2) of \cite[Section 3.1]{Morlais2009} hold with constants independent of $(t,\omega)$.
Accordingly, the existence argument of \cite{Morlais2009} carries over with $n(dx)dt$ replaced by $\lambda_t\delta_1(dx)dt$.

By \textit{a priori} estimates \cite[Lemma 3 and Corollary~1]{Morlais2009}, we have
\begin{equation}\label{eq:app-K0}
  -\|G\|_{\mathbb L^\infty}-\frac{T}{2}\|a\|_{\mathbb L^\infty}^2\leq y^m_t\leq\|G\|_{\mathbb L^\infty},~~~\mathrm{uniformly~~in~~}m\geq 1,
\end{equation}
together with $|u^m_t|\leq2\|y^m\|_{\mathbb{S}^\infty}$ on $\{\lambda>0\}$ and the uniform bounds $(z^{0,m},z^{1,m})\in\mathbb H^2_{\mathrm{BMO}}$, $u^m\in\mathbb L^2_{\mathrm{BMO}}(M)$.

Now, $L^{m+1}_t\supseteq L^m_t$ gives $g^{m+1}\leq g^m$ for each $m\geq1$.
Since $\bigcup_{m\geq1}L^m_t=L_t$, $g^m\searrow g$ pointwise.
Moreover, the convergence is locally uniform in $(z^0,z^1,u)$.
Indeed, by strict convexity, let $p^{*}_t$ be the unique minimizer defined by
\begin{equation}
  \frac{1}{2}|\gamma^1p^{*}_t-z^0-a_t|^2+\frac{1}{2}|z^1|^2  +\lambda_t\phi(\gamma^1p^{*}_t\eta_t-u)=\inf_{p\in L_t}\Bigl\{\frac{1}{2}|\gamma^1p-z^0-a_t|^2+\frac{1}{2}|z^1|^2  +\lambda_t\phi(\gamma^1p\eta_t-u)\Bigr\}.
\end{equation}
If $|z^0|\vee|z^1|\vee|u|\leq\rho$, then comparison with $p=0$ and $\phi\geq0$ gives
\begin{equation}
  \frac{1}{2}|\gamma^1p^*_t-z^0-a_t|^2
  \leq \frac{1}{2}|\gamma^1p^{*}_t-z^0-a_t|^2+\frac{1}{2}|z^1|^2  +\lambda_t\phi(\gamma^1p^{*}_t\eta_t-u)
  \leq \frac{1}{2}|z^0+a_t|^2+\frac{1}{2}|z^1|^2+\lambda_t\phi(-u)
  \leq C_\rho.
\end{equation}
Hence there exists a constant $m_\rho>0$, depending only on $\rho$, $\underline\gamma$, $c_\sigma$, $\|a\|_{\mathbb L^\infty}$ and $\overline\lambda$, such that $|(\sigma_t\sigma_t^\top)^{-1}\sigma_t(p^*_t)^\top|\leq c_\sigma^{-\frac{1}{2}}|p^*_t|\leq m_\rho$.
This yields $p^{*}_t=p^{m,*}_t$ and thus $g^m_t=g_t$ for every $m\geq m_\rho$ whenever $|z^0|\vee|z^1|\vee|u|\leq\rho$, uniformly in $(t,\omega)$.

For $m\geq 1$, set $\Delta y^m:=y^{m+1}-y^m$ (so $\Delta y^m_T=0$), $\Delta z^{0,m}:=z^{0,m+1}-z^{0,m}$, $\Delta z^{1,m}:=z^{1,m+1}-z^{1,m}$ and $\Delta u^m:=u^{m+1}-u^m$.
Then, by the local Lipschitz property in $(z^0,z^1)$ and the kernel $b^{m+1}_s$ above in $u$, we have
\begin{equation}
  \begin{split}
    &g^{m+1}_s(z^{0,m+1}_s,z^{1,m+1}_s,u^{m+1}_s)-g^m_s(z^{0,m}_s,z^{1,m}_s,u^m_s)\\
    &=
    g^{m+1}_s(z^{0,m+1}_s,z^{1,m+1}_s,u^{m+1}_s)-g^{m+1}_s(z^{0,m}_s,z^{1,m}_s,u^m_s)-g^m_s(z^{0,m}_s,z^{1,m}_s,u^m_s)+g^{m+1}_s(z^{0,m}_s,z^{1,m}_s,u^m_s)\\
    &=
    \alpha^{0,m}_s(\Delta z^{0,m}_s)^\top+\alpha^{1,m}_s(\Delta z^{1,m}_s)^\top+\lambda_sb^{m+1}_s(u^{m+1}_s,u^m_s)\Delta u^m_s - r^m_s
  \end{split}
\end{equation}
where $\alpha^{0,m}$ and $\alpha^{1,m}$ are processes satisfying
\begin{equation}
  |\alpha^{0,m}_s| + |\alpha^{1,m}_s| \leq C_m\bigl(1+|z^{0,m}_s|+|z^{0,m+1}_s|+|z^{1,m}_s|+|z^{1,m+1}_s|\bigr),
\end{equation}
and $r^m_s:=g^m_s(z^{0,m}_s,z^{1,m}_s,u^m_s)-g^{m+1}_s(z^{0,m}_s,z^{1,m}_s,u^m_s)$ for $s\in[0,T]$.
Then, notice that $(\alpha^{0,m},\alpha^{1,m})\in\mathbb H^2_{\mathrm{BMO}}$, $b^{m+1}_s(u^{m+1}_s,u^m_s)>-1$, $b^{m+1}_\cdot(u^{m+1}_\cdot,u^m_\cdot)\in\mathbb L^2_{\mathrm{BMO}}(M)$, and $r^m_s\geq 0$ for each $s\in[0,T]$.
Hence
\begin{equation}
  \begin{split}
  \Delta y^m_t
  &=
  \int_t^T\bigl(\alpha^{0,m}_s(\Delta z^{0,m}_s)^\top+\alpha^{1,m}_s(\Delta z^{1,m}_s)^\top+\lambda_sb^{m+1}_s(u^{m+1}_s,u^m_s)\Delta u^m_s-r^m_s\bigr)ds\\
  &~~~~-\int_t^T\Delta z^{0,m}_sdW^0_s-\int_t^T\Delta z^{1,m}_sdW^1_s-\int_t^T\Delta u^m_sdM_s,~~~t\in[0,T].
  \end{split}
\end{equation}
Similarly to (1), define a new probability measure $\mathbb{Q}^m$ by
\begin{equation}
  \left.\frac{d\mathbb{Q}^m}{d\mathbb{P}^{0,1}}\right|_{\mathcal{F}^{0,1}_t}:=\mathcal{E}\left(\int_0^\cdot\alpha^{0,m}_sdW^0_s +\int_0^\cdot\alpha^{1,m}_sdW^1_s+\int_0^\cdot b^{m+1}_s(u^{m+1}_s,u^m_s)dM_s\right)_t,~~~t\in[0,T],
\end{equation}
which is well-defined by Lemma \ref{lem:single-default-exponential}.
Then, since $\Delta y^m$ is bounded and $r^m_s\geq0$, we obtain
\begin{equation}
  \Delta y^m_t=-\mathbb E^{\mathbb{Q}^m}\Bigl[\int_t^Tr^m_sds\Bigm|\mathcal F^{0,1}_t\Bigr]\leq0,~~~t\in[0,T].
\end{equation}
Thus $y^{m+1}\leq y^m$ for all $m\geq 1$.  Since $(y^m)_{m\geq 1}$ is uniformly bounded, the limit $\displaystyle y_t:=\lim_{m\to\infty}y^m_t$ exists and inherits the bound
\begin{equation}
  -\|G\|_{\mathbb L^\infty}-\frac{T}{2}\|a\|_{\mathbb L^\infty}^2\leq y_t\leq\|G\|_{\mathbb L^\infty}.
\end{equation}
The monotone stability of \cite[Lemma 7]{Morlais2009}, which goes back to Kobylanski \cite{Kobylanski2000} in the continuous setting, carries over here since $g^m\searrow g$ locally uniformly and the bounds in \eqref{eq:app-H1} hold uniformly in $m$.
Consequently, $(z^{0,m},z^{1,m},u^m)$ converges in $\mathbb H^2\times\mathbb H^2\times\mathbb L^2(M)$ to some $(z^0,z^1,u)$, and $(y,z^0,z^1,u)$ solves the BSDE \eqref{eq:BSDE-y}.
Applying the Burkholder-Davis-Gundy inequality to $y^m-y$ and using the convergence of the integrands in $\mathbb H^2\times\mathbb H^2\times\mathbb L^2(M)$ and of the drivers in $L^1$, we obtain $\displaystyle\mathbb E\Bigl[\sup_{t\in[0,T]}|y^m_t-y_t|\Bigr]\to0$.
Along an appropriate subsequence, we have $\displaystyle\sup_{t\in[0,T]}|y^m_t-y_t|\to0$, $\mathbb P^{0,1}$-a.s., so $y$ is c\`adl\`ag, and thus $y\in\mathbb S^\infty$ by \eqref{eq:app-K0}.
Finally, Lemma \ref{lem:Sbounded-implies-BMO} gives $(z^0,z^1)\in\mathbb H^2_{\mathrm{BMO}}$, $u\in\mathbb L^2_{\mathrm{BMO}}(M)$ and $|u|\leq2\|y\|_{\mathbb S^\infty}$ on $\{\lambda>0\}$. $\square$

\begin{rem}\label{rem:BSDE-Y}~\\
  (1) Under Assumptions \ref{asm:market} and \ref{asm:agent}, the original BSDE \eqref{eq:BSDE-Y} has the corresponding unique bounded solution $(Y^1,Z^{1,0},Z^1,U^1)=(y^1/\gamma^1,z^{1,0}/\gamma^1,z^1/\gamma^1,u^1/\gamma^1)$.\\
  (2) The uniqueness proof of Theorem \ref{thm:BSDE-wellposed} uses Lemma \ref{lem:single-default-exponential} to justify the change of measure.
  The proof of the lemma relies on the single-jump identity \eqref{eq:DD-single-default}.
  Consequently, this uniqueness argument does not directly extend to general multiple-jump models, such as those driven by a L\'evy random measure.
  The uniqueness in the bounded class established here will be used to show asymptotic market clearing in Section \ref{sec:clearing}.
\end{rem}

\begin{thm}[Verification]~\\\label{thm:verify}
  Let Assumptions \ref{asm:market} and \ref{asm:agent} be in force, and let $(Y^1,Z^{1,0},Z^1,U^1)\in\mathbb{S}^{\infty}(\mathbb{P}^{0,1},\mathbb{F}^{0,1},\mathbb{R})\times\mathbb{H}^2_{\mathrm{BMO}}(\mathbb{P}^{0,1},\mathbb{F}^{0,1},\mathbb{R}^{1\times d_0}) \times\mathbb{H}^2_{\mathrm{BMO}}(\mathbb{P}^{0,1},\mathbb{F}^{0,1},\mathbb{R}^{1\times d})\times\mathbb{L}^2_{\mathrm{BMO}}(M;\mathbb{P}^{0,1},\mathbb{F}^{0,1},\mathbb{R})$ be the solution of the BSDE \eqref{eq:BSDE-Y}.
  We fix the canonical predictable representative $U^1=0$ on $\{\lambda=0\}$.
  Then, the process $p^{1,*}$, defined by
  \begin{equation}
    p^{1,*}_t  :=  Z^{1,0\|}_t+\frac{\theta_t^{\top}}{\gamma^1} + \frac{\lambda_t}{\gamma^1}\eta_t^{\top}e^{-q^{1,*}_t},~~~t\in[0,T],
  \end{equation}
  where $q^{1,*}$ is the unique predictable process satisfying
  \begin{equation}
    q^{1,*}_t-\lambda_t|\eta_t|^2e^{-q^{1,*}_t}=\gamma^1Z^{1,0\|}_t\eta_t-\gamma^1U^1_t+\theta_t^{\top}\eta_t,~~~t\in[0,T],
  \end{equation}
  is the unique optimal control for agent-1.
\end{thm}
\noindent\textbf{\textit{Proof}}\\
For notational simplicity, we omit the superscript ``1'' in the proof when it is obvious.
We first prove the admissibility of $p^*$, namely $p^*\in\mathcal A^1$.  Set $z^0:=\gamma Z^0$, $u:=\gamma U$ and
\begin{equation}
  \varrho:=z^{0\|}\eta-u+\theta^\top\eta.
\end{equation}
Let $q^*$ be the predictable process satisfying
\begin{equation}
  q^*-\lambda|\eta|^2e^{-q^*}=\varrho,
\end{equation}
given by Lemma \ref{lem:qstar-apriori}.
Then $p^*$ is predictable and $q^*=\gamma(p^*\eta-U)$, so this definition is equivalent to \eqref{eq:FOC}.
By Lemma \ref{lem:Sbounded-implies-BMO}, with the representative $u=0$ on $\{\lambda=0\}$, $u$ is bounded and $z^0\in\mathbb H^2_{\mathrm{BMO}}$, hence $\varrho\in\mathbb H^2_{\mathrm{BMO}}$.
The estimate \eqref{eq:qstar-root-exp-bound} yields
\begin{equation}
  \lambda e^{-q^*} \leq C(1+|\varrho|),
\end{equation}
so $p^*\in\mathbb H^2$.  Since $\eta$ and $\lambda$ are bounded, $\displaystyle \mathbb E\int_0^T\lambda_s|p^*_s\eta_s|^2ds<\infty$.
Thus Definition \ref{def:admissible} (i) holds for $\pi^*:=(\sigma\sigma^\top)^{-1}\sigma(p^*)^\top$.

For Definition \ref{def:admissible} (ii), note that we have
\begin{equation}
  R^{p^*}_t=R_0\mathcal E(L^*)_t,~~~L^*_t = \int_0^t\varphi^{*,0}_sdW^0_s +\int_0^t\varphi^{*,1}_sdW^1_s +\int_0^t\varphi^{*,J}_sdM_s,~~~t\in[0,T],
\end{equation}
where $R_0=-e^{-\gamma(\xi-Y_0)}$ and
\begin{equation}
  \varphi^{*,0}_t:=-\gamma(p^*_t-Z^0_t) =-\theta^\top_t-\lambda_t\eta^\top_t e^{-q^*_t}+\gamma Z_t^{0\perp},
  ~~~
  \varphi^{*,1}_t:=\gamma Z^1_t,
  ~~~
  \varphi^{*,J}_t:=e^{-q^*_t}-1,~~~t\in[0,T].
\end{equation}
Then, $\varphi^{*,0},\varphi^{*,1}\in\mathbb H^2_{\mathrm{BMO}}$.
Also, $\varphi^{*,J}>-1$, and \eqref{eq:qstar-root-exp-bound} gives $\lambda|\varphi^{*,J}|^2\leq C(1+|\varrho|^2)$. Hence $\varphi^{*,J}\in\mathbb L^2_{\mathrm{BMO}}(M)$.
Lemma \ref{lem:single-default-exponential} therefore implies that $\mathcal E(L^*)$ is a martingale of class-$\mathcal{D}$.
Since $R_0$, $Y$, and $\gamma$ are bounded, $\{\exp(-\gamma\mathcal W^{p^*}_\tau); \tau\in\mathcal T^{0,1}\}$ is uniformly integrable.  Hence $p^*\in\mathcal A^1$.

We now verify Condition-R.  Conditions~(i) and~(ii) follow from the terminal condition of \eqref{eq:BSDE-Y} and the identity $R^p_0=-e^{-\gamma(\xi-Y_0)}$, which is independent of $p$.
For any $p\in\mathcal A^1$, the drift term in \eqref{eq:drift-Rpi} is nonpositive, so $R^{p}$ is a local supermartingale. Together with the uniform integrability (by the definition of admissibility and the boundedness of $Y$), this yields the supermartingale property of $R^p$.
For $p^*$, the drift of $R^{p^*}$ is zero and $p^*\in\mathcal A^1$ makes $R^{p^*}$ a martingale of class-$\mathcal{D}$.
Therefore, Remark \ref{rem:condition-R-verification} gives $\widetilde{\mathcal U}^1(p)\leq\widetilde{\mathcal U}^1(p^*)$ for every $p\in\mathcal A^1$.
The strict convexity of the pointwise minimization problem in \eqref{eq:driver-inf} gives uniqueness of the optimal strategy. $\square$

\section{Mean-field equilibrium with default risk}\label{sec:mfg}
Based on the results of the previous section, we construct a financial market with multiple agents.

\subsection{Multi-agent setting}
\label{sec:multi-agent}

Suppose there are $N\in\mathbb{N}$ agents in the common financial market.
In order to study the equilibrium state, let us first introduce the relevant probability spaces.\\

\noindent
(1) We denote by $(\Omega^i,\mathcal{F}^i,\mathbb{P}^i)$ ($i=1,\ldots,N$) a complete probability space with complete and right-continuous filtration $\mathbb{F}^i:=(\mathcal{F}^i_t)_{t\in[0,T]}$, generated by a $d$-dimensional standard Brownian motion $W^i:=(W^i_t)_{t\in[0,T]}$ and a $\sigma$-algebra $\sigma(\xi^i,\gamma^i)$, where $W^i$ is independent of $(\xi^i,\gamma^i)$.
$\mathcal{F}^i_0$ is the completion of $\sigma(\xi^i,\gamma^i)$. We set $\mathcal{F}^i:=\mathcal{F}^i_T$. \\

\noindent
(2) We denote by $(\Omega^{0,i},\mathcal{F}^{0,i},\mathbb{P}^{0,i})$ ($i=1,\ldots,N$) a complete probability space over $\Omega^{0,i} := \Omega^0 \times \Omega^i$. Here, $(\mathcal{F}^{0,i},\mathbb{P}^{0,i})$ is the completion of $(\mathcal{F}^0 \otimes \mathcal{F}^i,\mathbb{P}^0\otimes \mathbb{P}^i)$ and $\mathbb{F}^{0,i}:=(\mathcal{F}^{0,i}_t)_{t\in[0,T]}$ denotes the complete and right-continuous augmentation of $(\mathcal{F}_t^0 \otimes \mathcal{F}_t^i)_{t\in[0,T]}$.\\

\noindent
(3) Let $(\Omega^{(N)},\mathcal{F}^{(N)},\mathbb{P}^{(N)})$ be an enlarged complete probability space defined on $\Omega^{(N)}:=\prod_{i=0}^N\Omega^i$.
$(\mathcal{F}^{(N)},\mathbb{P}^{(N)})$ is the completion of $\Bigl(\bigotimes_{i=0}^N\mathcal{F}^i, \bigotimes_{i=0}^N\mathbb{P}^i\Bigr)$ and the filtration $\mathbb{F}^{(N)}=(\mathcal{F}^{(N)}_t)_{t\in[0,T]}$ is the complete and right-continuous augmentation of $(\bigotimes_{i=0}^N\mathcal{F}^i_t)_{t\in[0,T]}$. \\

\begin{asm}[Multi-agent setting]~\label{asm:multiagent}
  \begin{enumerate}
    \item[(i)] For each $i=1,\ldots,N$, Assumption \ref{asm:agent} holds with ``1'' replaced by ``$i$''.
    \item[(ii)] The tuples $\{(\xi^i,\gamma^i,W^i,F^i)\}_{i=1,\ldots,N}$ are $\mathcal{F}^0$-conditionally i.i.d. on $(\Omega^{(N)},\mathcal{F}^{(N)},\mathbb{P}^{(N)})$.
  \end{enumerate}
\end{asm}

We consider a utility maximization problem for each agent-$i$ defined on $(\Omega^{0,i},\mathcal{F}^{0,i},\mathbb{P}^{0,i})$ for each $i=1,\ldots,N$:
\begin{equation}
  \sup_{\pi\in\mathbb{A}^i}\mathcal U^i(\pi), ~~~\text{where}~~~
  \mathcal U^i(\pi):= \mathbb{E}^{\mathbb{P}^{0,i}}\Bigl[-\exp\Bigl(-\gamma^i\Bigl(\mathcal{W}^{i,\pi}_T-F^i\Bigr)\Bigr)\Bigr],
\end{equation}
subject to
\begin{equation}\label{eq:wealth-i}
  \mathcal{W}^{i,\pi}_t= \xi^i + \int_0^t\pi_s^{\top}\sigma_sdW^0_s + \int_0^t\pi_s^{\top}\sigma_s\theta_sds+ \int_0^t\pi_s^{\top}\beta_sdN_s,~~~t\in[0,T].
\end{equation}

\begin{dfn}~\\\label{def:admissible-i}
  The admissible set $\mathbb{A}^i$ consists of all $\mathbb{R}^n$-valued, $\mathbb{F}^{0,i}$-predictable processes $\pi$ such that
  \begin{enumerate}
    \item[(i)] $\displaystyle\mathbb{E}^{\mathbb{P}^{0,i}}\Bigl[\int_0^T|\pi_s^{\top}\sigma_s|^2ds\Bigr]<\infty$ and $\displaystyle\mathbb{E}^{\mathbb{P}^{0,i}}\Bigl[\int_0^T\lambda_s|\pi_s^{\top}\beta_s|^2ds\Bigr]<\infty$,
    \item[(ii)] the family $\{\exp(-\gamma^i\mathcal{W}^{i,\pi}_{\tau}); \tau\in\mathcal{T}(\mathbb{F}^{0,i})\}$ is uniformly integrable (with respect to $\mathbb{P}^{0,i}$).
  \end{enumerate}
\end{dfn}

As in Section \ref{sec:individual-bsde}, we put
\begin{equation}
  p_t:=\pi_t^{\top}\sigma_t\in L_t\subset\mathbb{R}^{1\times d_0},~~~~~~t\in[0,T],
\end{equation}
and denote
\begin{equation}
  \mathcal{A}^i:=\{p=\pi^{\top}\sigma;~\pi\in\mathbb{A}^i\},~~~i=1,\ldots,N.
\end{equation}
Following the previous section, we consider a BSDE \eqref{eq:BSDE-Y} with ``1'' replaced by ``$i$''.
Theorem \ref{thm:BSDE-wellposed} and Theorem \ref{thm:verify} yield a unique bounded solution $(Y^i,Z^{i,0},Z^i,U^i)\in\mathbb{S}^{\infty}\times\mathbb{H}^2_{\mathrm{BMO}}\times\mathbb{H}^2_{\mathrm{BMO}}\times\mathbb{L}^2_{\mathrm{BMO}}(M)$ of the individual BSDE \eqref{eq:BSDE-Y}
and the unique optimal control $p^{i,*}\in\mathcal{A}^i$ characterized by
\begin{equation}\label{eq:FOC-i}
  p^{i,*}_s = Z^{i,0\|}_s + \frac{\theta_s^{\top}}{\gamma^i} + \frac{\lambda_s}{\gamma^i}\eta_s^{\top}e^{-\gamma^i(p^{i,*}_s\eta_s-U^i_s)},~~~s\in[0,T].
\end{equation}
The associated share-allocation is recovered by $\pi^{i,*}_s=(\sigma_s\sigma_s^{\top})^{-1}\sigma_s(p^{i,*}_s)^{\top}\in\mathbb{R}^n$.

\begin{dfn}[Market-clearing condition]~\\\label{def:MC}
    The financial market satisfies the market-clearing condition if the equality
    \begin{equation}
        \label{eq:MC}
        \frac{1}{N}\sum_{i=1}^N \pi_t^{i,*} = 0,~~~dt\otimes \mathbb{P}^{(N)}\text{-a.e.}
    \end{equation}
    holds. Here, for each $i=1,\ldots,N$, $\pi^{i,*}$ denotes the optimal trading strategy of agent-$i$.
\end{dfn}

\subsection{The mean-field BSDE}
\label{sec:mf-heuristic}
In what follows, we take agent-1 as a representative agent and work on $(\Omega^{0,1},\mathcal F^{0,1},\mathbb P^{0,1})$. 
As in Section \ref{sec:individual-bsde}, $\mathbb{E}$ denotes expectation under $\mathbb{P}^{0,1}$.
Define
\begin{equation}\label{eq:harmonic-and-cond}
  \widehat{\gamma} := \mathbb{E}\Bigl[\dfrac{1}{\gamma^1}\Bigr]^{-1}\in[\underline{\gamma},\overline{\gamma}],~~~~~~ \overline{\mathbb{E}}[X_t](\omega^0) := \mathbb{E}[X_t| \mathcal{F}^0](\omega^0)= \mathbb{E}[X_t(\omega^0,\cdot)],
\end{equation}
for an $\mathbb{F}^{0,1}$-predictable integrable process $X$.
$\overline{\mathbb{E}}[X]$ is taken in its $\mathbb{F}^0$-predictable version.
Moreover, if $X_t(\omega^0,\omega^1)\in L_t(\omega^0)$, $dt\otimes\mathbb{P}^{0,1}$-a.e., then $\overline{\mathbb{E}}[X_t](\omega^0)\in L_t(\omega^0)$, $dt\otimes \mathbb{P}^0$-a.e.

Following Fujii \& Sekine \cite[Section 4.1]{FujiiSekine2025}, we now derive heuristically a candidate equilibrium risk premium from the finite-population market above.
Since the idiosyncratic filtrations $\mathcal{F}^i_t$ and $\mathcal{F}^j_t$ ($i\neq j$) are mutually independent, and Assumption \ref{asm:multiagent} makes the agents symmetric, the finite family $\{\pi^{i,*}\}_{i=1}^N$ is expected to be exchangeable. 
By De Finetti's theorem (see, e.g., \cite[Theorem 2.1]{CarmonaDelarue2018II}), letting $N\to\infty$ in \eqref{eq:MC} formally replaces the average demand by the $\mathcal{F}^0$-conditional mean $\overline{\mathbb{E}}[\pi^{1,*}_t]$ of a representative agent.
This motivates us to consider the mean-field market-clearing relation
\begin{equation}
  \overline{\mathbb{E}}\bigl[\pi^{1,*}_t\bigr] = 0, ~~~ dt\otimes\mathbb{P}^0\text{-a.e.},
\end{equation}
which is equivalent to
\begin{equation}\label{eq:mc-limit}
  \overline{\mathbb{E}}\bigl[p^{1,*}_t\bigr] = 0, ~~~ dt\otimes\mathbb{P}^0\text{-a.e.}
\end{equation}

By \eqref{eq:FOC-i}, we have
\begin{equation}
  p^{1,*}_s = Z^{1,0\|}_s+\frac{\theta_s^{\top}}{\gamma^1}+\lambda_s\eta_s^{\top} \Phi^1_s,~~~~~~ \Phi^1_s:=\frac{e^{-q^{1,*}_s}}{\gamma^1},~~~ q^{1,*}_s:=\gamma^1(p^{1,*}_s\eta_s-U^1_s),~~~s\in[0,T].
\end{equation}
Since $(\theta,\eta,\lambda)$ are $\mathcal{F}^0$-measurable, taking $\overline{\mathbb{E}}[\cdot]$ yields
\begin{equation}\label{eq:Ep-avg}
  \overline{\mathbb{E}}\bigl[p^{1,*}_s\bigr] = \overline{\mathbb{E}}\bigl[Z^{1,0\|}_s\bigr]+\frac{\theta_s^{\top}}{\widehat\gamma} + \lambda_s\eta_s^{\top} \psi^{*}_s,~~~s\in[0,T],
\end{equation}
where
\begin{equation}\label{eq:Phi-def}
  \psi^{*}_s := \overline{\mathbb{E}}\bigl[\Phi^1_s\bigr] = \overline{\mathbb{E}}\Bigl[\frac{e^{-q^{1,*}_s}}{\gamma^1}\Bigr],~~~s\in[0,T].
\end{equation}
Solving \eqref{eq:mc-limit} for $\theta_s$ gives the candidate equilibrium risk-premium process
\begin{equation}\label{eq:theta-mfg-exact}
  \theta^{\mathrm{mfg}}_s = -\widehat{\gamma} \overline{\mathbb{E}}[Z^{1,0\|}_s]^\top - \widehat{\gamma} \lambda_s \eta_s \psi^{*}_s,~~~s\in[0,T].
\end{equation}

Set
\begin{equation}
  \mathcal{M}^1_s:=\overline{\mathbb{E}}[Z^{1,0\|}_s] +\lambda_s\eta_s^{\top}\psi^{*}_s\in L_s,~~~s\in[0,T].
\end{equation}
Then \eqref{eq:theta-mfg-exact} reads $\theta^{\mathrm{mfg}}_s=-\widehat{\gamma}(\mathcal{M}^1_s)^{\top}$.
Substituting $\theta=\theta^{\mathrm{mfg}}$ into the driver \eqref{eq:driver-q}, we obtain the mean-field BSDE
\begin{equation}\label{eq:MF-BSDE}
  Y^1_t = F^1 + \int_t^T f^{1,\mathrm{mfg}}(s,Z^{1,0}_s,Z^1_s,U^1_s) ds - \int_t^TZ^{1,0}_s dW^0_s - \int_t^TZ^1_s dW^1_s - \int_t^TU^1_s dM_s, ~~~ t\in[0,T],
\end{equation}
with driver
\begin{equation}\label{eq:MF-driver}
\begin{split}
  f^{1,\mathrm{mfg}}(s,Z^{1,0}_s,Z^1_s,U^1_s) &= \widehat{\gamma} Z^{1,0\|}_s(\mathcal{M}^1_s)^{\top} - \frac{\widehat{\gamma}^2}{2\gamma^1}|\mathcal{M}^1_s|^2 + \frac{\gamma^1}{2}\bigl(|Z^{1,0\perp}_s|^2+|Z^1_s|^2\bigr)\\
  &~~~ -\lambda_sU^1_s + \lambda_s\Phi^1_s - \frac{\lambda_s}{\gamma^1} + \frac{\gamma^1\lambda_s^2|\eta_s|^2}{2}(\Phi^1_s)^2,~~~s\in[0,T],
\end{split}
\end{equation}
where
\begin{equation}\label{eq:psi-consistency}
  \psi^{*}_s = \overline{\mathbb{E}}\Bigl[\frac{e^{-q^{1,*}_s}}{\gamma^1}\Bigr],~~~s\in[0,T],
\end{equation}
\begin{equation}\label{eq:q-star-eq}
  q^{1,*}_s-\lambda_s|\eta_s|^2 e^{-q^{1,*}_s} = \gamma^1\bigl(Z^{1,0\|}_s\eta_s-U^1_s\bigr) - \widehat{\gamma}\bigl(\overline{\mathbb{E}}[Z^{1,0\|}_s]\eta_s+\lambda_s|\eta_s|^2\psi^{*}_s\bigr),~~~s\in[0,T].
\end{equation}
We will show in Section \ref{sec:clearing} that this $\theta^{\mathrm{mfg}}$ actually clears the market as the population tends to infinity.
\begin{rem}\label{rem:theta-decomp}~\\
We provide an economic interpretation for the candidate \eqref{eq:theta-mfg-exact}. The stock return $\mu^{\mathrm{mfg}}=\sigma\theta^{\mathrm{mfg}}$ reads:
\begin{equation}
  \mu^{\mathrm{mfg}}_s = -\widehat{\gamma}\sigma_s \overline{\mathbb{E}}\bigl[Z^{1,0\|}_s\bigr]^\top - \widehat{\gamma} \lambda_s \psi^{*}_s\beta_s,~~~s\in[0,T].
\end{equation}
It consists of two components with the following economic roles.\\
(1) The term $-\widehat{\gamma}\sigma_s \overline{\mathbb{E}}\bigl[Z^{1,0\|}_s\bigr]^\top$ is the Brownian hedging component.
For the $k$-th stock, it is proportional to $\sigma^{(k)}_s\overline{\mathbb{E}}[Z^{1,0\|}_s]^\top$, the instantaneous covariance between the stock's diffusive shock $\sigma^{(k)}_s dW^0_s$ and the population-average common-noise exposure $\overline{\mathbb{E}}[Z^{1,0\|}_s]dW^0_s$, where $\sigma^{(k)}$ denotes the $k$-th row vector of $\sigma$.
Roughly speaking, $Z^{1,0}$ represents the common-noise exposure that agent-1 needs to hedge, taking into account both the terminal liability and the continuation value of the optimization problem.
If this covariance is positive, a long position in the stock serves as a hedge against the population's exposure, and the resulting hedging demand allows the stock to clear at a lower expected return.
If it is negative, the stock aggravates the exposure and must offer a higher return to clear.
See \cite[Example 4.3.15]{SekineThesis2025} for a closed-form illustration in the purely diffusive setting.\\
(2) The vector term $-\widehat{\gamma}\lambda_s\psi^{*}_s\beta_s$ is the default-risk component.
It is proportional to the default intensity $\lambda_s$ and the stock jump-size vector $\beta_s$, and is scaled by the endogenous factor $\widehat{\gamma}\psi^*_s$.
Here $\psi^*_s=\overline{\mathbb{E}}[\Phi^1_s]$ is the population average of the agent-level factor $\Phi^1_s=e^{-q^{1,*}_s}/\gamma^1$, which combines how unfavorable the default state is for the agent, measured by $e^{-q^{1,*}_s}$ (recall Remark \ref{rem:econ-q}), with the risk tolerance $1/\gamma^1$.
Thus, $\widehat{\gamma}\psi^*_s$ measures how unfavorable the default state is for the population as a whole, with the agent-level effects weighted by risk tolerance.
Since $\widehat{\gamma}\lambda_s\psi^*_s$ is nonnegative, the $k$-th component has the opposite sign to $\beta_s^{(k)}$.
Consequently, if the $k$-th security loses value at default ($\beta_s^{(k)}<0$), its default-risk component is positive and compensates its holders for the default loss, whereas if it pays off at default ($\beta_s^{(k)}>0$), the component is negative and reflects the insurance-like hedge it provides.
Once default has occurred, $\lambda_s=0$ and the component vanishes.
\end{rem}

As in Section \ref{sec:wellposed}, the scalar equation \eqref{eq:q-star-eq} determines $q^{1,*}$ pointwise through the map $Q_s$ for each fixed candidate $\psi$.
The new feature is that this candidate itself must also satisfy the consistency condition \eqref{eq:psi-consistency}.
Thus, the default term introduces an additional scalar fixed-point component, which will be handled together with the mean-field BSDE in the Markovian analysis below.

\subsection{The scalar fixed-point map and stability estimates}
\label{sec:scalar-fixed-point}

In this subsection, we analyze the scalar fixed-point problem \eqref{eq:psi-consistency} and \eqref{eq:q-star-eq}.
Here and in what follows, lowercase triples $(z^0,z^1,u)$ denote candidate integrands playing the role of $(Z^{1,0},Z^1,U^1)$ in \eqref{eq:MF-driver}. 
Note that they are not the $\gamma^1$-rescaled variables in Section \ref{sec:ito-driver}.
Since the representative agent index ``1'' is fixed throughout, we hereafter suppress it when it is obvious.

We first record the following simple observation.

\begin{lem}~\\\label{lem:conditional-BMO-compatibility}
Let $\zeta\in\mathbb{H}^2_{\mathrm{BMO}}(\mathbb{P}^{0,1},\mathbb{F}^{0,1})$ and set $\bar \zeta:=\overline{\mathbb{E}}[\zeta]$.
Then, $\|\bar \zeta\|_{\mathbb{H}^2_{\mathrm{BMO}}} \leq \|\zeta\|_{\mathbb{H}^2_{\mathrm{BMO}}}$.
\end{lem}

\noindent\textbf{\textit{Proof}}\\
\noindent
Note that $\bar\zeta$ is $\mathbb F^0$-predictable and in particular, independent of $\mathbb F^1$. By Jensen's inequality, we have
\begin{equation}
  \sup_{\tau_0\in\mathcal{T}^{0,1}}\Bigl\|\mathbb{E}\Bigl[\int_{\tau_0}^T|\bar \zeta_s|^2ds\Bigm|\mathcal{F}^{0,1}_{\tau_0}\Bigr]\Bigr\|_{\mathbb{L}^{\infty}} = \sup_{\tau_0\in\mathcal{T}^{0}}\Bigl\|\mathbb{E}\Bigl[\int_{\tau_0}^T|\bar \zeta_s|^2ds\Bigm|\mathcal{F}^0_{\tau_0}\Bigr]\Bigr\|_{\mathbb{L}^{\infty}}\leq \sup_{\tau_0\in\mathcal{T}^{0}}\Bigl\|\mathbb{E}\Bigl[\int_{\tau_0}^T|\zeta_s|^2ds\Bigm|\mathcal{F}^0_{\tau_0}\Bigr]\Bigr\|_{\mathbb{L}^{\infty}}.
\end{equation}
For any $\tau_0\in\mathcal T^0\subset\mathcal{T}^{0,1}$, it holds that
\begin{equation}
  \begin{split}
    \mathbb{E}\Bigl[\int_{\tau_0}^T|\zeta_s|^2ds\Bigm|\mathcal{F}^0_{\tau_0}\Bigr]
    =
    \mathbb{E}\Bigl[
    \mathbb{E}\Bigl[\int_{\tau_0}^T|\zeta_s|^2ds\Bigm|\mathcal{F}^{0,1}_{\tau_0}\Bigr]
    \Bigm|\mathcal{F}^0_{\tau_0}\Bigr]
    \leq
    \|\zeta\|^2_{\mathbb{H}^2_{\mathrm{BMO}}}.
  \end{split}
\end{equation}
Taking the essential supremum over $\tau_0\in\mathcal T^0$ proves the claim. $\square$

For $K\geq0$, define constants
\begin{equation}\label{eq:scalar-constants}
  \begin{split}
  &\bar a_K:=\overline{\gamma}K(\overline h+1)+\widehat{\gamma}K\overline h,
  ~~~
  b_{\psi}:=\widehat{\gamma}\overline{\lambda}\overline h^2,
  ~~~
  \Psi_K:=\frac{e^{\bar a_K}}{\underline{\gamma}},\\
  &\bar q_K:=\bar a_K+\overline{\lambda}\overline h^2+b_{\psi}\Psi_K,
  ~~~
  \delta_K:=\frac{e^{-\bar q_K}}{e^{-\bar q_K}+\overline{\lambda}\overline h^2}
  \in(0,1),
  ~~~
  \overline{\Phi}_K:=\frac{e^{\bar q_K}}{\underline{\gamma}}.
  \end{split}
\end{equation}

\begin{lem}~\\\label{lem:inner-branch}
  Let Assumptions \ref{asm:market} and \ref{asm:agent} hold.
  Fix a constant $K\geq0$ and predictable processes $(z^0,z^1,u)$ with $|z^0_t|\leq K$ and $|u_t|\leq K$, $dt\otimes\mathbb{P}^{0,1}$-a.e., and set $\bar z^{0\|}_t:=\overline{\mathbb{E}}[z^{0\|}_t]$ for $t\in[0,T]$.
  Define
  \begin{equation}
    \mathcal I_K:=\{\psi\in\mathbb{L}^{\infty}(\mathbb{F}^0,\mathbb{R}); \psi\text{ is }\mathbb F^0\text{-predictable and }|\psi_t|\leq\Psi_K,  dt\otimes\mathbb{P}^0\text{-a.e.}\}.
  \end{equation}
  For $\psi\in\mathcal I_K$, set
  \begin{equation}
    \varrho^\psi_t:=\gamma^1(z^{0\|}_t\eta_t-u_t) -\widehat{\gamma}(\bar z^{0\|}_t\eta_t+\lambda_t|\eta_t|^2\psi_t),~~~t\in[0,T].
  \end{equation}
  Define $q^\psi_t:=Q_t(\varrho^\psi_t)$ and $\Phi^{\psi}_t:=\dfrac{e^{-q^{\psi}_t}}{\gamma^1}$ for $t\in[0,T]$.
  Then,
  \begin{equation}
    |q^\psi_t|\leq \bar q_K,~~~0<\Phi^\psi_t\leq\overline{\Phi}_K,~~~dt\otimes\mathbb{P}^{0,1}\text{-a.e.}
  \end{equation}
  hold. Moreover, define
  $\mathcal{S}:\mathcal I_K \to \mathbb{L}^{\infty}(\mathbb{F}^0,\mathbb{R})$ by $\mathcal S(\psi)_t:=\overline{\mathbb{E}}[\Phi^{\psi}_t]$, $dt\otimes\mathbb{P}^0$-a.e.
  Then, $\mathcal{S}(\mathcal I_K)\subset \mathcal I_K$ and, for $\delta_K\in(0,1)$ defined in \eqref{eq:scalar-constants}, it holds that
  \begin{equation}
    \|\mathcal S(\psi)-\mathcal S(\psi')\|_{\mathbb{L}^{\infty}} \leq (1-\delta_K)\|\psi-\psi'\|_{\mathbb{L}^{\infty}},~~~ \psi,\psi'\in\mathcal I_K,
  \end{equation}
  i.e., $\mathcal S$ is a contraction map on $\mathcal I_K$.
  Consequently, there exists a unique fixed point $\psi^{*}\in\mathcal I_K$ of $\mathcal{S}$, which satisfies $0<\psi^{*}\leq\Psi_K$.
\end{lem}

\noindent\textbf{\textit{Proof}}\\
The predictability of $q^\psi$ follows from the definition of $Q_t$ and Lemma \ref{lem:qstar-apriori}.
For $\psi\in\mathcal I_K$, we have
\begin{equation}
  |\varrho^\psi_t|\leq \bar a_K+b_\psi\Psi_K,~~~ dt\otimes\mathbb{P}^{0,1}\text{-a.e.}
\end{equation}
and \eqref{eq:qstar-root-exp-bound} gives
\begin{equation}
  |q^\psi_t|\leq\bar q_K,~~~0<\Phi^\psi_t\leq\overline{\Phi}_K,~~~dt\otimes\mathbb{P}^{0,1}\text{-a.e.}
\end{equation}

Let $\psi,\psi'\in\mathcal I_K$.  For $\ell\in[0,1]$, set $\psi^\ell:=(1-\ell)\psi'+\ell\psi$ and $q^\ell:=q^{\psi^\ell}$.
Since $\dfrac{d}{d\ell}\varrho^{\psi^\ell}_t=-\widehat{\gamma}\lambda_t|\eta_t|^2(\psi_t-\psi'_t)$ and \eqref{eq:qstar-derivative} gives $\dfrac{d}{d\ell}q^\ell_t=-\dfrac{\widehat{\gamma}\lambda_t|\eta_t|^2}{1+\lambda_t|\eta_t|^2e^{-q^\ell_t}}(\psi_t-\psi'_t)$, we obtain
\begin{equation}\label{eq:S-increment}
\begin{split}
  \mathcal S(\psi)_t-\mathcal S(\psi')_t
  &=
  \overline{\mathbb{E}}\Bigl[
    \int_0^1
    \frac{d}{d\ell}\Bigl(\frac{e^{-q^\ell_t}}{\gamma^1}\Bigr)d\ell
  \Bigr]
  =
  (\psi_t-\psi'_t)
  \overline{\mathbb{E}}\Bigl[
    \int_0^1
    \frac{\widehat{\gamma}\lambda_t|\eta_t|^2}
    {\gamma^1\bigl(e^{q^\ell_t}+\lambda_t|\eta_t|^2\bigr)}
    d\ell
  \Bigr],
  ~~~ dt\otimes\mathbb{P}^0\text{-a.e.}
\end{split}
\end{equation}
Since $|q^\ell_t|\leq\bar q_K$ and $\widehat{\gamma}=\mathbb{E}[(\gamma^1)^{-1}]^{-1}$,
\begin{equation}\label{eq:S-slope-gap}
  \begin{split}
  0
  &\leq
  \overline{\mathbb{E}}\Bigl[
    \int_0^1
    \frac{\widehat{\gamma}\lambda_t|\eta_t|^2}
    {\gamma^1\bigl(e^{q^\ell_t}+\lambda_t|\eta_t|^2\bigr)}
    d\ell
  \Bigr]
  \leq
  \frac{\lambda_t|\eta_t|^2}{e^{-\bar q_K}+\lambda_t|\eta_t|^2}
  \leq
  \frac{\overline\lambda\overline h^2}{e^{-\bar q_K}+\overline\lambda\overline h^2}
  =1-\delta_K<1.
  \end{split}
\end{equation}
Here, we used the fact that $x\mapsto\dfrac{x}{e^{-\bar q_K}+x}$ is increasing for $x\geq0$.

We then show that $\mathcal{S}(\mathcal I_K)\subset \mathcal I_K$.
Since $\Phi^\psi>0$, we only need to give the upper bound of $\mathcal S(\psi)$.
Taking $\psi'_t\equiv\Psi_K$ in \eqref{eq:S-increment} and using the non-negativity in \eqref{eq:S-slope-gap}, we get $\mathcal S(\psi)_t-\mathcal S(\Psi_K)_t\leq0$, $dt\otimes\mathbb{P}^0\text{-a.e.}$
Thus it is enough to show $\mathcal S(\Psi_K)_t\leq\Psi_K$.
Set $\bar\varrho_t:=\gamma^1(z^{0\|}_t\eta_t-u_t)-\widehat{\gamma}\bar z^{0\|}_t\eta_t$ for $t\in[0,T]$.
On $\{\lambda=0\}$, $q^{\Psi_K}_t=\bar\varrho_t$ and
\begin{equation}
  \mathcal S(\Psi_K)_t
  =
  \overline{\mathbb{E}}\Bigl[\frac{e^{-\bar\varrho_t}}{\gamma^1}\Bigr]
  \leq \frac{e^{\bar a_K}}{\underline{\gamma}}
  \leq \Psi_K .
\end{equation}
On $\{\lambda>0\}$, the inequality
$\widehat{\gamma}\Psi_K\geq e^{\bar a_K}\geq e^{-\bar\varrho_t}$ implies
\begin{equation}
  q^{\Psi_K}_t-\lambda_t|\eta_t|^2e^{-q^{\Psi_K}_t}
  =
  \varrho^{\Psi_K}_t
  =
  \bar\varrho_t-\widehat{\gamma}\lambda_t|\eta_t|^2\Psi_K
  \leq
  \bar\varrho_t-\lambda_t|\eta_t|^2e^{-\bar\varrho_t}.
\end{equation}
Since $r\mapsto r-\lambda_t|\eta_t|^2e^{-r}$ is increasing, this gives $q^{\Psi_K}_t\leq\bar\varrho_t$.
This also yields
\begin{equation}
  \lambda_t|\eta_t|^2e^{-q^{\Psi_K}_t}
  =
  q^{\Psi_K}_t - \bar\varrho_t + \widehat{\gamma}\lambda_t|\eta_t|^2\Psi_K
  \leq
  \widehat{\gamma}\lambda_t|\eta_t|^2\Psi_K.
\end{equation}
Since $\lambda_t|\eta_t|^2>0$ and $\gamma^1>0$, we get $\dfrac{e^{-q^{\Psi_K}_t}}{\gamma^1}\leq\dfrac{\widehat{\gamma}}{\gamma^1}\Psi_K$.
Taking $\overline{\mathbb{E}}$, we deduce $\mathcal S(\Psi_K)_t \leq \Psi_K$, $dt\otimes\mathbb P^0$-a.e.
Therefore $\mathcal S(\mathcal I_K)\subseteq\mathcal I_K$.  Now, \eqref{eq:S-slope-gap} and \eqref{eq:S-increment} give
\begin{equation}
  |\mathcal S(\psi)_t-\mathcal S(\psi')_t|
  \leq
  (1-\delta_K)|\psi_t-\psi'_t|,
  ~~~ dt\otimes\mathbb{P}^0\text{-a.e.}
\end{equation}
Taking the essential supremum proves the contraction estimate.
The Banach fixed-point theorem gives a unique fixed point $\psi^*\in\mathcal I_K$, and $\psi^*>0$ follows from $\psi^*=\mathcal S(\psi^*)$. $\square$

\begin{rem}~\\\label{rem:scalar-bounded-uniqueness}
  The formula $\mathcal S(\psi)_t:=\overline{\mathbb E}[\Phi^\psi_t]$ for $t\in[0,T]$ also defines $\mathcal S$ on $\mathbb{L}^{\infty}(\mathbb{F}^0,\mathbb{R})$.
  In that case, the identity \eqref{eq:S-increment} remains valid for $\psi,\psi'\in\mathbb{L}^{\infty}(\mathbb{F}^0,\mathbb{R})$. 
  This implies that $\mathcal S$ has at most one fixed point in $\mathbb{L}^{\infty}(\mathbb{F}^0,\mathbb{R})$.
\end{rem}

\begin{lem}~\\\label{lem:inner-branch-stability}
  Let Assumptions \ref{asm:market} and \ref{asm:agent} hold and fix $K\geq0$.
  Take two triples of predictable processes $(z^0,z^1,u)$ and $(\widehat z^0,\widehat z^1,\widehat u)$ with $|z^0_t|,|\widehat z^0_t|, |u_t|,|\widehat u_t|\leq K$, and $u=\widehat u=0$ on $\{\lambda=0\}$, $dt\otimes\mathbb{P}^{0,1}$-a.e.
  For each $\psi\in\mathcal I_K$, let $(\varrho^\psi,q^\psi,\Phi^\psi)$ and $(\widehat\varrho^\psi,\widehat q^\psi,\widehat\Phi^\psi)$ be the processes constructed in Lemma \ref{lem:inner-branch} from $(z^0,z^1,u)$ and $(\widehat z^0,\widehat z^1,\widehat u)$, respectively.
  Let $\mathcal S$ and $\widehat{\mathcal S}$ be the corresponding maps and denote their unique fixed points by $\psi^*$ and $\widehat\psi^*$, respectively.
  Set
  \begin{equation}
    \Phi^*_t:=\Phi^{\psi^*}_t=\frac{e^{-q^{\psi^*}_t}}{\gamma^1},~~~
    \widehat\Phi^*_t:=\widehat\Phi^{\widehat\psi^*}_t=\frac{e^{-\widehat q^{\widehat\psi^*}_t}}{\gamma^1},~~~
    \mathcal M^*_t:=\overline{\mathbb E}[z^{0\|}_t]+\lambda_t\eta_t^\top\psi^{*}_t,~~~
    \widehat{\mathcal M}^*_t:=\overline{\mathbb E}[\widehat z^{0\|}_t]+\lambda_t\eta_t^\top\widehat\psi^{*}_t,
  \end{equation}
  for $t\in[0,T]$, and define
  \begin{equation}
  \begin{split}
    &\Delta z^0:=z^0-\widehat z^0,~~~
     \Delta u:=u-\widehat u,~~~
     \Delta\psi^{*}:=\psi^{*}-\widehat\psi^{*},\\
    &\Delta q^{*}:=q^{\psi^*}-\widehat q^{\widehat\psi^*},~~~
     \Delta\Phi^{*}:=\Phi^{*}-\widehat\Phi^{*},~~~
     \Delta\mathcal M^{*}:=\mathcal M^{*}-\widehat{\mathcal M}^{*}.
  \end{split}
  \end{equation}
  Then,
  \begin{equation}\label{eq:delta-scalar-Lip}
    \|\Delta\psi^{*}\|_{\mathbb{H}^2_{\mathrm{BMO}}} +\|\Delta\Phi^*\|_{\mathbb{H}^2_{\mathrm{BMO}}} +\|\Delta\mathcal{M}^*\|_{\mathbb{H}^2_{\mathrm{BMO}}} \leq C\bigl(\|\Delta z^0\|_{\mathbb{H}^2_{\mathrm{BMO}}} +\|\Delta u\|_{\mathbb{L}^2_{\mathrm{BMO}}(M)}\bigr),
  \end{equation}
  where $C$ depends only on $K$, $\Psi_K$, and the model constants.
\end{lem}

\noindent\textbf{\textit{Proof}}\\
Set $\Delta\bar z^{0\|}:=\overline{\mathbb{E}}[\Delta z^{0\|}]$.
By Lemma \ref{lem:inner-branch}, $q^{\psi^*}_t,q^{\widehat\psi^*}_t,\widehat q^{\widehat\psi^*}_t\in[-\bar q_K,\bar q_K]$, $dt\otimes\mathbb P^{0,1}$-a.e.
Since $Q_t$ is $1$-Lipschitz by \eqref{eq:qstar-derivative}, we have
\begin{equation}
  \begin{split}
    \bigl|q^{\widehat\psi^*}_t-\widehat q^{\widehat\psi^*}_t\bigr|
    \leq
    |\varrho^{\widehat\psi^*}_t-\widehat\varrho^{\widehat\psi^*}_t|
    =
    |\gamma^1(\Delta z^{0\|}_t\eta_t-\Delta u_t)-\widehat\gamma\Delta\bar z^{0\|}_t\eta_t|
    \leq
    C\bigl(|\Delta z^{0\|}_t|+|\Delta u_t|+|\Delta\bar z^{0\|}_t|\bigr),~~~dt\otimes\mathbb P^{0,1}\text{-a.e.}
  \end{split}
\end{equation}
Since the map $x\mapsto\dfrac{e^{-x}}{\gamma^1}$ is locally Lipschitz, with Lipschitz constant on $[-\bar q_K,\bar q_K]$ bounded by $\overline\Phi_K$, we have
\begin{equation}
  |\Phi^{\widehat\psi^*}_t-\widehat\Phi^{\widehat\psi^*}_t|
  \leq
  \overline\Phi_K|q^{\widehat\psi^*}_t-\widehat{q}^{\widehat\psi^*}_t|
  \leq
  C\bigl(|\Delta z^{0\|}_t|+|\Delta u_t|+|\Delta\bar z^{0\|}_t|\bigr),~~~dt\otimes\mathbb P^{0,1}\text{-a.e.}
\end{equation}
Moreover,
\begin{equation}
  |\Phi^{\psi^*}_t-\Phi^{\widehat\psi^*}_t|
  \leq
  \overline\Phi_K|q^{\psi^*}_t-q^{\widehat\psi^*}_t|
  \leq
  \overline\Phi_K|\varrho^{\psi^*}_t-\varrho^{\widehat\psi^*}_t|
  \leq
  C|\Delta\psi^*_t|,~~~dt\otimes\mathbb P^{0,1}\text{-a.e.}
\end{equation}
Conditional Jensen's inequality also gives $|\Delta\bar z^{0\|}_t|\leq\overline{\mathbb E}[|\Delta z^{0\|}_t|]$, $dt\otimes\mathbb P^0$-a.e.

Observe that, by Lemma \ref{lem:inner-branch},
\begin{equation}
  \begin{split}
    |\Delta\psi^*_t|
    &=
    |\mathcal S(\psi^*)_t-\widehat{\mathcal S}(\widehat\psi^*)_t|\\
    &\leq
    |\mathcal S(\psi^*)_t-\mathcal S(\widehat\psi^*)_t| + |\mathcal S(\widehat\psi^*)_t-\widehat{\mathcal S}(\widehat\psi^*)_t|\\
    &\leq
    (1-\delta_K)|\psi^*_t - \widehat\psi^*_t| + \overline{\mathbb E}\bigl[|\Phi^{\widehat\psi^*}_t-\widehat\Phi^{\widehat\psi^*}_t|\bigr]\\
    &\leq
    (1-\delta_K)|\psi^*_t - \widehat\psi^*_t| + C\bigl(\overline{\mathbb{E}}[|\Delta z^{0\|}_t|]+\overline{\mathbb{E}}[|\Delta u_t|]\bigr),~~~dt\otimes\mathbb P^0\text{-a.e.},
  \end{split}
\end{equation}
where we used the fact that $\psi^*$ and $\widehat\psi^*$ are fixed points of $\mathcal{S}$ and $\widehat{\mathcal{S}}$, respectively. As $\delta_K>0$, we deduce that
\begin{equation}
  \begin{split}
    |\Delta\psi^*_t|
    &\leq
    C\bigl(\overline{\mathbb{E}}[|\Delta z^{0\|}_t|]+\overline{\mathbb{E}}[|\Delta u_t|]\bigr),~~~dt\otimes\mathbb P^0\text{-a.e.}
  \end{split}
\end{equation}

Moreover, notice that
\begin{equation}
  \begin{split}
  |\Delta\Phi^*_t|
  &\leq
  |\Phi^{\psi^*}_t-\Phi^{\widehat\psi^*}_t| +|\Phi^{\widehat\psi^*}_t-\widehat\Phi^{\widehat\psi^*}_t|\\
  &\leq
  C\bigl(|\Delta\psi^*_t|+|\Delta z^{0\|}_t|+|\Delta u_t|+|\Delta\bar z^{0\|}_t|\bigr)\\
  &\leq
  C\bigl(\overline{\mathbb{E}}[|\Delta z^{0\|}_t|]+\overline{\mathbb{E}}[|\Delta u_t|]+|\Delta z^{0\|}_t|+|\Delta u_t|\bigr),~~~dt\otimes\mathbb P^{0,1}\text{-a.e.}
  \end{split}
\end{equation}
Also, we have
\begin{equation}
  |\Delta\mathcal M^*_t|
  \leq
  |\Delta\bar z^{0\|}_t|+C|\Delta\psi^*_t|
  \leq
  C\bigl(\overline{\mathbb{E}}[|\Delta z^{0\|}_t|]+\overline{\mathbb{E}}[|\Delta u_t|]\bigr),~~~dt\otimes\mathbb P^0\text{-a.e.}
\end{equation}
Combining the above estimates, we have
\begin{equation}
  |\Delta\psi^{*}_t| + |\Delta\Phi^*_t| + |\Delta\mathcal M^*_t|
  \leq
  C\bigl(\overline{\mathbb{E}}[|\Delta z^{0\|}_t|]+\overline{\mathbb{E}}[|\Delta u_t|]+|\Delta z^{0\|}_t|+|\Delta u_t|\bigr),~~~dt\otimes\mathbb P^{0,1}\text{-a.e.}
\end{equation}
Since $u=\widehat u=0$ on $\{\lambda=0\}$ and $\lambda\geq\underline\lambda$ on $\{\lambda>0\}$, we have, for every $\tau_0\in\mathcal T^{0,1}$,
\begin{equation}
  \mathbb E\Bigl[\int_{\tau_0}^T|\Delta u_t|^2dt\Bigm|\mathcal F^{0,1}_{\tau_0}\Bigr]
  \leq
  \frac{1}{\underline\lambda}\mathbb E\Bigl[\int_{\tau_0}^T\lambda_t|\Delta u_t|^2dt\Bigm|\mathcal F^{0,1}_{\tau_0}\Bigr],~~~\mathbb P^{0,1}\text{-a.s.}
\end{equation}
Hence $\|\Delta u\|_{\mathbb H^2_{\mathrm{BMO}}}\leq \underline\lambda^{-\frac{1}{2}}\|\Delta u\|_{\mathbb L^2_{\mathrm{BMO}}(M)}$.
Squaring the preceding pointwise estimate and applying Lemma \ref{lem:conditional-BMO-compatibility} give
\begin{equation}
  \begin{split}
  &\mathbb E\Bigl[\int_{\tau_0}^T (|\Delta\psi^*_t|^2+|\Delta\Phi^*_t|^2+|\Delta\mathcal M^*_t|^2)dt \Bigm|\mathcal F^{0,1}_{\tau_0}\Bigr]
  \leq
  C(\|\Delta z^0\|_{\mathbb H^2_{\mathrm{BMO}}}^2+\|\Delta u\|_{\mathbb L^2_{\mathrm{BMO}}(M)}^2),~~~\mathbb P^{0,1}\text{-a.s.}
  \end{split}
\end{equation}
Taking the essential supremum over $\tau_0\in\mathcal T^{0,1}$ yields \eqref{eq:delta-scalar-Lip}. $\square$

\subsection{Markovian factor model and scalar fixed-point representation}
\label{sec:markovian}

\begin{asm}[Markovian factor model]~\label{asm:markovian}
  \begin{enumerate}
  \item[(i)] The Markovian factor processes $X^0\in\mathbb{S}^2(\mathbb{F}^0,\mathbb{R}^{m_0})$ and $X^i\in\mathbb{S}^2(\mathbb{F}^i,\mathbb{R}^m)$ are given by
  \begin{equation}\label{eq:markovian-factors}
    \begin{split}
      X^0_t &=X^0_0 +\int_0^t b^0(s,X^0_s,N_{s-})ds +\int_0^t a^0(s,X^0_s,N_{s-})dW^0_s,~~~t\in[0,T],\\
      X^i_t &=X^i_0 +\int_0^t b(s,X^i_s)ds +\int_0^t a(s,X^i_s)dW^i_s, ~~~t\in[0,T],~~~ i=1,\ldots,N,
    \end{split}
  \end{equation}
  where $X^0_0\in\mathbb{R}^{m_0}$ and $X^i_0\in\mathbb{L}^2(\mathcal{F}^i_0,\mathbb{R}^m)$.
  $(X^i_0)_{i=1,\ldots,N}$ are i.i.d. on $(\Omega^{(N)},\mathcal{F}^{(N)},\mathbb{P}^{(N)})$ and
  \begin{equation}
    \begin{split}
      &b^0:[0,T]\times\mathbb{R}^{m_0}\times\{0,1\}\to\mathbb{R}^{m_0},
      ~~~
      a^0:[0,T]\times\mathbb{R}^{m_0}\times\{0,1\}\to\mathbb{R}^{m_0\times d_0},\\
      &b:[0,T]\times\mathbb{R}^m\to\mathbb{R}^m,
      ~~~
      a:[0,T]\times\mathbb{R}^m\to\mathbb{R}^{m\times d}
    \end{split}
  \end{equation}
  are bounded and measurable. For $j\in\{0,1\}$, we write $b^{0,(j)}(t,x^0):=b^0(t,x^0,j)$ and $a^{0,(j)}(t,x^0):=a^0(t,x^0,j)$.
  For some $\alpha\in(0,1)$, these maps are Lipschitz continuous in the spatial variables and $\dfrac{\alpha}{2}$-H\"older continuous in time: there exists $C>0$ such that for any $t,\acute t\in[0,T]$, $x^0,\acute x^0\in\mathbb{R}^{m_0}$, and $x^1,\acute x^1\in\mathbb{R}^{m}$,
  \begin{equation}
  \begin{split}
    &|b^{0,(j)}(t,x^0)-b^{0,(j)}(\acute t,\acute x^0)|+|a^{0,(j)}(t,x^0)-a^{0,(j)}(\acute t,\acute x^0)| +|b(t,x^1)-b(\acute t,\acute x^1)|+|a(t,x^1)-a(\acute t,\acute x^1)|\\
    &~~~\leq C\bigl(|t-\acute t|^{\alpha/2}+|x^0-\acute x^0|+|x^1-\acute x^1|\bigr).
  \end{split}
  \end{equation}
  Moreover, $a^{0,(j)}$ and $a$ are bounded and non-degenerate: for some $0<\underline{a}\leq\overline{a}$,
  \begin{equation}
    \underline{a}I_{m_0}\leq
    a^{0,(j)}(t,x^0)a^{0,(j)}(t,x^0)^\top
    \leq \overline{a}I_{m_0},~~~
    \underline{a}I_m\leq
    a(t,x^1)a(t,x^1)^\top
    \leq \overline{a}I_m,
  \end{equation}
  uniformly in $(t,x^0,x^1,j)$.

  \item[(ii)] The default intensity, the market coefficients, and liabilities are given by
  \begin{equation}
    \lambda_t=\lambda(t,X^0_t,N_{t-}),~~~\sigma_t=\sigma(t,X^0_t,N_{t-}),~~~\beta_t=\beta(t,X^0_t,N_{t-}),~~~t\in[0,T],
  \end{equation}
  and
  \begin{equation}
    F^i=\varphi(X^0_T,X^i_T,N_T),
    ~~~ i=1,\ldots,N,
  \end{equation}
  where
  \begin{equation}
    \begin{split}
      &\lambda:[0,T]\times\mathbb{R}^{m_0}\times\{0,1\}\to\mathbb{R}_{+},
      ~~~
      \sigma:[0,T]\times\mathbb{R}^{m_0}\times\{0,1\}\to\mathbb{R}^{n\times d_0},\\
      &\beta:[0,T]\times\mathbb{R}^{m_0}\times\{0,1\}\to\mathbb{R}^n,
      ~~~
      \varphi:\mathbb{R}^{m_0}\times\mathbb{R}^m\times\{0,1\}\to\mathbb{R}
    \end{split}
  \end{equation}
  are bounded and measurable. For $j\in\{0,1\}$, we write $\lambda^{(j)}(t,x^0):=\lambda(t,x^0,j)$, and similarly for $\sigma^{(j)},\beta^{(j)}$, and $\varphi^{(j)}$.

  \item[(iii)] For $j\in\{0,1\}$, $\varphi^{(j)}\in C^1_b(\mathbb{R}^{m_0+m})$ and the maps $\lambda^{(j)}$, $\sigma^{(j)}$, and $\beta^{(j)}$ satisfy the pointwise bounds and non-degeneracy conditions of Assumption \ref{asm:market}: in particular,
  \begin{equation}
    \underline{\lambda}\leq\lambda^{(0)}(t,x^0)\leq\overline{\lambda}, ~~~\lambda^{(1)}(t,x^0)=0,
  \end{equation}
  $c_{\sigma}I_n\leq\sigma^{(j)}(\sigma^{(j)})^\top\leq C_{\sigma}I_n$ and $\beta^{(j),k}>-1$ for every stock component $k$, uniformly in $(t,x^0,j)$.
  Moreover, for the constant $\underline\beta$ in Assumption \ref{asm:market} (iv),
  \begin{equation}
    |\beta^{(0)}(t,x^0)|^2\geq \underline\beta^2,~~~(t,x^0)\in[0,T]\times\mathbb R^{m_0}.
  \end{equation}
  \end{enumerate}
\end{asm}

Under this assumption, we set
\begin{equation}
  \eta^{(j)}(t,x^0):=\sigma^{(j)}(t,x^0)^\top
  \bigl(\sigma^{(j)}(t,x^0)\sigma^{(j)}(t,x^0)^\top\bigr)^{-1}\beta^{(j)}(t,x^0).
\end{equation}
As in Section \ref{sec:preliminary}, the pointwise bounds on $\beta$ and $\sigma$ imply
\begin{equation}
  |\eta^{(j)}(t,x^0)|^2\leq \overline h^2~~~(j=0,1),
  ~~~
  |\eta^{(0)}(t,x^0)|^2\geq \underline h^2,
\end{equation}
uniformly in $(t,x^0)$.
Moreover, for $p\in\mathbb{R}^{1\times d_0}$, set
\begin{equation}
  \Pi[t,x^0,j](p):=p\sigma^{(j)}(t,x^0)^\top\bigl(\sigma^{(j)}(t,x^0)\sigma^{(j)}(t,x^0)^\top\bigr)^{-1}\sigma^{(j)}(t,x^0).
\end{equation}
Then the processes and projections defined in Section \ref{sec:individual-bsde} are represented by
\begin{equation}
  \eta_t=\eta^{(N_{t-})}(t,X^0_t),~~~\Pi_t(p)=\Pi[t,X^0_t,N_{t-}](p),~~~t\in[0,T].
\end{equation}

Set
\begin{equation}
  \mathcal X:=
  \mathbb H^2_{\mathrm{BMO}}(\mathbb{P}^{0,1},\mathbb{F}^{0,1},\mathbb{R}^{1\times d_0})
  \times
  \mathbb H^2_{\mathrm{BMO}}(\mathbb{P}^{0,1},\mathbb{F}^{0,1},\mathbb{R}^{1\times d})
  \times
  \mathbb L^2_{\mathrm{BMO}}(M;\mathbb{P}^{0,1},\mathbb{F}^{0,1},\mathbb{R}),
\end{equation}
with product norm
\begin{equation}
  \|(z^0,z^1,u)\|_{\mathcal X}:=\|z^0\|_{\mathbb H^2_{\mathrm{BMO}}}+\|z^1\|_{\mathbb H^2_{\mathrm{BMO}}}+\|u\|_{\mathbb L^2_{\mathrm{BMO}}(M)}.
\end{equation}
For $R\geq0$, define
\begin{equation}
  \mathcal D_R:=\Bigl\{(z^0,z^1,u)\in\mathcal X: \|(z^0,z^1,u)\|_{\mathcal X}\leq R, ~u=0\text{ on }\{\lambda=0\}\Bigr\}.
\end{equation}

Set $E:=\mathbb R^{m_0}\times\mathbb R^m\times\{0,1\}\times[\underline\gamma,\overline\gamma]$.
For a bounded measurable function $f$, we write $\|f\|_\infty$ for the supremum of $|f|$ over its domain.
For a measurable map $v=(v^0,v^1,v^M):[0,T]\times E\to\mathbb{R}^{1\times d_0}\times\mathbb{R}^{1\times d}\times\mathbb{R}$, we set $\|v\|_{\infty}:=\|v^0\|_\infty+\|v^1\|_\infty+\|v^M\|_\infty$.
For $K\geq0$ we set the Markovian subclass of $\mathcal D_R$ by
\begin{equation}\label{eq:D-Mk-def}
  \begin{split}
  \mathcal{D}^{\mathrm{Mk}}_{R,K}:=
  \Bigl\{(z^0,z^1,u)\in\mathcal{D}_R;&~~
  \text{there exists a measurable map }v\text{ with }\|v\|_{\infty}\leq K,~~
  v^M(t,x^0,x^1,1,\gamma)=0,\text{ such that}\\
  &(z^0_t,z^1_t,u_t)=(v^0(t,X^0_t,X^1_t,N_{t-},\gamma^1),v^1(t,X^0_t,X^1_t,N_{t-},\gamma^1),v^M(t,X^0_t,X^1_t,N_{t-},\gamma^1)).
  \Bigr\}.
  \end{split}
\end{equation}
In particular $|z^0_t|+|z^1_t|+|u_t|\leq\|v\|_\infty\leq K$, $dt\otimes\mathbb P^{0,1}$-a.e. for $(z^0,z^1,u)\in\mathcal{D}^{\mathrm{Mk}}_{R,K}$.

\begin{lem}~\\\label{lem:Dmk-closed}
Let Assumptions \ref{asm:market}, \ref{asm:agent} and \ref{asm:markovian} hold. Then $\mathcal D^{\mathrm{Mk}}_{R,K}$ is closed in $(\mathcal X,\|\cdot\|_{\mathcal X})$, hence complete.
\end{lem}

\noindent\textbf{\textit{Proof}}\\
Consider a sequence $(z^{0,n},z^{1,n},u^n)_{n\in \mathbb{N}}\subset\mathcal D^{\mathrm{Mk}}_{R,K}$ such that $(z^{0,n},z^{1,n},u^n)\to(z^0,z^1,u)\in \mathcal X$ as $n\to\infty$ with respect to $\|\cdot\|_{\mathcal X}$.
Notice that we have $(z^{0,n},z^{1,n},u^n)\to(z^0,z^1,u)$ as $n\to\infty$ in $\mathbb{H}^2\times\mathbb{H}^2\times\mathbb{L}^2(M)$ as well.
We take the representative such that $u=0$ on $\{\lambda=0\}$.
Write $\Theta_t:=(X^0_t,X^1_t,N_{t-},\gamma^1)$, let $\nu_t$ be its law on $E$ for $t\in[0,T]$, and write $\vartheta=(x^0,x^1,j,\gamma)$ for a generic element of $E$.
Take $v_n=(v^0_n,v^1_n,v^M_n)$ such that $\|v_n\|_\infty\leq K$, $v^M_n(t,x^0,x^1,1,\gamma)=0$, and
\begin{equation}
  z^{0,n}_t=v^0_n(t,\Theta_t),~~~z^{1,n}_t=v^1_n(t,\Theta_t),~~~u^n_t=v^M_n(t,\Theta_t),~~~t\in[0,T].
\end{equation}

Then,
\begin{equation}
  \begin{split}
    \int_0^T\int_E |v_n(t,\vartheta)-v_m(t,\vartheta)|^2\nu_t(d\vartheta)dt
    &\leq
    \mathbb{E}\Bigl[\int_0^T(|z^{0,n}_t-z^{0,m}_t|^2 + |z^{1,n}_t-z^{1,m}_t|^2)dt\Bigr] + \frac{1}{\underline{\lambda}}\mathbb{E}\Bigl[\int_0^T\lambda_t|u^n_t-u^m_t|^2dt\Bigr] \\
    &\to 0,~~~m,n\to \infty,
  \end{split}
\end{equation}
where we used $v^M_n(t,x^0,x^1,1,\gamma)=0$ and $\lambda^{(0)}\geq\underline\lambda$.
This implies that $(v_n)_{n\in\mathbb{N}}$ is a Cauchy sequence in $L^2(\nu_t(d\vartheta)dt;\mathbb{R}^{1\times(d_0+d+1)})$, which is complete.
Thus, $(v_n)_{n\in\mathbb{N}}$ has a limit $v$ such that
\begin{equation}
  \lim_{n\to\infty}\int_0^T\int_E |v_n(t,\vartheta)-v(t,\vartheta)|^2\nu_t(d\vartheta)dt =0.
\end{equation}
Passing through an appropriate subsequence (still denoted by $v_n$ for simplicity), we also have $v_n\to v$, $\nu_t(d\vartheta)dt$-a.e.
Since $v^M_n=0$ on $\{j=1\}$ for every $n$, we also have
\begin{equation}
  v^M(t,\vartheta)=0~~~\text{on }\{j=1\},~~~\nu_t(d\vartheta)dt\text{-a.e.}
\end{equation}
Combining the two convergences above gives
\begin{equation}\label{eq:Dmk-rep}
  (z^0_t,z^1_t,u_t)=(v^0(t,\Theta_t),v^1(t,\Theta_t),v^M(t,\Theta_t)),~~~dt\otimes\mathbb{P}^{0,1}\text{-a.e.}
\end{equation}

Modifying $v$ on a $\nu_t(d\vartheta)dt$-null set and setting $v^M(t,x^0,x^1,1,\gamma):=0$ for all $(t,x^0,x^1,\gamma)$, we may assume that
\begin{equation}
  \|v\|_\infty \leq\liminf_{n\to\infty}\|v_n\|_\infty\leq K.
\end{equation}
Thus $v$ is measurable with $\|v\|_\infty\leq K$ and $v^M(\cdot,\cdot,1,\cdot)\equiv0$.
Taking $u_t=v^M(t,\Theta_t)$ as the representative of $u$, we have $u_t=0$ on $\{\lambda_t=0\}=\{N_{t-}=1\}$.
Finally, $\|(z^0,z^1,u)\|_{\mathcal X}\leq\displaystyle\liminf_{n\to\infty}\|(z^{0,n},z^{1,n},u^n)\|_{\mathcal X}\leq R$, so $(z^0,z^1,u)\in\mathcal D_R$.
Together with the measurable map $v$ above, we deduce that $(z^0,z^1,u)\in\mathcal D^{\mathrm{Mk}}_{R,K}$. $\square$

\begin{lem}~\\\label{lem:inner-branch-Mk}
  Let Assumptions \ref{asm:market}, \ref{asm:agent} and \ref{asm:markovian} hold and fix $K\geq0$.
  Take $(z^0,z^1,u)\in\mathcal{D}^{\mathrm{Mk}}_{R,K}$ and let $\psi^*$, $q^*:=q^{\psi^*}$, and $\Phi^*:=\Phi^{\psi^*}$ be as in Lemma \ref{lem:inner-branch}.
  Then, there exist bounded measurable functions
  \begin{equation}
    v^q,v^{\Phi}:[0,T]\times E\to\mathbb{R},
    ~~~
    v^{\psi}:[0,T]\times\mathbb{R}^{m_0}\times\{0,1\}\to\mathbb{R}
  \end{equation}
  satisfying
  \begin{equation}
    \|v^q\|_\infty\leq \bar q_K, ~~~ \|v^\Phi\|_\infty\leq \overline\Phi_K, ~~~ \|v^\psi\|_\infty\leq\Psi_K
  \end{equation}
  such that
  \begin{equation}
    q^{*}_t=v^q(t,X^0_t,X^1_t,N_{t-},\gamma^1),~~~
    \Phi^*_t=v^{\Phi}(t,X^0_t,X^1_t,N_{t-},\gamma^1),~~~
    \psi^{*}_t=v^{\psi}(t,X^0_t,N_{t-}),~~~dt\otimes\mathbb P^{0,1}\text{-a.e.}
  \end{equation}
\end{lem}

\noindent\textbf{\textit{Proof}}\\
Choose a Markovian representative $v=(v^0,v^1,v^M)$ with $\|v\|_{\infty}\leq K$ such that
\begin{equation}
  z^0_t=v^0(t,X^0_t,X^1_t,N_{t-},\gamma^1),~~~
  z^1_t=v^1(t,X^0_t,X^1_t,N_{t-},\gamma^1),~~~
  u_t=v^M(t,X^0_t,X^1_t,N_{t-},\gamma^1),~~~t\in[0,T],
\end{equation}
and define the projected Markovian representative $v^{0\|}:[0,T]\times E\to\mathbb R^{1\times d_0}$ by
\begin{equation}
  v^{0\|}(t,x^0,x^1,j,\gamma) :=\Pi[t,x^0,j]\bigl(v^0(t,x^0,x^1,j,\gamma)\bigr).
\end{equation}
Then $z^{0\|}_t=v^{0\|}(t,X^0_t,X^1_t,N_{t-},\gamma^1)$.  The conditional mean $\bar z^{0\|}=\overline{\mathbb E}[z^{0\|}]$ is represented by
\begin{equation}
  \bar v^{0\|}(t,x^0,j):=\overline{\mathbb E}\bigl[
  v^{0\|}(t,x^0,X^1_t,j,\gamma^1)\bigr],
  ~~~
  \bar z^{0\|}_t=\bar v^{0\|}(t,X^0_t,N_{t-}).
\end{equation}
Since $\Pi[\cdot]$ is an orthogonal projection, $v^{0\|}=\Pi[\cdot]v^0$ is bounded by $K$ and measurable, and so is $\bar v^{0\|}$.

For $(t,x^0,j)\in[0,T]\times\mathbb R^{m_0}\times\{0,1\}$, let $Q^{(j)}_{t,x^0}:\mathbb R\to\mathbb R$ be the inverse of $q\mapsto q-\lambda^{(j)}(t,x^0)|\eta^{(j)}(t,x^0)|^2e^{-q}$, that is,
\begin{equation}
  Q^{(j)}_{t,x^0}(r)-\lambda^{(j)}(t,x^0)|\eta^{(j)}(t,x^0)|^2e^{-Q^{(j)}_{t,x^0}(r)}=r,~~~ r\in\mathbb R.
\end{equation}
Define
\begin{equation}
  \mathcal I_K^{\mathrm{Mk}}:=\left\{\chi:[0,T]\times\mathbb R^{m_0}\times\{0,1\}\to\mathbb R;~\chi \text{ is measurable and } \|\chi\|_\infty\leq\Psi_K \right\}.
\end{equation}
Equipped with the sup norm $\|\cdot\|_\infty$, $\mathcal I_K^{\mathrm{Mk}}$ is complete.
For $\chi\in\mathcal I_K^{\mathrm{Mk}}$, define $q^\chi:[0,T]\times E\to\mathbb R$ by
\begin{equation}
  \begin{split}
  q^\chi(t,x^0,x^1,j,\gamma)&:=Q^{(j)}_{t,x^0}\Bigl( \gamma\bigl(v^{0\|}(t,x^0,x^1,j,\gamma)\eta^{(j)}(t,x^0) -v^M(t,x^0,x^1,j,\gamma)\bigr)\\
  &~~~~~~~~~~~~ -\widehat\gamma\bigl(\bar v^{0\|}(t,x^0,j)\eta^{(j)}(t,x^0) +\lambda^{(j)}(t,x^0)|\eta^{(j)}(t,x^0)|^2\chi(t,x^0,j)\bigr)
  \Bigr).
  \end{split}
\end{equation}
Suppressing the arguments, we equivalently have
\begin{equation}\label{eq:q-chi-implicit}
  q^\chi-\lambda^{(j)}|\eta^{(j)}|^2e^{-q^\chi} = \gamma(v^{0\|}\eta^{(j)}-v^M) -\widehat\gamma\bigl(\bar v^{0\|}\eta^{(j)}+\lambda^{(j)}|\eta^{(j)}|^2\chi\bigr)=:\varrho^\chi.
\end{equation}
Define
\begin{equation}
  \Phi^\chi(t,x^0,x^1,j,\gamma):= \frac{e^{-q^\chi(t,x^0,x^1,j,\gamma)}}{\gamma}, ~~~ (t,x^0,x^1,j,\gamma)\in[0,T]\times E.
\end{equation}
Recall from Section 3.3 that $Q_a$ denotes the inverse of $q\mapsto q-ae^{-q}$.
Since $(a,r)\mapsto Q_a(r)$ is continuous on $[0,\infty)\times\mathbb R$, the definition of $Q^{(j)}_{t,x^0}$ and the measurability of the coefficients imply that $(t,x^0,j,r)\mapsto Q^{(j)}_{t,x^0}(r)$ is measurable.
Hence $q^\chi$ and $\Phi^\chi$ are measurable.
Below we again suppress the common arguments for notational simplicity.
For $\chi\in\mathcal I_K^{\mathrm{Mk}}$, we have $|\varrho^\chi|\leq\bar a_K+b_\psi\Psi_K$ exactly as $|\varrho^\psi_t|$ in the proof of Lemma \ref{lem:inner-branch}.
Hence \eqref{eq:qstar-root-exp-bound} gives 
\begin{equation}
  |q^\chi|\leq\bar a_K+b_\psi\Psi_K+\overline\lambda\overline h^2=\bar q_K,~~~0<\Phi^\chi=\frac{e^{-q^\chi}}{\gamma}\leq \frac{e^{\bar q_K}}{\underline\gamma}=\overline\Phi_K.
\end{equation}
For $\chi,\acute\chi\in\mathcal I_K^{\mathrm{Mk}}$, we have
\begin{equation}
  |\varrho^\chi-\varrho^{\acute\chi}|=\widehat\gamma\lambda^{(j)}|\eta^{(j)}|^2|\chi-\acute\chi|\leq b_\psi|\chi-\acute\chi|.
\end{equation}
Moreover, notice that $q^\chi,q^{\acute\chi}\in[-\bar q_K,\bar q_K]$, the map $r\mapsto Q^{(j)}(r)$ is $1$-Lipschitz by \eqref{eq:qstar-derivative}, and $x\mapsto e^{-x}/\gamma$ is $\overline{\Phi}_K$-Lipschitz on $[-\bar q_K,\bar q_K]$.
Hence
\begin{equation}
  |q^\chi-q^{\acute\chi}|\leq|\varrho^\chi-\varrho^{\acute\chi}|\leq b_\psi|\chi-\acute\chi|,
  ~~~
  |\Phi^\chi-\Phi^{\acute\chi}|\leq\overline{\Phi}_K|q^\chi-q^{\acute\chi}|\leq b_\psi\overline{\Phi}_K|\chi-\acute\chi|.
\end{equation}

For $\ell\in[0,1]$, set $\chi^\ell:=(1-\ell)\acute\chi+\ell\chi$ and $q^\ell:=q^{\chi^\ell}$.
Repeating the calculation leading to \eqref{eq:S-increment}, we obtain
\begin{equation}
  \overline{\mathbb E}\bigl[\Phi^\chi\bigr] -\overline{\mathbb E}\bigl[\Phi^{\acute\chi}\bigr]
  =
  (\chi-\acute\chi)\overline{\mathbb E}\Bigl[\int_0^1\frac{\widehat{\gamma}\lambda|\eta|^2}{\gamma\bigl(e^{q^\ell}+\lambda|\eta|^2\bigr)}d\ell\Bigr],
\end{equation}
and the bound \eqref{eq:S-slope-gap} gives
\begin{equation}\label{eq:SMk-averaged-Lip}
  \Bigl|\overline{\mathbb E}\bigl[\Phi^\chi(t,x^0,X^1_t,j,\gamma^1)\bigr]
  -\overline{\mathbb E}\bigl[\Phi^{\acute\chi}(t,x^0,X^1_t,j,\gamma^1)\bigr]\Bigr|
  \leq(1-\delta_K)\bigl|\chi(t,x^0,j)-\acute\chi(t,x^0,j)\bigr|.
\end{equation}

Define $\mathcal S_{\mathrm{Mk}}:\mathcal I_K^{\mathrm{Mk}}\to\mathcal I_K^{\mathrm{Mk}}$ by
\begin{equation}
  [\mathcal S_{\mathrm{Mk}}(\chi)](t,x^0,j):= \overline{\mathbb E}\bigl[\Phi^\chi(t,x^0,X^1_t,j,\gamma^1)\bigr],~~~(t,x^0,j)\in[0,T]\times\mathbb{R}^{m_0}\times\{0,1\}.
\end{equation}
Note that, if we put $\psi_t=\chi(t,X^0_t,N_{t-})$, then $\mathcal S(\psi)_t=[\mathcal S_{\mathrm{Mk}}(\chi)](t,X^0_t,N_{t-})$, where $\mathcal{S}$ is a map defined in Lemma \ref{lem:inner-branch}.
Following the proof of Lemma \ref{lem:inner-branch}, we can show that $0<[\mathcal S_{\mathrm{Mk}}(\chi)](t,x^0,j)\leq\Psi_K$.
Since $\Phi^\chi$ is bounded and jointly measurable in $(t,x^0,j,\omega^1)$, Fubini's theorem ensures that $(t,x^0,j)\mapsto\overline{\mathbb E}[\Phi^\chi]$ is measurable, so $\mathcal S_{\mathrm{Mk}}(\chi)\in\mathcal I_K^{\mathrm{Mk}}$ for any $\chi\in\mathcal I_K^{\mathrm{Mk}}$.
Moreover, for every $\chi,\acute\chi\in\mathcal I_K^{\mathrm{Mk}}$, taking the supremum over $(t,x^0,j)$ in \eqref{eq:SMk-averaged-Lip} yields
  \begin{equation}
  \|\mathcal S_{\mathrm{Mk}}(\chi)-\mathcal S_{\mathrm{Mk}}(\acute\chi)\|_\infty
  \leq (1-\delta_K)\|\chi-\acute\chi\|_\infty.
  \end{equation}
Therefore $\mathcal S_{\mathrm{Mk}}$ has a unique fixed point $v^\psi\in\mathcal I_K^{\mathrm{Mk}}$.
These observations yield that the fixed point $\psi^*\in \mathcal I_K$ of $\mathcal S$ satisfies $\psi^*_t=v^\psi(t,X^0_t,N_{t-})$.

With $v^q:=q^{v^\psi}$ and $v^\Phi:=\Phi^{v^\psi}$, we have
\begin{equation}
  q^*_t=v^q(t,X^0_t,X^1_t,N_{t-},\gamma^1),
  ~~~
  \Phi^*_t=v^\Phi(t,X^0_t,X^1_t,N_{t-},\gamma^1),
\end{equation}
and these are bounded measurable functions. The pointwise \textit{a priori} bounds $|q^\chi|\leq\bar q_K$ and $0<\Phi^\chi\leq\overline\Phi_K$ established above give the bound estimates for $v^q$ and $v^\Phi$. $\square$

\subsection{PDE representation and gradient estimates}
\label{sec:kernel}

For $j\in\{0,1\}$ and a test function $\phi=\phi(t,x^0,x^1)$, write
\begin{equation}
\begin{split}
  (\mathcal L^{(j)}\phi)(t,x^0,x^1)
  &:= b^{0,(j)}(t,x^0)^\top\nabla_{x^0}\phi(t,x^0,x^1) + b(t,x^1)^\top\nabla_{x^1}\phi(t,x^0,x^1)\\
  &~~~ +\frac{1}{2}\mathrm{tr}\bigl[a^{0,(j)}(t,x^0)a^{0,(j)}(t,x^0)^\top\nabla^2_{x^0}\phi(t,x^0,x^1)\bigr]\\
  &~~~ +\frac{1}{2}\mathrm{tr}\bigl[a(t,x^1)a(t,x^1)^\top\nabla^2_{x^1}\phi(t,x^0,x^1)\bigr].
\end{split}
\end{equation}
When obvious, we omit the arguments.

Given $(z^0,z^1,u)\in\mathcal D^{\mathrm{Mk}}_{R,K}$ with bounded measurable representative $v=(v^0,v^1,v^M)$ satisfying $\|v\|_\infty\leq K$, let $(v^q,v^\Phi,v^\psi)$ be the maps of Lemma \ref{lem:inner-branch-Mk}.
Set $\Theta_t:=(X^0_t,X^1_t,N_{t-},\gamma^1)$, write $v^{0\|}:=\Pi[\cdot]v^0$ and $v^{0\perp}:=v^0-v^{0\|}$ for the parallel and orthogonal parts of $v^0$, and set
\begin{equation}
  \bar v^{0\|}(t,x^0,j):=\overline{\mathbb E}[v^{0\|}(t,x^0,X^1_t,j,\gamma^1)].
\end{equation}
We also put
\begin{equation}
  \mathcal M^{(j)}(t,x^0) := \bar v^{0\|}(t,x^0,j)+\lambda^{(j)}(t,x^0)\eta^{(j)}(t,x^0)^\top v^\psi(t,x^0,j).
\end{equation}
Then, along the factor process,
\begin{equation}
  \bar z^{0\|}_t=\bar v^{0\|}(t,X^0_t,N_{t-}),~~~ \mathcal M_t:=\bar z^{0\|}_t+\lambda_t\eta_t^\top\psi^*_t = \mathcal M^{(N_{t-})}(t,X^0_t),
\end{equation}
and Lemma \ref{lem:inner-branch-Mk} gives
\begin{equation}
  \Phi^*_t=v^\Phi(t,\Theta_t),~~~\psi^*_t=v^\psi(t,X^0_t,N_{t-}).
\end{equation}
Hence the driver \eqref{eq:MF-driver}, evaluated at the candidate $(z^0,z^1,u)\in\mathcal D^{\mathrm{Mk}}_{R,K}$ is represented by
\begin{equation}\label{eq:driver-markovian-rep}
  \begin{split}
  f^{\mathrm{mfg}}(t,z^0_t,z^1_t,u_t)
  &= \widehat{\gamma}v^{0\|}(t,\Theta_t)\mathcal M^{(N_{t-})}(t,X^0_t)^{\top} -\frac{\widehat{\gamma}^2}{2\gamma^1}|\mathcal M^{(N_{t-})}(t,X^0_t)|^2 +\frac{\gamma^1}{2}\bigl(|v^{0\perp}(t,\Theta_t)|^2+|v^1(t,\Theta_t)|^2\bigr)\\
  &~~~-\lambda^{(N_{t-})}(t,X^0_t)v^M(t,\Theta_t) +\lambda^{(N_{t-})}(t,X^0_t)v^\Phi(t,\Theta_t) -\frac{\lambda^{(N_{t-})}(t,X^0_t)}{\gamma^1}\\
  &~~~+\frac{\gamma^1(\lambda^{(N_{t-})}(t,X^0_t))^2|\eta^{(N_{t-})}(t,X^0_t)|^2}{2}\bigl(v^\Phi(t,\Theta_t)\bigr)^2,~~~dt\otimes\mathbb P^{0,1}\text{-a.e.}
  \end{split}
\end{equation}
Replacing $(X^0_t,X^1_t,N_{t-},\gamma^1)$ in \eqref{eq:driver-markovian-rep} by $(x^0,x^1,j,\gamma)$, we can define the source term functions $h^{(j)}$ in regime $j=0,1$ by
\begin{equation}\label{eq:hj-def2}
  \begin{split}
  h^{(j)}(t,x^0,x^1,\gamma)
  &:= \widehat{\gamma}v^{0\|}(t,x^0,x^1,j,\gamma)\mathcal M^{(j)}(t,x^0)^{\top} -\frac{\widehat{\gamma}^2}{2\gamma}|\mathcal M^{(j)}(t,x^0)|^2\\
  &~~~ +\frac{\gamma}{2}\bigl(|v^{0\perp}(t,x^0,x^1,j,\gamma)|^2+|v^1(t,x^0,x^1,j,\gamma)|^2\bigr)\\
  &~~~ -\lambda^{(j)}(t,x^0)v^M(t,x^0,x^1,j,\gamma)+\lambda^{(j)}(t,x^0)v^{\Phi}(t,x^0,x^1,j,\gamma)-\frac{\lambda^{(j)}(t,x^0)}{\gamma}\\
  &~~~ +\frac{\gamma\bigl(\lambda^{(j)}(t,x^0)\bigr)^2|\eta^{(j)}(t,x^0)|^2}{2}\bigl(v^{\Phi}(t,x^0,x^1,j,\gamma)\bigr)^2.
  \end{split}
\end{equation}
By Lemma \ref{lem:inner-branch-Mk} and $\|v\|_\infty\leq K$, we have
\begin{equation}\label{eq:h-bounds}
  \|h^{(j)}\|_\infty\leq C_K,~~~ j\in\{0,1\},
\end{equation}
with $C_K$ independent of $T$.

For $j\in\{0,1\}$ and $(t,x)=(t,x^0,x^1)\in[0,T]\times\mathbb R^{m_0}\times\mathbb R^m$, let $X^{t,x,(j)}=(X^{0,t,x,(j)},X^{1,t,x})$ be the strong solution on $[t,T]$ of
\begin{equation}\label{eq:frozen-factor}
\begin{split}
  X^{0,t,x,(j)}_s&=x^0+\int_t^s b^{0,(j)}(r,X^{0,t,x,(j)}_r)dr+\int_t^s a^{0,(j)}(r,X^{0,t,x,(j)}_r)dW^0_r,~~~s\in[t,T],\\
  X^{1,t,x}_s&=x^1+\int_t^s b(r,X^{1,t,x}_r)dr+\int_t^s a(r,X^{1,t,x}_r)dW^1_r,~~~s\in[t,T].
\end{split}
\end{equation}

We now establish the estimates needed to prove existence of a solution to the mean-field BSDE for a sufficiently short horizon.
Fix $\overline T>0$ and consider $T\in(0,\overline T]$ in the arguments below, with the constants in the assumptions held fixed.
Whenever a generic constant is stated to be independent of $T$, it may depend on $\overline T$ but is uniform over $T\in(0,\overline T]$.

\begin{lem}~\\\label{lem:branchwise-uniform}
Let Assumption \ref{asm:markovian} hold, fix $K\geq0$, and let $h^{(0)}$ and $h^{(1)}$ be the functions given by \eqref{eq:hj-def2}, which satisfy \eqref{eq:h-bounds}.
Moreover, define $w^{(0)}$ and $w^{(1)}$ by
\begin{align}
  w^{(1)}(t,x,\gamma)
  &:=\mathbb E\Bigl[\varphi^{(1)}(X^{t,x,(1)}_T)+\int_t^T h^{(1)}(s,X^{t,x,(1)}_s,\gamma)ds\Bigr],\label{eq:w1FK}\\
  w^{(0)}(t,x,\gamma)
  &:=\mathbb E\Bigl[e^{-\int_t^T\lambda^{(0)}(r,X^{0,t,x,(0)}_r)dr}\varphi^{(0)}(X^{t,x,(0)}_T)\\
  &~~~~~+\int_t^T e^{-\int_t^s\lambda^{(0)}(r,X^{0,t,x,(0)}_r)dr}\bigl(h^{(0)}+\lambda^{(0)} w^{(1)}\bigr)(s,X^{t,x,(0)}_s,\gamma)ds\Bigr],\label{eq:w0FK}
\end{align}
for $(t,x,\gamma)\in[0,T]\times\mathbb{R}^{m_0+m}\times[\underline{\gamma},\overline{\gamma}]$ and write $w=(w^{(0)},w^{(1)})$. 
Here, the process $X^{t,x,(j)}=(X^{0,t,x,(j)},X^{1,t,x})$ is given by \eqref{eq:frozen-factor}.
Then, the following statements hold.
\begin{enumerate}
\item[(i)] $w=(w^{(0)},w^{(1)})$ is bounded, jointly measurable in $(t,x^0,x^1,\gamma)$, satisfies $w^{(j)}(\cdot,\cdot,\gamma)\in W^{1,2}_{p,\mathrm{loc}}\bigl((0,T)\times\mathbb R^{m_0+m}\bigr)$ for every $p\in(1,\infty)$, and solves the system:
\begin{equation}\label{eq:branchwise-pde2}
\begin{aligned}
  \partial_t w^{(1)}+\mathcal L^{(1)}w^{(1)}+h^{(1)}&=0,& w^{(1)}(T,\cdot,\gamma)&=\varphi^{(1)},\\
  \partial_t w^{(0)}+\mathcal L^{(0)}w^{(0)}+h^{(0)}+\lambda^{(0)}(w^{(1)}-w^{(0)})&=0,& w^{(0)}(T,\cdot,\gamma)&=\varphi^{(0)}.
\end{aligned}
\end{equation}

\item[(ii)] For each $j\in\{0,1\}$, $\nabla_xw^{(j)}$ admits a version that is jointly measurable in $(t,x,\gamma)$ and satisfies
\begin{equation}\label{eq:w-bound}
  \|w^{(j)}\|_\infty+\|\nabla_x w^{(j)}\|_\infty\leq C(\|\varphi\|_\infty\vee\|\nabla\varphi\|_\infty)+\varepsilon_K(T),
\end{equation}
where
\begin{equation}
  \|\varphi\|_\infty:=\max_{j\in\{0,1\}}\|\varphi^{(j)}\|_\infty,
  ~~~
  \|\nabla\varphi\|_\infty:=\max_{j\in\{0,1\}}\|\nabla_x\varphi^{(j)}\|_\infty .
\end{equation}
Here $C>0$ is independent of $K$ and $T$, and $\varepsilon_K(T)=O(\sqrt T)$ as $T\to0$ with implicit constant depending on $K$, $\|\varphi\|_\infty$, and the model constants.

\item[(iii)] For each $j\in\{0,1\}$ and $\gamma\in[\underline\gamma,\overline\gamma]$, $w^{(j)}(\cdot,\cdot,\gamma)$ is continuous on $[0,T]\times\mathbb R^{m_0+m}$.
\end{enumerate}
\end{lem}

\noindent\textbf{\textit{Proof}}\\
(i)
From \eqref{eq:w1FK} and \eqref{eq:w0FK}, we have
\begin{equation}
\begin{split}
  \|w^{(1)}\|_\infty&\leq\|\varphi^{(1)}\|_\infty+T\|h^{(1)}\|_\infty<\infty,\\
  \|w^{(0)}\|_\infty&\leq\|\varphi^{(0)}\|_\infty+T\|h^{(0)}\|_\infty+\overline{\lambda}T\|w^{(1)}\|_\infty<\infty,
\end{split}
\end{equation}
as $0\leq\lambda^{(0)}\leq\overline{\lambda}$.
Joint measurability is inherited from the flow.  We fix $\gamma\in[\underline{\gamma},\overline{\gamma}]$ below.

We regularize the data as in Krylov \cite[Section 2.1]{Krylov1980}.
Pick $C^\infty$ functions $\rho_{\mathrm t}\geq0$ and $\rho_{\mathrm x}\geq0$ with $\int\rho_{\mathrm t}=\int\rho_{\mathrm x}=1$, which are supported in $\{|t|\leq1\}$ and $\{|x|\leq1\}$, respectively.
Set the spatial and space-time mollifiers by
\begin{equation}
  \begin{split}
  \rho_{\mathrm x,n}(x)&:=n^{m_0+m}\rho_{\mathrm x}(nx),\\
  \rho_n(t,x)&:=n^{1+m_0+m}\rho_{\mathrm t}(nt)\rho_{\mathrm x}(nx),~~~n\geq 1.
  \end{split}
\end{equation}
Then, $\mathrm{supp}(\rho_{\mathrm x,n})\subset\{x\in\mathbb{R}^{m_0+m};|x|\leq1/n\}$ and $\mathrm{supp}(\rho_n)\subset\{t\in\mathbb{R};|t|\leq1/n\}\times\{x\in\mathbb{R}^{m_0+m};|x|\leq1/n\}$.
For $j\in\{0,1\}$, we set $\varphi^{(j)}_n:=\varphi^{(j)}*\rho_{\mathrm x,n}$.
The standard result for mollification gives $\varphi^{(j)}_n\to\varphi^{(j)}$, $dx$-a.e.

Moreover, when a bounded function $\phi$ on $[0,T]\times\mathbb R^{m_0+m}$ is convolved with $\rho_n$, it is first extended to the domain $\mathbb R\times\mathbb R^{m_0+m}$ by
\begin{equation}
  \phi(t,x):=\phi\bigl(0\vee t\wedge T,x\bigr)~~~(t\in\mathbb R),
\end{equation}
i.e., $\phi(t,x)=\phi(0,x)$ for $t<0$ and $\phi(t,x)=\phi(T,x)$ for $t>T$, which leaves $\|\phi\|_\infty$ unchanged.
For each $t\in(0,T)$, take $n_t\geq 1$ such that $[t-1/n_t,t+1/n_t]\subset[0,T]$.
Since $\rho_n$ is supported in $\{|s|\leq 1/n\}$ in the time variable, $(\phi*\rho_n)(t,x)$ involves $\phi(t-s,\cdot)$ only for $|s|\leq 1/n$, that is, only for times $t-s\in[t-1/n,t+1/n]$.
For $n\geq n_t$, this interval lies in $[0,T]$, so $(\phi*\rho_n)(t,x)$ does not involve the extended values of $\phi$ and is thus unaffected by the choice of extension outside $[0,T]$.

We set $h^{(1)}_n(\cdot,\cdot,\gamma):=h^{(1)}(\cdot,\cdot,\gamma)*\rho_n$ for each fixed $\gamma$.
The standard result for mollification gives $h^{(1)}_n(\cdot,\cdot,\gamma)\to h^{(1)}(\cdot,\cdot,\gamma)$, $dt\otimes dx$-a.e. on $(0,T)\times\mathbb{R}^{m_0+m}$.
The dominated convergence theorem then shows that it also converges in $L^p_{\mathrm{loc}}$ for any $p\geq 1$.
Notice that, by the preceding observation, the limit on $(0,T)$ is independent of the choice of extension.
Set
\begin{equation}
  w^{(1)}_n(t,x,\gamma):=\mathbb E\Bigl[\varphi^{(1)}_n(X^{t,x,(1)}_T)+\int_t^T h^{(1)}_n(s,X^{t,x,(1)}_s,\gamma)ds\Bigr],~~~(t,x)\in[0,T]\times\mathbb{R}^{m_0+m}.
\end{equation}

$\varphi^{(1)}_n$ and $h^{(1)}_n(\cdot,\cdot,\gamma)$ are bounded and locally H\"older continuous in $x$ uniformly in $t$, and the coefficients of $\mathcal L^{(1)}$ satisfy the parabolicity and H\"older conditions by Assumption \ref{asm:markovian}.
Then, Friedman \cite[Chapter 1, Theorem 12]{Friedman1964} (after the standard time reversal) shows that the Cauchy problem
\begin{equation}
  \partial_t u+\mathcal L^{(1)} u+h^{(1)}_n=0,~~~u(T,\cdot)=\varphi^{(1)}_n
\end{equation}
has a classical solution $u\in C^{1,2}$ represented by \cite[Chapter 1, Eq.(7.6)]{Friedman1964}. The estimate \cite[Chapter 1, Eq.(6.12)]{Friedman1964} shows that $u$ is bounded.
Consequently, the Feynman-Kac formula (see, e.g., Pardoux-R\u{a}\c{s}canu \cite[Proposition 3.41]{PardouxRascanu2014}) applies and yields $u=w^{(1)}_n$.
Thus $w^{(1)}_n\in C^{1,2}$ and in particular, $w^{(1)}_n\in W^{1,2}_{p,\mathrm{loc}}$ for $p\in(1,\infty)$.
It solves the Cauchy problem
\begin{equation}\label{eq:branchwise-pde-n}
  \partial_t w^{(1)}_n+\mathcal L^{(1)} w^{(1)}_n+h^{(1)}_n=0,~~~ w^{(1)}_n(T,\cdot)=\varphi^{(1)}_n.
\end{equation}
Since $w^{(1)}_n\in W^{1,2}_{p,\mathrm{loc}}$ with $\sup_n\|w^{(1)}_n\|_\infty<\infty$ and $\sup_n\|h^{(1)}_n\|_\infty\leq C_K$, and the coefficients of $\mathcal L^{(1)}$ are independent of $n$, applying Krylov \cite[Chapter 5, Section 2, Theorem 5, Remark 6]{Krylov2008} gives
\begin{equation}\label{eq:wn-unif}
  \sup_{n\in\mathbb N}\|w^{(1)}_n\|_{W^{1,2}_p(Q')}<\infty
\end{equation}
for every relatively compact open set $Q'\subset(0,T)\times\mathbb R^{m_0+m}$ and every $p\in(1,\infty)$.

Fix a relatively compact open set $Q'\subset(0,T)\times\mathbb R^{m_0+m}$.
Since the law of $X^{t,x,(1)}_r$ is absolutely continuous with respect to the Lebesgue measure for $r>t$, and since $\varphi^{(1)}_n\to\varphi^{(1)}$, $dx$-a.e., and $h^{(1)}_n(\cdot,\cdot,\gamma)\to h^{(1)}(\cdot,\cdot,\gamma)$, $dt\otimes dx$-a.e., the dominated convergence theorem gives $w^{(1)}_n\to w^{(1)}$, $dt\otimes dx$-a.e., hence in $L^p(Q')$ as $w^{(1)}_n$ is uniformly bounded.
By \eqref{eq:wn-unif} and reflexivity of $W^{1,2}_p$ (for $p\in(1,\infty)$), we can extract a subsequence (still denoted by $(w^{(1)}_n)$ for simplicity) converging weakly to $w^{(1)}$ in $W^{1,2}_p(Q')$. See, e.g., Brezis \cite[Theorem 3.18]{Brezis}.
Thus $w^{(1)}(\cdot,\cdot,\gamma)\in W^{1,2}_{p,\mathrm{loc}}$ for every $p\in(1,\infty)$.
In particular, we have
\begin{equation}
  \partial_t w^{(1)}_n\rightharpoonup\partial_t w^{(1)},~~~\nabla_x w^{(1)}_n\rightharpoonup\nabla_x w^{(1)},~~~\nabla^2_x w^{(1)}_n\rightharpoonup\nabla^2_x w^{(1)},~~~\mathcal L^{(1)} w^{(1)}_n\rightharpoonup\mathcal L^{(1)} w^{(1)}
\end{equation}
weakly in $L^p(Q')$ as $n\to\infty$ since the coefficients of $\mathcal L^{(1)}$ are fixed and bounded.
With $h^{(1)}_n\to h^{(1)}$ in $L^p(Q')$, letting $n\to\infty$ in \eqref{eq:branchwise-pde-n} yields
\begin{equation}
  \partial_t w^{(1)}+\mathcal L^{(1)} w^{(1)}+h^{(1)}=0~~~dt\otimes dx\text{-a.e.~on }(0,T)\times\mathbb R^{m_0+m}.
\end{equation}
By \eqref{eq:w1FK} and $\|\nabla\varphi^{(1)}\|_\infty<\infty$,
\begin{equation}
  |w^{(1)}(t,x,\gamma)-\varphi^{(1)}(x)|
  \leq\|\nabla\varphi^{(1)}\|_\infty\mathbb E[|X^{t,x,(1)}_T-x|]+(T-t)\|h^{(1)}\|_\infty.
\end{equation}
As $\mathbb E[|X^{t,x,(1)}_T-x|^2]\leq C(T-t)$ uniformly in $x\in\mathbb R^{m_0+m}$, $w^{(1)}(t,\cdot,\gamma)\to\varphi^{(1)}$ uniformly as $t\nearrow T$.

Since $\|w^{(1)}\|_\infty<\infty$, the term $h^{(0)}+\lambda^{(0)} w^{(1)}$ is bounded.
Put $\lambda^{(0)}_n:=\lambda^{(0)}*\rho_n$, $g^{(0)}:=h^{(0)}+\lambda^{(0)} w^{(1)}$, and $g^{(0)}_n:=g^{(0)}*\rho_n$, and let
\begin{equation}
  w^{(0)}_n(t,x,\gamma)=\mathbb E\Bigl[e^{-\int_t^T\lambda^{(0)}_n(r,X^{0,t,x,(0)}_r)dr}\varphi^{(0)}_n(X^{t,x,(0)}_T)
  +\int_t^T e^{-\int_t^s\lambda^{(0)}_n(r,X^{0,t,x,(0)}_r)dr}g^{(0)}_n(s,X^{t,x,(0)}_s,\gamma)ds\Bigr],
\end{equation}
for $(t,x)\in[0,T]\times\mathbb{R}^{m_0+m}$. A similar argument shows that $w^{(0)}_n\in C^{1,2}$ and is bounded, and it satisfies $\partial_t w^{(0)}_n+\mathcal L^{(0)} w^{(0)}_n-\lambda^{(0)}_n w^{(0)}_n+g^{(0)}_n=0$.
As for $w^{(1)}$, the dominated convergence theorem gives $w^{(0)}_n\to w^{(0)}$, $dt\otimes dx$-a.e., hence in $L^p(Q')$, and the uniform $W^{1,2}_{p,\mathrm{loc}}$-estimate and weak limit give $w^{(0)}(\cdot,\cdot,\gamma)\in W^{1,2}_{p,\mathrm{loc}}$ for every $p\in(1,\infty)$.
Since $\lambda^{(0)}_n$ and $w^{(0)}_n$ are uniformly bounded and converge $dt\otimes dx$-a.e., we also have $\lambda^{(0)}_n w^{(0)}_n\to\lambda^{(0)}w^{(0)}$ in $L^p(Q')$, and $w^{(0)}$ satisfies
\begin{equation}
  \partial_t w^{(0)}+\mathcal L^{(0)} w^{(0)}-\lambda^{(0)} w^{(0)}+h^{(0)}+\lambda^{(0)} w^{(1)}=0~~~dt\otimes dx\text{-a.e.~on }(0,T)\times\mathbb R^{m_0+m}.
\end{equation}
$w^{(0)}(t,\cdot,\gamma)\to\varphi^{(0)}$ follows as for $w^{(1)}$. Thus $w$ solves \eqref{eq:branchwise-pde2}.\\

\noindent
(ii)
For $w^{(1)}$, \eqref{eq:w1FK} gives $\|w^{(1)}\|_\infty\leq\|\varphi^{(1)}\|_\infty+T\|h^{(1)}\|_\infty\leq\|\varphi\|_\infty+C_KT$.
For $w^{(0)}$, we use $e^{-\int_t^s\lambda^{(0)}dr}\leq 1$ and the pathwise identity $\int_t^T\lambda^{(0)} e^{-\int_t^s\lambda^{(0)}dr}ds=1-e^{-\int_t^T\lambda^{(0)}dr}$.
Bounding the integrand of \eqref{eq:w0FK} for each path, we obtain
\begin{equation}
\begin{split}
  |w^{(0)}(t,x,\gamma)|
  &\leq\mathbb E\Bigl[e^{-\int_t^T\lambda^{(0)}dr}\|\varphi^{(0)}\|_\infty+\bigl(1-e^{-\int_t^T\lambda^{(0)}dr}\bigr)\|w^{(1)}\|_\infty+\int_t^T e^{-\int_t^s\lambda^{(0)}dr}ds~\|h^{(0)}\|_\infty\Bigr]\\
  &\leq\|\varphi^{(0)}\|_\infty\vee\|w^{(1)}\|_\infty+T\|h^{(0)}\|_\infty\leq\|\varphi\|_\infty+C_KT,
\end{split}
\end{equation}
where the first two terms form a convex combination of $\|\varphi^{(0)}\|_\infty$ and $\|w^{(1)}\|_\infty$ and hence are bounded by their maximum.
In particular,
\begin{equation}\label{eq:wsup}
  \|w^{(j)}\|_\infty\leq\|\varphi\|_\infty+C_KT~~~(j=0,1).
\end{equation}

Let $p^{(1)}(t,x;s,y)$ be the transition density of $X^{t,x,(1)}_s$ for $s>t$, i.e., $\mathbb E[f(X^{t,x,(1)}_s)]=\int_{\mathbb R^{m_0+m}}f(y)p^{(1)}(t,x;s,y)dy$.
Writing \eqref{eq:w1FK} through $p^{(1)}$ yields the representation
\begin{equation}\label{eq:wj-kernel}
  w^{(1)}(t,x,\gamma)=\int_{\mathbb R^{m_0+m}} p^{(1)}(t,x;T,y)\varphi^{(1)}(y)dy +\int_t^T\int_{\mathbb R^{m_0+m}} p^{(1)}(t,x;s,y)h^{(1)}(s,y,\gamma)dyds.
\end{equation}
See also \cite[Eq.(7.6)]{Friedman1964}.
In the PDE literature, $p^{(1)}$ is the fundamental solution of $\partial_t+\mathcal L^{(1)}$ (after the standard time reversal).
Then, \cite[Chapter 1, Eqs.(6.12)--(6.13)]{Friedman1964} gives the Gaussian estimates
\begin{equation}\label{eq:p-gauss}
\begin{split}
  p^{(1)}(t,x;s,y)&\leq C(s-t)^{-\frac{m_0+m}{2}}\exp\Bigl(-\frac{c|x-y|^2}{s-t}\Bigr),\\
  |\nabla_x p^{(1)}(t,x;s,y)|&\leq C(s-t)^{-\frac{m_0+m+1}{2}}\exp\Bigl(-\frac{c|x-y|^2}{s-t}\Bigr),
\end{split}
\end{equation}
for $0\leq t<s\leq T$ and some $c>0$.
Then, \eqref{eq:p-gauss} yields, uniformly in $x$,
\begin{equation}\label{eq:p-L1}
  \int_{\mathbb R^{m_0+m}}|\nabla_x p^{(1)}(t,x;s,y)|dy\leq C(s-t)^{-1/2},~~~
  \int_{\mathbb R^{m_0+m}}|\nabla_x p^{(1)}(t,x;s,y)||x-y|dy\leq C .
\end{equation}
Since $p^{(1)}$ is a density, $\int_{\mathbb R^{m_0+m}} p^{(1)}(t,x;s,y)dy=1$.
Differentiating in $x$ gives
\begin{equation}
  \int_{\mathbb R^{m_0+m}}\nabla_x p^{(1)}(t,x;s,y)dy=0.
\end{equation}
The interchange of $\int dy$ and $\nabla_x$ is justified by the dominated convergence theorem.
For every compact set $\mathcal K\subset\mathbb R^{m_0+m}$, the bound $|x-y|^2\geq\frac{1}{2}|y|^2-|x|^2$ and \eqref{eq:p-gauss} give
\begin{equation}
  |\nabla_x p^{(1)}(t,x;s,y)|\leq C_{\mathcal K,s}(s-t)^{-(m_0+m+1)/2}e^{-c|y|^2/(2(s-t))}~~~(x\in\mathcal K),
\end{equation}
whose right-hand side is integrable in $y$ and independent of $x$.
Similarly, the first bound of \eqref{eq:p-L1}, being uniform in $x$ and integrable in $s$ over $(t,T)$, justifies the interchange of $\nabla_x$ and $\int dyds$ in \eqref{eq:wj-kernel}.

Differentiating \eqref{eq:wj-kernel} in $x$ then gives
\begin{equation}\label{eq:grad-w1}
  \nabla_x w^{(1)}(t,x,\gamma)=\int_{\mathbb R^{m_0+m}}\nabla_x p^{(1)}(t,x;T,y)\varphi^{(1)}(y)dy+\int_t^T\int_{\mathbb R^{m_0+m}}\nabla_x p^{(1)}(t,x;s,y)h^{(1)}(s,y,\gamma)dyds.
\end{equation}
The right-hand side of \eqref{eq:grad-w1} is jointly measurable in $(t,x,\gamma)$ on $[0,T)\times\mathbb R^{m_0+m}\times[\underline\gamma,\overline\gamma]$ and defines a version of the weak gradient $\nabla_xw^{(1)}$.
At $t=T$, we set $\nabla_xw^{(1)}(T,x,\gamma):=\nabla_x\varphi^{(1)}(x)$ and use this jointly measurable version henceforth.
For the terminal term, we have
\begin{equation}
\begin{split}
  \Bigl|\int_{\mathbb R^{m_0+m}}\nabla_x p^{(1)}(t,x;T,y)\varphi^{(1)}(y)dy\Bigr|
  &=\Bigl|\int_{\mathbb R^{m_0+m}}\nabla_x p^{(1)}(t,x;T,y)\bigl(\varphi^{(1)}(y)-\varphi^{(1)}(x)\bigr)dy\Bigr|\\
  &\leq\|\nabla\varphi^{(1)}\|_\infty\int_{\mathbb R^{m_0+m}}|\nabla_x p^{(1)}(t,x;T,y)||x-y|dy
  \leq C\|\nabla\varphi^{(1)}\|_\infty,
\end{split}
\end{equation}
the last bound by \eqref{eq:p-L1}.  For the source term, \eqref{eq:p-L1} and $\|h^{(1)}\|_\infty\leq C_K$ give
\begin{equation}
  \Bigl|\int_t^T\int_{\mathbb R^{m_0+m}}\nabla_x p^{(1)}(t,x;s,y)h^{(1)}(s,y,\gamma)dyds\Bigr|
  \leq\|h^{(1)}\|_\infty\int_t^T C(s-t)^{-1/2}ds
  \leq C_K\sqrt T.
\end{equation}
Therefore, we have
\begin{equation}
  \|\nabla_x w^{(1)}(t,\cdot,\gamma)\|_\infty\leq C\|\nabla\varphi^{(1)}\|_\infty+C_K\sqrt T.
\end{equation}

For $w^{(0)}$, let $p^{(0)}(t,x;s,y)$ be the transition density of $X^{t,x,(0)}_s$, equivalently the fundamental solution of $\partial_t+\mathcal L^{(0)}$.
Since $\mathcal L^{(0)}$ has the same regularity as $\mathcal L^{(1)}$, $p^{(0)}$ satisfies the same bounds \eqref{eq:p-gauss}--\eqref{eq:p-L1}.
By \eqref{eq:branchwise-pde2}, $w^{(0)}$ solves $\partial_t w^{(0)}+\mathcal L^{(0)} w^{(0)}+g^{(0)}-\lambda^{(0)}w^{(0)}=0$, $w^{(0)}(T,\cdot)=\varphi^{(0)}$, where the source $g^{(0)}-\lambda^{(0)}w^{(0)}=h^{(0)}+\lambda^{(0)}(w^{(1)}-w^{(0)})$ satisfies $\|g^{(0)}-\lambda^{(0)}w^{(0)}\|_\infty\leq C\|\varphi\|_\infty+C_K(1+T)$ by \eqref{eq:wsup}.
Applying \cite[Proposition 3.41]{PardouxRascanu2014} to the equation for $w^{(0)}_n$ in the proof of (i), viewed as a Cauchy problem with source $g^{(0)}_n-\lambda^{(0)}_n w^{(0)}_n$, gives
\begin{equation}
  w^{(0)}_n(t,x,\gamma)=\int_{\mathbb R^{m_0+m}} p^{(0)}(t,x;T,y)\varphi^{(0)}_n(y)dy +\int_t^T\int_{\mathbb R^{m_0+m}} p^{(0)}(t,x;s,y)\bigl(g^{(0)}_n-\lambda^{(0)}_n w^{(0)}_n\bigr)(s,y,\gamma)dyds.
\end{equation} 
Since $\varphi^{(0)}_n$ and $g^{(0)}_n-\lambda^{(0)}_n w^{(0)}_n$ are uniformly bounded and $\varphi^{(0)}_n\to\varphi^{(0)}$ and $g^{(0)}_n-\lambda^{(0)}_n w^{(0)}_n\to g^{(0)}-\lambda^{(0)}w^{(0)}$, $dt\otimes dx$-a.e., the dominated convergence theorem yields
\begin{equation}\label{eq:w0-kernel}
  w^{(0)}(t,x,\gamma)=\int_{\mathbb R^{m_0+m}} p^{(0)}(t,x;T,y)\varphi^{(0)}(y)dy +\int_t^T\int_{\mathbb R^{m_0+m}} p^{(0)}(t,x;s,y)\bigl(g^{(0)}-\lambda^{(0)}w^{(0)}\bigr)(s,y,\gamma)dyds.
\end{equation}
Differentiating \eqref{eq:w0-kernel} in the same way defines a jointly measurable version of $\nabla_xw^{(0)}$ on $[0,T)\times\mathbb R^{m_0+m}\times[\underline\gamma,\overline\gamma]$.
We extend it to $t=T$ by setting $\nabla_xw^{(0)}(T,x,\gamma):=\nabla_x\varphi^{(0)}(x)$ and use this version henceforth.
Repeating the preceding gradient estimate gives
\begin{equation}
  \|\nabla_x w^{(0)}(t,\cdot,\gamma)\|_\infty\leq C\|\nabla\varphi^{(0)}\|_\infty+\bigl(C\|\varphi\|_\infty+C_K(1+T)\bigr)\sqrt T .
\end{equation}
Together with the $w^{(1)}$ estimate and \eqref{eq:wsup}, the $T$-dependent terms can be collected into a function $\varepsilon_K(T)$ satisfying $\varepsilon_K(T)=O(\sqrt T)$ as $T\to0$.
Thus the chosen jointly measurable versions satisfy \eqref{eq:w-bound}.\\

\noindent
(iii)
It remains to prove the continuity of $w^{(j)}(\cdot,\cdot,\gamma)$.
The fundamental solution $p^{(j)}(t,x;s,y)$ is jointly continuous in $(t,x)$ for $t<s$ (see \cite[Chapter 1, Sections 4 and 6]{Friedman1964}) and bounded by \eqref{eq:p-gauss}.
Hence the representations \eqref{eq:wj-kernel} and \eqref{eq:w0-kernel}, whose terminal condition and source terms are bounded, and the dominated convergence theorem show that $w^{(j)}(\cdot,\cdot,\gamma)$ is continuous on $[0,T)\times\mathbb R^{m_0+m}$.
The continuity extends to $t=T$ since $w^{(j)}(t,\cdot,\gamma)\to\varphi^{(j)}$ locally uniformly as $t\nearrow T$ and $\varphi^{(j)}$ is continuous. $\square$

\begin{lem}~\\\label{lem:mk-closure2}
Let Assumptions \ref{asm:market}, \ref{asm:agent} and \ref{asm:markovian} hold and fix $K\geq0$. Then, the following statements hold.
\begin{enumerate}
\item[(i)] For $(z^0,z^1,u)\in\mathcal D^{\mathrm{Mk}}_{R,K}$ the BSDE
\begin{equation}\label{eq:Gamma-affine-BSDE}
  Y_t=F+\int_t^T f^{\mathrm{mfg}}(s,z^0_s,z^1_s,u_s)ds -\int_t^T Z^0_sdW^0_s-\int_t^T Z^1_sdW^1_s-\int_t^T U_sdM_s,~~~t\in[0,T]
\end{equation}
has a unique solution $(Y,Z^0,Z^1,U)\in\mathbb S^\infty\times\mathcal X$, which is given by
\begin{equation}
  Y_t=w^{(N_t)}(t,X^0_t,X^1_t,\gamma^1),~~~t\in[0,T],
\end{equation}
\begin{equation}
  Z^0_t=\nabla_{x^0}w^{(N_{t-})}(t,X^0_t,X^1_t,\gamma^1)a^0(t,X^0_t,N_{t-}),~~~
  Z^1_t=\nabla_{x^1}w^{(N_{t-})}(t,X^0_t,X^1_t,\gamma^1)a(t,X^1_t),
\end{equation}
\begin{equation}
  U_t=(w^{(1)}-w^{(0)})(t,X^0_t,X^1_t,\gamma^1)\mathbf{1}_{\{N_{t-}=0\}},
\end{equation}
$dt\otimes\mathbb{P}^{0,1}$-a.e. Here, $w^{(0)}$ and $w^{(1)}$ are the functions defined by \eqref{eq:w1FK}--\eqref{eq:w0FK} in Lemma \ref{lem:branchwise-uniform}.

\item[(ii)] There exists $K_{\mathrm{Mk}}>0$ such that, for $K\geq K_{\mathrm{Mk}}$, there is $T_0=T_0(K)>0$ with the following property.
As long as $T\leq T_0$, the solution of \eqref{eq:Gamma-affine-BSDE} satisfies
\begin{equation}
  \|(Z^0,Z^1,U)\|_{\mathcal X}\leq(1+\sqrt{\overline\lambda})K\sqrt T=:R,
\end{equation}
and $(Z^0,Z^1,U)\in\mathcal D^{\mathrm{Mk}}_{R,K}$.
\end{enumerate}
\end{lem}

\noindent\textbf{\textit{Proof}}\\
(i) For $(z^0,z^1,u)\in\mathcal D^{\mathrm{Mk}}_{R,K}$, the driver $f^{\mathrm{mfg}}(t,z^0_t,z^1_t,u_t)$ is bounded by Lemma \ref{lem:inner-branch}. Then,
\begin{equation}
  Y_t=\mathbb E\Bigl[F+\int_t^T f^{\mathrm{mfg}}(s,z^0_s,z^1_s,u_s)ds\Bigm|\mathcal F^{0,1}_t\Bigr],~~~t\in[0,T]
\end{equation}
is the unique bounded solution. Applying It\^o formula to $|Y|^2$ gives, for $\tau_0\in\mathcal T^{0,1}$,
\begin{equation}
  \mathbb E\Bigl[\int_{\tau_0}^T(|Z^0_s|^2+|Z^1_s|^2+\lambda_sU_s^2)ds\Bigm|\mathcal F^{0,1}_{\tau_0}\Bigr] \leq C,
\end{equation}
so $(Z^0,Z^1,U)\in\mathcal X$.
Write $X_t:=(X^0_t,X^1_t)$ and $\Theta_t=(X^0_t,X^1_t,N_{t-},\gamma^1)$ for $t\in[0,T]$.
By definition of $\mathcal D^{\mathrm{Mk}}_{R,K}$, there exist $(v^0,v^1,v^M)$ such that $z^0_t=v^0(t,\Theta_t)$, $z^1_t=v^1(t,\Theta_t)$, $u_t=v^M(t,\Theta_t)$.
Then, define $h^{(j)}$ by \eqref{eq:hj-def2} for $j\in\{0,1\}$, that is,
\begin{equation}
  h^{(N_{t-})}(t,X^0_t,X^1_t,\gamma^1)=f^{\mathrm{mfg}}(t,z^0_t,z^1_t,u_t) ,~~~t\in[0,T].
\end{equation}

Let $w^{(0)}$ and $w^{(1)}$ be the functions defined by \eqref{eq:w1FK}--\eqref{eq:w0FK} with this $h^{(j)}$ ($j=0,1$). By Lemma \ref{lem:branchwise-uniform} they are bounded, belong to $W^{1,2}_{p,\mathrm{loc}}$, and solve \eqref{eq:branchwise-pde2}.
Then we have
\begin{equation}
\begin{split}
  w^{(N_t)}(t,X_t,\gamma^1)
  &=w^{(0)}(t,X_t,\gamma^1)\mathbf{1}_{\{t<\tau\}}+w^{(1)}(t,X_t,\gamma^1)\mathbf{1}_{\{\tau\leq t\}}\\
  &=w^{(0)}(0,X_0,\gamma^1)+\{w^{(0)}(t\wedge\tau,X_{t\wedge\tau},\gamma^1)-w^{(0)}(0,X_0,\gamma^1)\}\\
  &~~~+\{w^{(1)}(t,X_t,\gamma^1)-w^{(1)}(t\wedge\tau,X_{t\wedge\tau},\gamma^1)\}
  +(w^{(1)}-w^{(0)})(\tau,X_\tau,\gamma^1)\mathbf{1}_{\{\tau\leq t\}},
\end{split}
\end{equation}
where the second equality follows by separating the events $\{t<\tau\}$ and $\{\tau\leq t\}$.

Fix $\gamma\in[\underline\gamma,\overline\gamma]$.
By Assumption \ref{asm:markovian}, the coefficients of $X$ are progressively measurable, bounded, and uniformly non-degenerate, so the It\^o-Krylov formula \cite[Chapter 2, Section 10, Theorem 1, Remark 6]{Krylov1980} gives
\begin{equation}
\begin{split}
  w^{(0)}(t\wedge\tau,X_{t\wedge\tau},\gamma)-w^{(0)}(0,X_0,\gamma)
  &=
  \int_0^t\mathbf{1}_{\{s\leq\tau\}}\{\partial_sw^{(0)}+\mathcal L^{(0)}w^{(0)}\}(s,X_s,\gamma)ds\\
  &~~~+
  \int_0^t\mathbf{1}_{\{s\leq\tau\}}\nabla_{x^0}w^{(0)}(s,X_s,\gamma)a^{0,(0)}(s,X^0_s)dW^0_s\\
  &~~~+
  \int_0^t\mathbf{1}_{\{s\leq\tau\}}\nabla_{x^1}w^{(0)}(s,X_s,\gamma)a(s,X^1_s)dW^1_s,
\end{split}
\end{equation}
\begin{equation}
\begin{split}
  w^{(1)}(t,X_t,\gamma)-w^{(1)}(t\wedge\tau,X_{t\wedge\tau},\gamma)
  &=
  \int_0^t\mathbf{1}_{\{\tau<s\}}\{\partial_sw^{(1)}+\mathcal L^{(1)}w^{(1)}\}(s,X_s,\gamma)ds\\
  &~~~+
  \int_0^t\mathbf{1}_{\{\tau<s\}}\nabla_{x^0}w^{(1)}(s,X_s,\gamma)a^{0,(1)}(s,X^0_s)dW^0_s\\
  &~~~+
  \int_0^t\mathbf{1}_{\{\tau<s\}}\nabla_{x^1}w^{(1)}(s,X_s,\gamma)a(s,X^1_s)dW^1_s,
\end{split}
\end{equation}
for $t\in[0,T]$.

More precisely, since the cited formula is stated on bounded cylinders, we obtain these two identities by localization.
Set $0<t_0<t_1<T$ and, for $r>0$, $\tau_r:=\inf\{t\in[0,T]:|X_t|\geq r\}\wedge t_1$.
Then, for $t\in[t_0,t_1]$,
\begin{equation}
\begin{split}
  w^{(0)}((t\wedge\tau\wedge\tau_r)\vee t_0,X_{(t\wedge\tau\wedge\tau_r)\vee t_0},\gamma)-w^{(0)}(t_0,X_{t_0},\gamma)
  &=\int_{t_0}^{t}\mathbf{1}_{\{s\leq\tau\wedge\tau_r\}}\{\partial_sw^{(0)}+\mathcal L^{(0)}w^{(0)}\}(s,X_s,\gamma)ds\\
  &~~~+\int_{t_0}^{t}\mathbf{1}_{\{s\leq\tau\wedge\tau_r\}}\nabla_{x^0}w^{(0)}(s,X_s,\gamma)a^{0,(0)}(s,X^0_s)dW^0_s\\
  &~~~+\int_{t_0}^{t}\mathbf{1}_{\{s\leq\tau\wedge\tau_r\}}\nabla_{x^1}w^{(0)}(s,X_s,\gamma)a(s,X^1_s)dW^1_s,
\end{split}
\end{equation}
and similarly for $w^{(1)}$ with the indicator $\mathbf{1}_{\{\tau<s\leq\tau_r\}}$.
We first let $r\to\infty$, then $t_0\downarrow0$, and finally $t_1\uparrow T$, which yields the identities above.
The drift terms converge by dominated convergence, as they are bounded by \eqref{eq:branchwise-pde2} and \eqref{eq:h-bounds}.
The stochastic integrals converge in $\mathbb{L}^2$ as their integrands are bounded.
The left-hand sides converge by the continuity of $w^{(j)}$ in Lemma \ref{lem:branchwise-uniform} (iii) and the continuity of the paths of $X$.

Since $\gamma^1$ is initial and independent of the driving noises, conditioning on $\gamma^1$ yields the same identities with $\gamma=\gamma^1$.
Notice that
\begin{equation}
  \int_0^t(w^{(1)}-w^{(0)})(s,X_s,\gamma^1)dN_s=(w^{(1)}-w^{(0)})(\tau,X_\tau,\gamma^1)\mathbf{1}_{\{\tau\leq t\}},~~~t\in[0,T],
\end{equation}
and hence
\begin{equation}
\begin{split}
  &w^{(N_t)}(t,X_t,\gamma^1)-w^{(0)}(0,X_0,\gamma^1)\\
  &=
  \{w^{(0)}(t\wedge\tau,X_{t\wedge\tau},\gamma^1)-w^{(0)}(0,X_0,\gamma^1)\}+\{w^{(1)}(t,X_t,\gamma^1)-w^{(1)}(t\wedge\tau,X_{t\wedge\tau},\gamma^1)\}\\
  &~~~+(w^{(1)}-w^{(0)})(\tau,X_\tau,\gamma^1)\mathbf{1}_{\{\tau\leq t\}}\\
  &=
  \int_0^t\Bigl[\mathbf{1}_{\{s\leq\tau\}}\{\partial_sw^{(0)}+\mathcal L^{(0)}w^{(0)}\}(s,X_s,\gamma^1) + \mathbf{1}_{\{\tau<s\}}\{\partial_sw^{(1)}+\mathcal L^{(1)}w^{(1)}\}(s,X_s,\gamma^1)\Bigr]ds\\
  &~~~+
  \int_0^t\Bigl[\mathbf{1}_{\{s\leq\tau\}}\nabla_{x^0}w^{(0)}(s,X_s,\gamma^1)a^{0,(0)}(s,X^0_s) + \mathbf{1}_{\{\tau<s\}}\nabla_{x^0}w^{(1)}(s,X_s,\gamma^1)a^{0,(1)}(s,X^0_s)\Bigr]dW^0_s\\
  &~~~+
  \int_0^t\Bigl[\mathbf{1}_{\{s\leq\tau\}}\nabla_{x^1}w^{(0)}(s,X_s,\gamma^1)a(s,X^1_s) + \mathbf{1}_{\{\tau<s\}}\nabla_{x^1}w^{(1)}(s,X_s,\gamma^1)a(s,X^1_s)\Bigr]dW^1_s\\
  &~~~+\int_0^t(w^{(1)}-w^{(0)})(s,X_s,\gamma^1)dN_s\\
  &=\int_0^t\{\partial_sw^{(N_{s-})}+\mathcal L^{(N_{s-})}w^{(N_{s-})}\}(s,X_s,\gamma^1)ds+\int_0^t\nabla_{x^0}w^{(N_{s-})}(s,X_s,\gamma^1)a^0(s,X^0_s,N_{s-})dW^0_s\\
  &~~~+\int_0^t\nabla_{x^1}w^{(N_{s-})}(s,X_s,\gamma^1)a(s,X^1_s)dW^1_s+\int_0^t(w^{(1)}-w^{(0)})(s,X_s,\gamma^1)dN_s,
\end{split}
\end{equation}
where the last equality uses $\mathbf{1}_{\{s\leq\tau\}}=\mathbf{1}_{\{N_{s-}=0\}}$ and $\mathbf{1}_{\{\tau<s\}}=\mathbf{1}_{\{N_{s-}=1\}}$.
By \eqref{eq:branchwise-pde2}, the following identity holds $ds\otimes\mathbb P^{0,1}$-a.e.
\begin{equation}
  \{\partial_sw^{(N_{s-})}+\mathcal L^{(N_{s-})}w^{(N_{s-})}\}(s,X_s,\gamma^1)
  =-h^{(N_{s-})}(s,X^0_s,X^1_s,\gamma^1)-\lambda^{(0)}(s,X^0_s)(w^{(1)}-w^{(0)})(s,X_s,\gamma^1)\mathbf{1}_{\{N_{s-}=0\}}.
\end{equation}
Together with $dN_s=dM_s+\lambda^{(0)}(s,X^0_s)\mathbf{1}_{\{N_{s-}=0\}}ds$, we obtain
\begin{equation}
\begin{split}
  &\int_0^t\{\partial_sw^{(N_{s-})}+\mathcal L^{(N_{s-})}w^{(N_{s-})}\}(s,X_s,\gamma^1)ds + \int_0^t(w^{(1)}-w^{(0)})(s,X_s,\gamma^1)dN_s\\
  &=
  -\int_0^th^{(N_{s-})}(s,X^0_s,X^1_s,\gamma^1)ds + \int_0^t(w^{(1)}-w^{(0)})(s,X_s,\gamma^1)dM_s.
\end{split}
\end{equation}

Combining the above, we have
\begin{equation}
\begin{split}
  &w^{(N_t)}(t,X_t,\gamma^1)\\
  &=
  w^{(N_T)}(T,X_T,\gamma^1) + \int_t^T h^{(N_{s-})}(s,X^0_s,X^1_s,\gamma^1)ds - \int_t^T \nabla_{x^0}w^{(N_{s-})}(s,X_s,\gamma^1)a^0(s,X^0_s,N_{s-})dW^0_s\\
  &~~~-\int_t^T \nabla_{x^1}w^{(N_{s-})}(s,X_s,\gamma^1)a(s,X^1_s)dW^1_s-\int_t^T (w^{(1)}-w^{(0)})(s,X_s,\gamma^1)dM_s\\
  &=
  \varphi^{(N_T)}(X^0_T,X^1_T) + \int_t^T f^{\mathrm{mfg}}(s,z^0_s,z^1_s,u_s)ds - \int_t^T \nabla_{x^0}w^{(N_{s-})}(s,X_s,\gamma^1)a^0(s,X^0_s,N_{s-})dW^0_s\\
  &~~~-\int_t^T \nabla_{x^1}w^{(N_{s-})}(s,X_s,\gamma^1)a(s,X^1_s)dW^1_s-\int_t^T (w^{(1)}-w^{(0)})(s,X_s,\gamma^1)\mathbf{1}_{\{N_{s-}=0\}}dM_s,
\end{split}
\end{equation}
where the second equality uses the terminal condition in \eqref{eq:branchwise-pde2} and the identity $dN_s=\mathbf{1}_{\{N_{s-}=0\}}dN_s$.
Since Assumption \ref{asm:markovian} (ii) gives $\varphi^{(N_T)}(X^0_T,X^1_T)=F$, we deduce that this coincides with the BSDE \eqref{eq:Gamma-affine-BSDE}.
The uniqueness of the solution of \eqref{eq:Gamma-affine-BSDE} shown above identifies these Markovian processes with $(Y,Z^0,Z^1,U)$. \\

\noindent
(ii) With the jointly measurable bounded representatives of $\nabla_x w^{(j)}$ ($j=0,1$) from Lemma \ref{lem:branchwise-uniform} (ii), define the function $(V^0,V^1,V^M):[0,T]\times E\to\mathbb{R}^{1\times d_0}\times\mathbb{R}^{1\times d}\times\mathbb{R}$ by
\begin{equation}
  V^0(t,x^0,x^1,j,\gamma):=(\nabla_{x^0}w^{(j)})(t,x^0,x^1,\gamma)a^0(t,x^0,j),~~~
  V^1(t,x^0,x^1,j,\gamma):=(\nabla_{x^1}w^{(j)})(t,x^0,x^1,\gamma)a(t,x^1),
\end{equation}
\begin{equation}
  V^M(t,x^0,x^1,j,\gamma):=(w^{(1)}-w^{(0)})(t,x^0,x^1,\gamma)\mathbf{1}_{\{j=0\}},
\end{equation}
so that we have $Z^0_t=V^0(t,\Theta_t)$, $Z^1_t=V^1(t,\Theta_t)$ and $U_t=V^M(t,\Theta_t)$, $dt\otimes\mathbb{P}^{0,1}$-a.e. by (i).
By Lemma \ref{lem:branchwise-uniform},
\begin{equation}\label{eq:V-bound}
  \|V^0\|_\infty+\|V^1\|_\infty+\|V^M\|_\infty\leq C(\|\varphi\|_\infty\vee\|\nabla\varphi\|_\infty)+\varepsilon_K(T).
\end{equation}
Take $K_{\mathrm{Mk}}:=1+2C(\|\varphi\|_\infty\vee\|\nabla\varphi\|_\infty)$.
For $K\geq K_{\mathrm{Mk}}$, there exists $T_0:=T_0(K)>0$ such that $\varepsilon_K(T)<\dfrac{K}{2}$ for $T\leq T_0$ since $\varepsilon_K(T)\to0$ as $T\to0$. Then, $\|(V^0,V^1,V^M)\|_\infty\leq K$.
Since $(Z^0,Z^1,U)$ are bounded by $\|(V^0,V^1,V^M)\|_\infty\leq K$ and $\lambda_s\leq\overline\lambda$,
\begin{equation}
  \|(Z^0,Z^1,U)\|_{\mathcal X}
  \leq\bigl(\|V^0\|_\infty+\|V^1\|_\infty\bigr)\sqrt T+\sqrt{\overline\lambda}\|V^M\|_\infty\sqrt T
  \leq(1+\sqrt{\overline\lambda})K\sqrt T=R.
\end{equation}
Finally, $V^M(t,x^0,x^1,1,\gamma)=0$ gives $U_t=0$ on $\{\lambda_t=0\}=\{N_{t-}=1\}$, where the equality of these two events holds by Assumption \ref{asm:markovian} (iii). 
Hence $(Z^0,Z^1,U)\in\mathcal D^{\mathrm{Mk}}_{R,K}$. $\square$

\subsection{Fixed point construction for a short horizon}
\label{sec:MF-existence-small}

We now prove that the mean-field BSDE \eqref{eq:MF-BSDE}--\eqref{eq:MF-driver} has a solution.
By Lemma \ref{lem:mk-closure2} (i), we can define a map $\Gamma:\mathcal D^{\mathrm{Mk}}_{R,K}\to\mathcal X$ by $\Gamma(z^0,z^1,u):=(Z^0,Z^1,U)$, where $(Z^0,Z^1,U)$ is the solution of the BSDE
\begin{equation}
  Y_t=F+\int_t^T f^{\mathrm{mfg}}(s,z^0_s,z^1_s,u_s)ds -\int_t^T Z^0_sdW^0_s-\int_t^T Z^1_sdW^1_s-\int_t^T U_sdM_s,~~~t\in[0,T],
\end{equation}
with $U=0$ on $\{\lambda=0\}$.

This is the first main result of this paper.

\begin{thm}[Existence of a solution to the mean-field BSDE]~\\\label{thm:MF-exist}
  Let Assumptions \ref{asm:market}, \ref{asm:agent} and \ref{asm:markovian} be in force, and fix $K\geq K_{\mathrm{Mk}}$, where $K_{\mathrm{Mk}}>0$ is the constant in Lemma \ref{lem:mk-closure2}.
  Then, there exists a constant $T_0>0$ such that, whenever $T\leq T_0$, the mean-field BSDE \eqref{eq:MF-BSDE}--\eqref{eq:MF-driver} admits a solution
    $(Y,Z^0,Z^1,U)
    \in
    \mathbb{S}^{\infty}(\mathbb{P}^{0,1},\mathbb{F}^{0,1},\mathbb{R})
    \times\mathbb{H}^2_{\mathrm{BMO}}(\mathbb{P}^{0,1},\mathbb{F}^{0,1},\mathbb{R}^{1\times d_0})
    \times\mathbb{H}^2_{\mathrm{BMO}}(\mathbb{P}^{0,1},\mathbb{F}^{0,1},\mathbb{R}^{1\times d})
    \times\mathbb{L}^2_{\mathrm{BMO}}(M;\mathbb{P}^{0,1},\mathbb{F}^{0,1},\mathbb{R})$
  satisfying $(Z^0,Z^1,U)\in\mathcal{D}^{\mathrm{Mk}}_{R,K}$ with $R:=(1+\sqrt{\overline{\lambda}})K\sqrt T$.
\end{thm}

\noindent\textbf{\textit{Proof}}\\
We show that $\Gamma$ is a contraction on $\mathcal D^{\mathrm{Mk}}_{R,K}$ in $\|\cdot\|_{\mathcal X}$.
Throughout the proof, $C$ denotes a constant independent of $T$, $R$, and $K$, and $C_K$ one that may depend on $K$ but not on $T$ and $R$.
For $(z^0,z^1,u)\in\mathcal{D}^{\mathrm{Mk}}_{R,K}$ and $T\leq T_0$, Lemma \ref{lem:mk-closure2} (ii) gives $\Gamma(z^0,z^1,u)\in\mathcal D^{\mathrm{Mk}}_{R,K}$, hence $\Gamma(\mathcal{D}^{\mathrm{Mk}}_{R,K})\subset\mathcal{D}^{\mathrm{Mk}}_{R,K}$.
Take $(z^0,z^1,u),(\widehat z^0,\widehat z^1,\widehat u)\in\mathcal{D}^{\mathrm{Mk}}_{R,K}$, let $(Z^0,Z^1,U):=\Gamma(z^0,z^1,u)$, $(\widehat Z^0,\widehat Z^1,\widehat U):=\Gamma(\widehat z^0,\widehat z^1,\widehat u)$, and $Y,\widehat Y$ be the corresponding solutions to \eqref{eq:Gamma-affine-BSDE}.
Denote the associated processes of Lemma \ref{lem:inner-branch} by $(\psi^{*},\Phi,\mathcal M)$ and $(\widehat\psi^{*},\widehat\Phi,\widehat{\mathcal M})$.
Write $\Delta z^0:=z^0-\widehat z^0$, $\Delta z^1:=z^1-\widehat z^1$, $\Delta u:=u-\widehat u$, $\Delta\Phi:=\Phi-\widehat\Phi$ and $\Delta\mathcal M:=\mathcal M-\widehat{\mathcal M}$.

By the explicit form of $f^{\mathrm{mfg}}$ in \eqref{eq:MF-driver}, we have
\begin{equation}\label{eq:delta-f-pointwise}
\begin{split}
  \Bigl|f^{\mathrm{mfg}}(s,z^0_s,z^1_s,u_s) -f^{\mathrm{mfg}}(s,\widehat z^0_s,\widehat z^1_s,\widehat u_s)\Bigr|
  &\leq
  C\Bigl(|\Delta z^{0\|}_s||\mathcal M_s|+|\widehat z^{0\|}_s||\Delta\mathcal M_s| +(|\mathcal M_s|+|\widehat{\mathcal M}_s|)|\Delta\mathcal M_s|\\
  &~~~~~
  +(|z^{0\perp}_s|+|\widehat z^{0\perp}_s|)|\Delta z^{0\perp}_s| +(|z^1_s|+|\widehat z^1_s|)|\Delta z^1_s|\\
  &~~~~~
  +|\Delta u_s| +(1 + |\Phi_s|+|\widehat\Phi_s|)|\Delta\Phi_s|
  \Bigr),~~~ds\otimes\mathbb P^{0,1}\text{-a.e.}
\end{split}
\end{equation}
Moreover, boundedness of $\psi^*,\widehat\psi^*$ and \eqref{eq:delta-scalar-Lip} give
\begin{equation}\label{eq:delta-f-aux-bmo}
\begin{split}
  &\|\mathcal M\|_{\mathbb H^2_{\mathrm{BMO}}} +\|\widehat{\mathcal M}\|_{\mathbb H^2_{\mathrm{BMO}}}
  \leq
  \|z^0\|_{\mathbb H^2_{\mathrm{BMO}}} +\|\widehat z^0\|_{\mathbb H^2_{\mathrm{BMO}}} +C\bigl(\|\psi^*\|_{\mathbb H^2_{\mathrm{BMO}}} +\|\widehat\psi^*\|_{\mathbb H^2_{\mathrm{BMO}}}\bigr)
  \leq
  2R+C_K\sqrt T,\\
  &\|\Phi\|_{\mathbb H^2_{\mathrm{BMO}}} +\|\widehat\Phi\|_{\mathbb H^2_{\mathrm{BMO}}}
  \leq
  C_K\sqrt T,\\
  &\|\Delta\mathcal M\|_{\mathbb H^2_{\mathrm{BMO}}} +\|\Delta\Phi\|_{\mathbb H^2_{\mathrm{BMO}}}
  \leq
  C_K\|(\Delta z^0,\Delta z^1,\Delta u)\|_{\mathcal X}.
\end{split}
\end{equation}
Thus, the Cauchy-Schwarz inequality yields
\begin{equation}\label{eq:delta-f-bound}
\begin{split}
  &\mathbb E\Bigl[\int_{\tau_0}^T \Bigl|f^{\mathrm{mfg}}(s,z^0_s,z^1_s,u_s) -f^{\mathrm{mfg}}(s,\widehat z^0_s,\widehat z^1_s,\widehat u_s)\Bigr|ds \Bigm|\mathcal F^{0,1}_{\tau_0} \Bigr]\\
  &\leq
  C_K(R+\sqrt T)
  \bigl(
    \|\Delta z^0\|_{\mathbb H^2_{\mathrm{BMO}}}
    +\|\Delta z^1\|_{\mathbb H^2_{\mathrm{BMO}}}
    +\|\Delta\mathcal M\|_{\mathbb H^2_{\mathrm{BMO}}}
    +\|\Delta\Phi\|_{\mathbb H^2_{\mathrm{BMO}}}
  \bigr)
  +C\mathbb E\Bigl[\int_{\tau_0}^T |\Delta u_s|ds\Bigm|\mathcal F^{0,1}_{\tau_0}\Bigr]\\
  &~~~+C_K\mathbb E\Bigl[\int_{\tau_0}^T |\Delta\Phi_s|ds\Bigm|\mathcal F^{0,1}_{\tau_0} \Bigr]\\
  &\leq
  C_K(R+\sqrt T)\|(\Delta z^0,\Delta z^1,\Delta u)\|_{\mathcal X}+C\frac{\sqrt{T}}{\sqrt{\underline\lambda}}\|\Delta u\|_{\mathbb L^2_{\mathrm{BMO}}(M)}
  +C_K\sqrt{T}\|\Delta\Phi\|_{\mathbb H^2_{\mathrm{BMO}}}
  \\
  &\leq
  C_K(R+\sqrt T)\|(\Delta z^0,\Delta z^1,\Delta u)\|_{\mathcal X},
\end{split}
\end{equation}
for every $\tau_0\in\mathcal T^{0,1}$, where the estimate of the $\Delta u$-term uses $\Delta u=0$ on $\{\lambda=0\}$ and $\lambda\geq\underline\lambda$ on $\{\lambda>0\}$.

Set $(\Delta Y,\Delta Z^0,\Delta Z^1,\Delta U):=(Y,Z^0,Z^1,U)-(\widehat Y,\widehat Z^0,\widehat Z^1,\widehat U)$. Then,
\begin{equation}
  \Delta Y_t = \mathbb E\Bigl[\int_t^T \Bigl(f^{\mathrm{mfg}}(s,z^0_s,z^1_s,u_s) -f^{\mathrm{mfg}}(s,\widehat z^0_s,\widehat z^1_s,\widehat u_s)\Bigr)ds \Bigm|\mathcal F^{0,1}_t \Bigr],~~~t\in[0,T]
\end{equation}
and \eqref{eq:delta-f-bound} gives
\begin{equation}\label{eq:delta-Y-bound}
  \|\Delta Y\|_{\mathbb S^\infty}
  \leq
  C_K(R+\sqrt T)\|(\Delta z^0,\Delta z^1,\Delta u)\|_{\mathcal X},
\end{equation}
and It\^o formula applied to $|\Delta Y|^2$ gives, for every $\tau_0\in\mathcal T^{0,1}$,
\begin{equation}\label{eq:Ito-contraction}
\begin{split}
  &|\Delta Y_{\tau_0}|^2 + \mathbb E\Bigl[\int_{\tau_0}^T\bigl(|\Delta Z^0_s|^2+|\Delta Z^1_s|^2+\lambda_s|\Delta U_s|^2\bigr)ds\Bigm|\mathcal F^{0,1}_{\tau_0}\Bigr]\\
  &=
  2\mathbb E\Bigl[\int_{\tau_0}^T \Delta Y_s\Bigl(f^{\mathrm{mfg}}(s,z^0_s,z^1_s,u_s) -f^{\mathrm{mfg}}(s,\widehat z^0_s,\widehat z^1_s,\widehat u_s)\Bigr)ds\Bigm|\mathcal F^{0,1}_{\tau_0}\Bigr]\\
  &\leq
  2\|\Delta Y\|_{\mathbb S^\infty}\mathbb E\Bigl[\int_{\tau_0}^T \Bigl|f^{\mathrm{mfg}}(s,z^0_s,z^1_s,u_s) -f^{\mathrm{mfg}}(s,\widehat z^0_s,\widehat z^1_s,\widehat u_s)\Bigr|ds \Bigm|\mathcal F^{0,1}_{\tau_0} \Bigr]\\
  &\leq
  C_K(R+\sqrt T)^2\|(\Delta z^0,\Delta z^1,\Delta u)\|^2_{\mathcal X}.
  \end{split}
\end{equation}
Taking the supremum over $\tau_0\in\mathcal T^{0,1}$ yields
\begin{equation}\label{eq:Gamma-contraction-chain}
  \|(\Delta Z^0,\Delta Z^1,\Delta U)\|_{\mathcal X}
  \leq
  C_K(R+\sqrt T)\|(\Delta z^0,\Delta z^1,\Delta u)\|_{\mathcal X}.
\end{equation}
Since $R=(1+\sqrt{\overline{\lambda}})K\sqrt T$, we have
\begin{equation}
  C_K(R+\sqrt T)=C_K\bigl((1+\sqrt{\overline{\lambda}})K+1\bigr)\sqrt T\longrightarrow0~~~(T\to0).
\end{equation}
Shrinking $T_0$ if necessary, we may also assume
\begin{equation}\label{eq:MF-smallness-final}
  C_K(R+\sqrt T)<1~~~\text{for }T\leq T_0.
\end{equation}
Thus, by \eqref{eq:Gamma-contraction-chain} and \eqref{eq:MF-smallness-final}, $\Gamma$ is a contraction on $\mathcal D^{\mathrm{Mk}}_{R,K}$ for $T\leq T_0$.

By Lemma \ref{lem:Dmk-closed}, $\mathcal D^{\mathrm{Mk}}_{R,K}$ is complete under $\|\cdot\|_{\mathcal X}$.
The Banach fixed-point theorem gives a unique fixed point $(Z^{0,*},Z^{1,*},U^*)\in\mathcal D^{\mathrm{Mk}}_{R,K}$.  With
\begin{equation}
  Y^*_t:=\mathbb{E}\Bigl[F+\int_t^T f^{\mathrm{mfg}}(s,Z^{0,*}_s,Z^{1,*}_s,U^*_s)ds\Bigm|\mathcal{F}^{0,1}_t\Bigr],~~~t\in[0,T],
\end{equation}
the quadruple $(Y^*,Z^{0,*},Z^{1,*},U^*)$ solves the mean-field BSDE \eqref{eq:MF-BSDE}--\eqref{eq:MF-driver}. $\square$

\begin{rem}~\\\label{rem:MF-local-uniqueness}
  By the uniqueness of the fixed point of $\Gamma$, the solution obtained in Theorem \ref{thm:MF-exist} is unique within the class $\mathbb S^\infty\times\mathcal D^{\mathrm{Mk}}_{R,K}$.
  This does not claim uniqueness in the larger class $\mathbb S^\infty\times\mathcal X$.
\end{rem}

\section{Market clearing in the large-population limit}\label{sec:clearing}
We now prove that the risk premium characterized by the solution to the mean-field BSDE asymptotically clears the market in the large-population limit.
To formulate this limit, we work on the countable product extension $(\overline{\Omega},\overline{\mathcal{F}},\overline{\mathbb{P}})$ with $\overline{\Omega}:=\prod_{i=0}^{\infty}\Omega^i$.
Here, $(\overline{\mathcal{F}},\overline{\mathbb{P}})$ is the completion of
\begin{equation}
  \Bigl(\bigotimes_{i=0}^{\infty}\mathcal{F}^i,  \bigotimes_{i=0}^{\infty}\mathbb{P}^i\Bigr),
\end{equation}
and $\overline{\mathbb{F}}=(\overline{\mathcal{F}}_t)_{t\in[0,T]}$ is the complete and right-continuous augmentation of $(\bigotimes_{i=0}^{\infty}\mathcal{F}^i_t)_{t\in[0,T]}$.  
In this section, $\mathbb E[\cdot]$ and $\overline{\mathbb E}[\cdot]$ denote expectation under $\overline{\mathbb P}$ and the corresponding $\mathcal F^0$-conditional expectation, respectively.
We use the natural countable-product analogue of Assumption \ref{asm:multiagent}.
The idiosyncratic coordinates $\{(\xi^i,\gamma^i,W^i,F^i)\}_{i\geq1}$ are $\mathcal F^0$-conditionally i.i.d., and the common market data are carried by the zeroth coordinate.

Suppose that the mean-field BSDE \eqref{eq:MF-BSDE}--\eqref{eq:MF-driver}
\begin{equation}\label{eq:selected-MF-BSDE-clearing}
  \mathcal Y_t
  =
  F^1+\int_t^Tf^{\mathrm{mfg}}(s,\mathcal Z^0_s,\mathcal Z^1_s,\mathcal U_s)ds-\int_t^T\mathcal Z^0_sdW^0_s -\int_t^T\mathcal Z^1_sdW^1_s -\int_t^T\mathcal U_sdM_s,~~~ t\in[0,T]
\end{equation}
admits a bounded solution on $(\Omega^{0,1},\mathcal F^{0,1},\mathbb P^{0,1})$. (Theorem \ref{thm:MF-exist} provides such a solution with $(\mathcal Z^0,\mathcal Z^1,\mathcal U)\in\mathcal D^{\mathrm{Mk}}_{R,K}$ for sufficiently short horizons.) 
We fix one such solution, identified with its canonical extension to the countable product, and denote it by $(\mathcal Y,\mathcal Z^0,\mathcal Z^1,\mathcal U) \in \mathbb S^\infty\times\mathbb H^2_{\mathrm{BMO}} \times\mathbb H^2_{\mathrm{BMO}}\times\mathbb L^2_{\mathrm{BMO}}(M)$.
Let $(q^{1,*},\psi^{*})$ be a pair satisfying \eqref{eq:psi-consistency}--\eqref{eq:q-star-eq} for the selected solution.
For the solution provided by Theorem \ref{thm:MF-exist}, the existence of such a pair follows from Lemma \ref{lem:inner-branch}, because $(\mathcal Z^0,\mathcal U)$ is bounded.
Set
\begin{equation}
  \Phi^1_t:=\frac{e^{-q^{1,*}_t}}{\gamma^1},~~~ t\in[0,T],
\end{equation}
so that $\psi^{*}_t=\overline{\mathbb E}[\Phi^1_t]$.  
We define the market risk premium by
\begin{equation}\label{eq:theta-mfg-clearing}
  \theta^{\mathrm{mfg}}_t
  :=
  -\widehat{\gamma}\overline{\mathbb E}[\mathcal Z^{0\|}_t]^{\top} -\widehat{\gamma}\lambda_t\eta_t\psi^{*}_t,~~~ t\in[0,T].
\end{equation}

We consider the financial market with risk premium $\theta^{\mathrm{mfg}}$.
For such a market, the BSDE that characterizes the optimality of agent-$i$ reads:
\begin{equation}\label{eq:individual-BSDE-clearing}
  \begin{split}
  Y^i_t
  &=
  F^i+\int_t^T f^i_{\theta^{\mathrm{mfg}}}(s,Z^{i,0}_s,Z^i_s,U^i_s)ds -\int_t^T Z^{i,0}_sdW^0_s -\int_t^T Z^i_sdW^i_s -\int_t^T U^i_sdM_s,~~~ t\in[0,T],
  \end{split}
\end{equation}
where $f^i_{\theta^{\mathrm{mfg}}}$ is the driver in \eqref{eq:BSDE-Y} with $(\gamma^1,W^1,\theta)$ replaced by $(\gamma^i,W^i,\theta^{\mathrm{mfg}})$.
Note that $(\theta^{\mathrm{mfg}}_t)^{\top}\in L_t$ by construction.
If $\theta^{\mathrm{mfg}}$ is bounded, the specification $\mu:=\sigma\theta^{\mathrm{mfg}}$ defines a market satisfying Assumption \ref{asm:market}.
Theorem \ref{thm:BSDE-wellposed} then gives a unique bounded solution to \eqref{eq:individual-BSDE-clearing}, and Theorem \ref{thm:verify} characterizes the corresponding optimizer by
\begin{equation}\label{eq:p-i-star-clearing}
  p^{i,*}_t
  =
  Z^{i,0\|}_t +\frac{(\theta^{\mathrm{mfg}}_t)^\top}{\gamma^i} +\frac{\lambda_t}{\gamma^i}\eta_t^\top e^{-\gamma^i(p^{i,*}_t\eta_t-U^i_t)},~~~ t\in[0,T].
\end{equation}
The optimal share allocation is then given by
\begin{equation}\label{eq:pi-i-star-clearing}
  \pi^{i,*}_t
  :=
  (\sigma_t\sigma_t^\top)^{-1}\sigma_t(p^{i,*}_t)^\top,~~~t\in[0,T],~~~ i\in\mathbb N.
\end{equation}
This is the second main result of this paper.
\begin{thm}[Asymptotic market clearing]~\\\label{thm:market-clearing}
  Let Assumptions \ref{asm:market}, \ref{asm:agent}, and \ref{asm:multiagent} hold.
  Let
  \begin{equation}
    (\mathcal Y,\mathcal Z^0,\mathcal Z^1,\mathcal U)\in\mathbb S^\infty(\mathbb P^{0,1},\mathbb F^{0,1},\mathbb R)\times\mathbb H^2_{\mathrm{BMO}}(\mathbb P^{0,1},\mathbb F^{0,1},\mathbb R^{1\times d_0})
    \times \mathbb H^2_{\mathrm{BMO}}(\mathbb P^{0,1},\mathbb F^{0,1},\mathbb R^{1\times d}) \times\mathbb L^2_{\mathrm{BMO}}(M;\mathbb P^{0,1},\mathbb F^{0,1},\mathbb R)
  \end{equation}
  be a solution of the mean-field BSDE \eqref{eq:selected-MF-BSDE-clearing}. Let $(q^{1,*},\psi^{*})$ be a pair satisfying \eqref{eq:psi-consistency}--\eqref{eq:q-star-eq} for this solution, and assume that
  \begin{equation}
    \|\mathcal Z^0\|_{\mathbb{L}^{\infty}}<\infty,~~~\|\psi^{*}\|_{\mathbb{L}^{\infty}}<\infty.
  \end{equation}
  We take the canonical predictable representative $\mathcal U=0$ on $\{\lambda=0\}$.
  Then, the risk premium process $\theta^{\mathrm{mfg}}$, defined in \eqref{eq:theta-mfg-clearing}, belongs to $\mathbb{L}^{\infty}(\mathbb{P}^0,\mathbb{F}^0,\mathbb{R}^{d_0})$ and clears the market asymptotically in the sense that
  \begin{equation}\label{eq:clearing-L2}
    \mathbb{E}\Bigl[\int_0^T\Bigl|\frac{1}{N}\sum_{i=1}^N\pi^{i,*}_t\Bigr|^2dt\Bigr] \leq \frac{C}{N}
  \end{equation}
  for a constant $C>0$ independent of $N$.
\end{thm}

\noindent\textbf{\textit{Proof}}\\
First, $\theta^{\mathrm{mfg}}\in\mathbb{L}^{\infty}$. Indeed,
\begin{equation}
  \|\theta^{\mathrm{mfg}}\|_{\mathbb{L}^{\infty}}
  \leq
  \widehat{\gamma}\|\overline{\mathbb{E}}[\mathcal Z^{0\|}]\|_{\mathbb{L}^{\infty}}
  +\widehat{\gamma}\overline{\lambda}\overline h\|\psi^{*}\|_{\mathbb{L}^{\infty}}
  \leq
  \widehat{\gamma}\|\mathcal Z^0\|_{\mathbb{L}^{\infty}}
  +\widehat{\gamma}\overline{\lambda}\overline h\|\psi^{*}\|_{\mathbb{L}^{\infty}}<\infty.
\end{equation}
Next, since $\mathcal Y\in\mathbb S^\infty$, the same argument as in Lemma \ref{lem:Sbounded-implies-BMO} gives $|\mathcal U_t|\leq 2\|\mathcal Y\|_{\mathbb S^\infty}$ on $\{\lambda_t>0\}$.
With the chosen representative, we obtain $\mathcal U\in\mathbb{L}^{\infty}$.
Therefore the argument of $Q_s$ in \eqref{eq:q-star-eq} is bounded, and Lemma \ref{lem:qstar-apriori} yields
\begin{equation}
  \|q^{1,*}\|_{\mathbb{L}^{\infty}}<\infty,
  ~~~
  \|\Phi^1\|_{\mathbb{L}^{\infty}}
  \leq
  \frac{\exp(\|q^{1,*}\|_{\mathbb{L}^{\infty}})}{\underline\gamma}
  <\infty.
\end{equation}

By Theorem \ref{thm:BSDE-wellposed} and Remark \ref{rem:BSDE-Y}, the BSDE \eqref{eq:individual-BSDE-clearing} admits a unique solution $(Y^i,Z^{i,0},Z^i,U^i)\in\mathbb S^\infty\times\mathbb H^2_{\mathrm{BMO}} \times\mathbb H^2_{\mathrm{BMO}}\times\mathbb L^2_{\mathrm{BMO}}(M)$ under the risk premium process $\theta^{\mathrm{mfg}}$.
We fix the canonical predictable representatives $U^i=0$ on $\{\lambda=0\}$ for all $i\geq1$.
For $i=1$, substituting
\begin{equation}
  \theta^{\mathrm{mfg}}_s=-\widehat{\gamma}\bigl(\overline{\mathbb E}[\mathcal Z^{0\|}_s]^{\top}+\lambda_s\eta_s\psi^{*}_s\bigr)
\end{equation}
into the individual driver and using \eqref{eq:psi-consistency}--\eqref{eq:q-star-eq}, we have
\begin{equation}
  f^1_{\theta^{\mathrm{mfg}}}(t,\mathcal Z^0_t,\mathcal Z^1_t,\mathcal U_t)
  =
  f^{\mathrm{mfg}}(t,\mathcal Z^0_t,\mathcal Z^1_t,\mathcal U_t),
  ~~~dt\otimes\overline{\mathbb P}\text{-a.e.}
\end{equation}
Thus $(\mathcal Y,\mathcal Z^0,\mathcal Z^1,\mathcal U)$ also solves the BSDE \eqref{eq:individual-BSDE-clearing}.
Uniqueness therefore shows $(Y^1,Z^{1,0},Z^1,U^1)=(\mathcal Y,\mathcal Z^0,\mathcal Z^1,\mathcal U)$.
In particular, we have $\theta^{\mathrm{mfg}}_t=-\widehat{\gamma}\overline{\mathbb E}[Z^{1,0\|}_t]^{\top} -\widehat{\gamma}\lambda_t\eta_t\psi^{*}_t$.

Set $\Xi^0:=(W^0,N,\theta^{\mathrm{mfg}},\lambda,\sigma,\beta)$.
By Theorem \ref{thm:BSDE-wellposed}, the BSDE \eqref{eq:individual-BSDE-clearing} admits a pathwise unique solution in the bounded class.
Applying the Yamada-Watanabe theorem of Kurtz \cite[Theorem 1.5]{Kurtz2014} to the canonical input-output formulation of this BSDE yields a Borel measurable map $\Upsilon$ such that
\begin{equation}
  (Y^i,Z^{i,0},Z^i,U^i)=\Upsilon\bigl(\Xi^0;\xi^i,\gamma^i,W^i,F^i\bigr),~~~ i\geq1.
\end{equation}
The identity is understood up to indistinguishability for $Y^i$ and up to $dt\otimes\overline{\mathbb P}$-null modifications for $(Z^{i,0},Z^i,U^i)$.
Since $\Xi^0$ is common and $\{(\xi^i,\gamma^i,W^i,F^i)\}_{i\in\mathbb{N}}$ are $\mathcal F^0$-conditionally i.i.d., this representation implies that $(Y^i,Z^{i,0},Z^i,U^i)_{i\in\mathbb{N}}$ are $\mathcal F^0$-conditionally i.i.d. as well.
\footnote{If $(\mathcal Y,\mathcal Z^0,\mathcal Z^1,\mathcal U)\in\mathbb{S}^\infty\times\mathcal D^{\mathrm{Mk}}_{R,K}$, which holds for the solution of Theorem \ref{thm:MF-exist}, then $(Y^i,Z^{i,0},Z^i,U^i)$ is represented through the value functions $w^{(0)}$ and $w^{(1)}$ of Lemma \ref{lem:mk-closure2} evaluated at $(t,X^0_t,X^i_t,\gamma^i)$, from which the $\mathcal F^0$-conditional i.i.d. property is clear.}
Furthermore, since
\begin{equation}
  q^{i,*}_s=Q_s\bigl(\gamma^i(Z^{i,0\|}_s\eta_s-U^i_s)+(\theta^{\mathrm{mfg}}_s)^\top\eta_s\bigr),~~~s\in[0,T],
\end{equation}
where $Q_s$ is defined in \eqref{eq:Q-def}, $(q^{i,*})_{i\in\mathbb{N}}$ are also $\mathcal F^0$-conditionally i.i.d., and thus so is $\Phi^i:=\dfrac{e^{-q^{i,*}}}{\gamma^i}$.

Substituting \eqref{eq:theta-mfg-clearing} into \eqref{eq:p-i-star-clearing} gives
\begin{equation}\label{eq:p-istar-decomp}
  p^{i,*}_t = \underbrace{\left(Z^{i,0\|}_t -\frac{\widehat{\gamma}}{\gamma^i}  \overline{\mathbb{E}}[Z^{1,0\|}_t]\right)}_{=:A^i_t}
  +
  \underbrace{\lambda_t\eta^{\top}_t\left(\Phi^i_t -\frac{\widehat{\gamma}}{\gamma^i}\psi^{*}_t\right)}_{=:B^i_t}.
\end{equation}
The families $(A^i)_{i\in\mathbb N}$ and $(B^i)_{i\in\mathbb N}$ are $\mathcal F^0$-conditionally i.i.d.
They satisfy
\begin{equation}
  |A^1_t|\leq C\|\mathcal Z^0\|_{\mathbb{L}^{\infty}},~~~|B^1_t|\leq C(\|\psi^{*}\|_{\mathbb{L}^{\infty}} + |\Phi^1_t|).
\end{equation}
Moreover, their conditional means vanish, as
\begin{equation}
  \overline{\mathbb E}[A^1_t]
  =
  \overline{\mathbb E}[Z^{1,0\|}_t]
  -\widehat\gamma\overline{\mathbb E}\Bigl[\frac{1}{\gamma^1}\Bigr]\overline{\mathbb E}[Z^{1,0\|}_t]
  =0,~~~
  \overline{\mathbb E}[B^1_t]
  =
  \lambda_t\eta_t^\top
  \left(
    \overline{\mathbb E}[\Phi^1_t]
    -\widehat\gamma\overline{\mathbb E}\Bigl[\frac{1}{\gamma^1}\Bigr]\psi^*_t
  \right)
  =0.
\end{equation}
For $i=1$ this gives $\overline{\mathbb{E}}[p^{1,*}_t]=0$, $dt\otimes\mathbb{P}^0$-a.e., i.e., $\theta^{\mathrm{mfg}}$ satisfies the mean-field clearing condition \eqref{eq:mc-limit}. Then,
\begin{equation}
  \begin{split}
      \mathbb{E}\Bigl[\int_0^T\Bigl|\frac{1}{N}\sum_{i=1}^N\pi^{i,*}_t\Bigr|^2dt\Bigr]
      &\leq
      C\mathbb{E}\Bigl[\int_0^T\Bigl|\frac{1}{N}\sum_{i=1}^Np^{i,*}_t\Bigr|^2dt\Bigr]\\
      &\leq
      C\mathbb E\Bigl[\int_0^T \Bigl|\frac{1}{N}\sum_{i=1}^N A^i_t\Bigr|^2dt\Bigr] + C\mathbb E\Bigl[\int_0^T \Bigl|\frac{1}{N}\sum_{i=1}^N B^i_t\Bigr|^2dt\Bigr]\\
      &\leq
      \frac{C}{N}\mathbb E\Bigl[\int_0^T |A^1_t|^2dt\Bigr] + \frac{C}{N}\mathbb E\Bigl[\int_0^T |B^1_t|^2dt\Bigr]
      \leq
      \frac{C}{N}.
  \end{split}
\end{equation}
This proves the desired result. $\square$

\begin{rem}~\\
    Using the conditional i.i.d. and centered structure established above, we can also show
    \begin{equation}
      \lim_{N\to\infty}\frac{1}{N}\sum_{i=1}^N\pi^{i,*}_t=0,~~~dt\otimes\overline{\mathbb P}\text{-a.e.}
    \end{equation}
    This follows by considering the dyadic subsequence $(2^m)_{m\geq0}$ and controlling the fluctuations for $2^m\leq N<2^{m+1}$ by Doob's maximal inequality.
\end{rem}

\section{Conclusion and discussion}\label{sec:conclusion}

This paper developed a mean-field equilibrium model of price formation in which agents with exponential utilities are exposed to a single-default event.
The equilibrium risk premium decomposes into a Brownian hedging component and a default-risk component that aggregates default intensity, jump size, and heterogeneity in risk aversion and terminal liabilities.
Section \ref{sec:individual-bsde} characterized each agent's optimal strategy by a quadratic-growth BSDE with jumps and established its well-posedness together with the verification theorem.
Section \ref{sec:mfg} introduced a Markovian factor model, derived \textit{a priori} gradient estimates for the associated PDE system, and proved existence of a mean-field equilibrium for sufficiently short horizons by a contraction argument.
This formulation may provide a basis for numerical analysis of the model in future work.
Section \ref{sec:clearing} established asymptotic market clearing in the large-population limit.

Several directions remain open.
The existence result (Theorem \ref{thm:MF-exist}) is local in the horizon $T$. Extending it to arbitrary maturities, possibly via stronger \textit{a priori} estimates or a continuation argument, is an important open question.
The single-default restriction is essential for the uniqueness argument of Section \ref{sec:individual-bsde}. Accommodating multiple defaults would require new techniques for the underlying quadratic BSDE.
Another natural direction is a partially observable market, where agents do not directly observe the default intensity and equilibrium price formation must be combined with the filtering of default risk.
Finally, the no-default model of \cite{FujiiSekine2025} treats the risk premium as an $\mathbb{H}^2_{\mathrm{BMO}}$ process rather than assuming uniform boundedness.
Extending this level of generality to the defaultable setting would require a new stability and solvability theory for the mean-field BSDE beyond the bounded risk-premium setting.

\section*{Acknowledgments}
The author is grateful to Masaaki Fujii at the Graduate School of Economics, The University of Tokyo for useful discussions related to the earlier joint work \cite{FujiiSekine2025}.

\section*{Declarations}
\noindent
\textbf{Funding}~~ The author was supported by the Grant-in-Aid for JSPS Fellows (Grant Number JP23KJ0648).\\
\textbf{Conflict of interest}~~ The author has no competing interests to declare.\\
\textbf{Data availability}~~ No datasets were generated or analyzed for this study.


\begin{thebibliography}{100}

\bibitem{AliprantisBorder2006}
C.~D.~Aliprantis and K.~C.~Border,
\textit{Infinite Dimensional Analysis: A Hitchhiker's Guide},
3rd ed., Springer, Berlin, Germany, 2006.

\bibitem{Brezis}
H.~Brezis,
\textit{Functional Analysis, Sobolev Spaces and Partial Differential Equations},
Springer, New York, USA, 2011.

\bibitem{carmonaProbabilisticAnalysisMeanField2013}
R.~Carmona and F.~Delarue,
Probabilistic analysis of mean-field games,
\textit{SIAM Journal on Control and Optimization}, 51(4), 2705--2734, 2013.

\bibitem{carmonaForwardBackwardStochastic2015}
R.~Carmona and F.~Delarue,
Forward-backward stochastic differential equations and controlled McKean-Vlasov dynamics,
\textit{Annals of Probability}, 43(5), 2647--2700, 2015.

\bibitem{CarmonaDelarue2018I}
R.~Carmona and F.~Delarue,
\textit{Probabilistic Theory of Mean Field Games with Applications I: Mean Field FBSDEs, Control, and Games},
Probability Theory and Stochastic Modelling, vol.~83, Springer, Cham, Switzerland, 2018.

\bibitem{CarmonaDelarue2018II}
R.~Carmona and F.~Delarue,
\textit{Probabilistic Theory of Mean Field Games with Applications II: Mean Field Games with Common Noise and Master Equations},
Probability Theory and Stochastic Modelling, vol.~84, Springer, Cham, Switzerland, 2018.

\bibitem{FeronTankovTinsi2022}
O.~F\'eron, P.~Tankov and L.~Tinsi,
Price formation and optimal trading in intraday electricity markets,
\textit{Mathematics and Financial Economics}, 16, 205--237, 2022.

\bibitem{Friedman1964}
A.~Friedman,
\textit{Partial Differential Equations of Parabolic Type},
Prentice-Hall, Englewood Cliffs, NJ, USA, 1964.

\bibitem{Fujii2023}
M.~Fujii,
Equilibrium pricing of securities in the co-presence of cooperative and non-cooperative populations,
\textit{ESAIM: Control, Optimisation and Calculus of Variations}, 29, 56, 2023.

\bibitem{Fujii2025Tree}
M.~Fujii,
Mean-field price formation on trees with multi-population and non-rational agents,
preprint, arXiv:2510.11261, 2025.

\bibitem{Fujii2025TreeRelative}
M.~Fujii,
Mean-field price formation on trees with a network of relative performance concerns,
preprint, arXiv:2512.21621, 2025.

\bibitem{FujiiSekine2025}
M.~Fujii and M.~Sekine,
Mean-field equilibrium price formation with exponential utility,
\textit{Stochastics and Dynamics}, 24(8), 2024.

\bibitem{FujiiSekineHabit2026}
M.~Fujii and M.~Sekine,
Mean field equilibrium asset pricing model with habit formation,
\textit{Asia-Pacific Financial Markets}, 33, 263--314, 2026.

\bibitem{FujiiTakahashi2022SICON}
M.~Fujii and A.~Takahashi,
A mean field game approach to equilibrium pricing with market clearing condition,
\textit{SIAM Journal on Control and Optimization}, 60(1), 259--279, 2022.

\bibitem{FujiiTakahashi2022ESAIM}
M.~Fujii and A.~Takahashi,
Equilibrium price formation with a major player and its mean field limit,
\textit{ESAIM: Control, Optimisation and Calculus of Variations}, 28, 21, 2022.

\bibitem{FujiiTakahashi2022SIFIN}
M.~Fujii and A.~Takahashi,
Strong convergence to the mean field limit of a finite agent equilibrium,
\textit{SIAM Journal on Financial Mathematics}, 13(2), 459--490, 2022.

\bibitem{GomesGutierrezRibeiro2023}
D.~A.~Gomes, J.~Gutierrez and R.~Ribeiro,
A random-supply mean field game price model,
\textit{SIAM Journal on Financial Mathematics}, 14(1), 188--222, 2023.

\bibitem{GomesSaude2021}
D.~A.~Gomes and J.~Sa\'ude,
A mean-field game approach to price formation,
\textit{Dynamic Games and Applications}, 11, 29--53, 2021.

\bibitem{HIM2005}
Y.~Hu, P.~Imkeller and M.~M\"{u}ller,
Utility maximization in incomplete markets,
\textit{Annals of Applied Probability}, 15(3), 1691--1712, 2005.

\bibitem{huangLargePopulationStochastic2006}
M.~Huang, R.~P.~Malham\'e and P.~E.~Caines,
Large population stochastic dynamic games: closed-loop McKean-Vlasov systems and the Nash certainty equivalence principle,
\textit{Communications in Information and Systems}, 6(3), 221--252, 2006.

\bibitem{jeanblanc_mathematical_2009}
M.~Jeanblanc, M.~Yor and M.~Chesney,
\textit{Mathematical Methods for Financial Markets},
Springer Finance, Springer, London, UK, 2009.

\bibitem{karatzas_methods_1998}
I.~Karatzas and S.~E.~Shreve,
\textit{Methods of Mathematical Finance},
Applications of Mathematics, vol.~39, Springer, New York, USA, 1998.

\bibitem{Kazamaki1994}
N.~Kazamaki,
\textit{Continuous Exponential Martingales and BMO},
Lecture Notes in Mathematics, vol.~1579, Springer, Berlin, Germany, 1994.

\bibitem{Kobylanski2000}
M.~Kobylanski,
Backward stochastic differential equations and partial differential equations with quadratic growth,
\textit{Annals of Probability}, 28(2), 558--602, 2000.

\bibitem{Krylov1980}
N.~V.~Krylov,
\textit{Controlled Diffusion Processes},
Applications of Mathematics, vol.~14, Springer, New York, USA, 1980.

\bibitem{Krylov2008}
N.~V.~Krylov,
\textit{Lectures on Elliptic and Parabolic Equations in Sobolev Spaces},
Graduate Studies in Mathematics, vol.~96, American Mathematical Society, Providence, RI, USA, 2008.

\bibitem{Kurtz2014}
T.~G.~Kurtz,
Weak and strong solutions of general stochastic models,
\textit{Electronic Communications in Probability}, 19, 1--16, 2014.

\bibitem{lasryMeanFieldGames2007}
J.-M.~Lasry and P.-L.~Lions,
Mean field games,
\textit{Japanese Journal of Mathematics}, 2(1), 229--260, 2007.

\bibitem{LimQuenez2011}
T.~Lim and M.-C.~Quenez,
Exponential utility maximization in an incomplete market with defaults,
\textit{Electronic Journal of Probability}, 16, 1434--1464, 2011.

\bibitem{Morlais2009}
M.-A.~Morlais,
Utility maximization in a jump market model,
\textit{Stochastics}, 81(1), 1--27, 2009.

\bibitem{Morlais2010}
M.-A.~Morlais,
A new existence result for quadratic BSDEs with jumps with application to the utility maximization problem,
\textit{Stochastic Processes and their Applications}, 120(10), 1966--1995, 2010.

\bibitem{PardouxRascanu2014}
E.~Pardoux and A.~R\u{a}\c{s}canu,
\textit{Stochastic Differential Equations, Backward SDEs, Partial Differential Equations},
Stochastic Modelling and Applied Probability, vol.~69, Springer, Cham, Switzerland, 2014.

\bibitem{Protter2005}
P.~E.~Protter,
\textit{Stochastic Integration and Differential Equations},
2nd ed., Stochastic Modelling and Applied Probability, vol.~21, Springer, Berlin, Germany, 2005.

\bibitem{Sekine2025EQG}
M.~Sekine,
Mean field equilibrium asset pricing model under partial observation: an exponential quadratic Gaussian approach,
\textit{Japan Journal of Industrial and Applied Mathematics}, 42, 853--878, 2025.

\bibitem{SekineThesis2025}
M.~Sekine,
\textit{Mean Field Equilibrium Asset Pricing Models with Exponential Utility},
Ph.D.~thesis, The University of Tokyo, 2025, available at arXiv:2603.22058.

\bibitem{ShrivatsFirooziJaimungal2022}
A.~V.~Shrivats, D.~Firoozi and S.~Jaimungal,
A mean-field game approach to equilibrium pricing in solar renewable energy certificate markets,
\textit{Mathematical Finance}, 32(3), 779--824, 2022.

\bibitem{ZhangBSDE}
J.~Zhang,
\textit{Backward Stochastic Differential Equations: From Linear to Fully Nonlinear Theory},
Probability Theory and Stochastic Modelling, vol.~86, Springer, New York, USA, 2017.

\end{thebibliography}
\end{document}